\documentclass[UKenglish,cleveref,numberwithinsect,thm-restate,letter]{no-lipics-v2019}

\RequirePackage{doi}
\usepackage{hyperref}
\usepackage{blkarray}
\usepackage{bm}
\usepackage{rotating}
\usepackage{afterpage}
\usepackage{booktabs}
\allowdisplaybreaks

\theoremstyle{definition}
\newtheorem{defn}[theorem]{Definition}

\newcommand{\Kgroup}{\ensuremath{\mathcal{K}}}

\newcommand{\NP}{\ensuremath{\mathsf{NP}}}
\newcommand{\ccP}{\ensuremath{\mathsf{P}}}
\newcommand{\ccFPT}{\ensuremath{\mathsf{FPT}}}
\newcommand{\W}[1]{\ensuremath{\mathsf{W[#1]}}}

\newcommand{\homs}[2]{\mbox{\ensuremath{\mathsf{Hom}(#1 \to #2)}}}

\newcommand{\embs}[2]{\mbox{\ensuremath{\mathsf{Emb}(#1 \to #2)}}}

\newcommand{\subs}[2]{\mbox{\ensuremath{\mathsf{Sub}(#1 \to #2)}}}

\newcommand{\indsubs}[2]{\mbox{\ensuremath{\mathsf{IndSub}(#1 \to #2)}}}
\newcommand{\edgesubs}[2]{\mbox{\ensuremath{\mathsf{EdgeSub}(#1 \to #2)}}}
\newcommand{\auts}[1]{\ensuremath{\mathsf{Aut}(#1)}}

\newcommand{\cphoms}[2]{\ensuremath{\mathsf{cp}\text{-}\mathsf{Hom}}(#1 \to #2)}
\newcommand{\necphoms}[3]{\ensuremath{\mathsf{Hom}_{\mathsf{cp}}}(#1 \to_{#2} #3)}
\newcommand{\ecphoms}[2]{\ensuremath{\mathsf{Hom}_{\mathsf{cp}}}(#1 \to #2)}

\newcommand{\necpembs}[3]{\ensuremath{\mathsf{Emb}_{\mathsf{cp}}}(#1 \to_{#2} #3)}

\newcommand{\ncoledgesubs}[3]{\mbox{\ensuremath{\mathsf{ColEdgeSub}(#1 \to_{#2} #3)}}}

\newcommand{\clique}{\text{\sc{Clique}}}
\newcommand{\homsprob}{\text{\sc{Hom}}}
\newcommand{\cphomsprob}{\text{\sc{cp-Hom}}}

\newcommand{\edgesubsprob}{\text{\sc{EdgeSub}}}

\newcommand{\coledgesubsprob}{\text{\sc{ColEdgeSub}}}
\newcommand{\subsprob}{\text{\sc{Sub}}}
\newcommand{\embsprob}{\text{\sc{Emb}}}

\newcommand{\torus}{\ensuremath{\pmb{\circledcirc}}}

\newcommand{\tu}{\ensuremath{\vartriangle}}
\newcommand{\td}{\ensuremath{\triangledown}}
\newcommand{\tl}{\ensuremath{\vartriangleleft}}
\newcommand{\tr}{\ensuremath{\vartriangleright}}

\newcommand{\ztwol}{\ensuremath{\mathbb{Z}^2_\ell}}

\newcommand{\coltkg}{\ensuremath{\mathsf{col}\text{-}\widehat{T}^k_G}}
\newcommand{\coltlg}{\ensuremath{\mathsf{col}\text{-}\widehat{T}^{2\ell^2}_{\star}}}

\newcommand{\cfunction}{coefficient~function}

\def\pr#1{\ensuremath{\mathsf{Pr}\!\bm{\left[}\,#1\,\bm{\right]}}}

\def\fracture#1#2{\ensuremath{#1\raisebox{.2ex}{\rotatebox[origin=c]{-15}{$\sharp$}}#2}}

\newcommand{\fptred}{\ensuremath{\leq^{\mathrm{fpt}}_{\mathrm{T}}}}

\usetikzlibrary{conference}
\usetikzlibrary{cd}

\title{Detecting and Counting Small Subgraphs, and Evaluating a Parameterized Tutte
Polynomial: Lower Bounds via Toroidal Grids and Cayley Graph Expanders}
\titlerunning{Detecting and Counting Small Subgraphs, and a Parameterized Tutte Polynomial}

\author{Marc Roth}{Merton College, University of Oxford, United Kingdom}{marc.roth@merton.ox.ac.uk}{https://orcid.org/0000-0003-3159-9418}{}

\author{Johannes Schmitt}{Mathematical Institute, University of Bonn, Germany}{schmitt@math.uni-bonn.de}{https://orcid.org/0000-0001-5774-3508}{}

\author{Philip Wellnitz}{Max Planck Institute for Informatics, Saarland Informatics Campus
    (SIC), Saarbrücken,
Germany}{wellnitz@mpi-inf.mpg.de}{https://orcid.org/0000-0002-6482-8478}{}

\authorrunning{M. Roth, J. Schmitt, and P. Wellnitz}

\Copyright{Marc Roth, Johannes Schmitt, and Philip Wellnitz}

\ccsdesc[500]{Theory of computation~Parameterized complexity and exact algorithms}
\ccsdesc[300]{Theory of computation~Problems, reductions and completeness}

\keywords{Counting complexity, parameterized complexity, Tutte polynomial, graph homomorphisms, Tutte polynomial}

\acknowledgements{
    We thank Alina Vdovina and Norbert Peyerimhoff for explaining
    their construction of $2$-group Cayley graph expanders~\cite{PeyerimhoffV11}.
    We thank Johannes Lengler and Holger Dell for helpful discussions on early
    drafts of this work.\\
    The second author was supported by the SNF early postdoc mobility grant 184245 and
    thanks the Max Planck Institute for Mathematics in Bonn for its hospitality.}

\nolinenumbers 

\begin{document}

\setlength\marginparwidth{67pt}
\setlength\marginparsep{2pt}

\maketitle

\begin{abstract}
    Given a graph property $\Phi$, we consider the problem $\edgesubsprob(\Phi)$,
    where the input is a pair of a graph $G$ and a positive integer $k$,
    and the task is to decide whether $G$ contains a $k$-edge subgraph that satisfies~$\Phi$.
    Specifically, we study the parameterized complexity
    of $\edgesubsprob(\Phi)$ and of its counting problem $\#\edgesubsprob(\Phi)$ with
    respect to both approximate and exact counting.
    We obtain a complete picture for minor-closed properties $\Phi$:
    the decision problem $\edgesubsprob(\Phi)$ always admits an FPT
    (``fixed-parameter tractable'') algorithm  and the counting
    problem $\#\edgesubsprob(\Phi)$ always admits an FPTRAS (``fixed-parameter tractable
    randomized approximation scheme''). For exact counting, we present an exhaustive and
    explicit criterion on the property~$\Phi$ which, if satisfied, yields fixed-parameter
    tractability and otherwise $\#\W1$-hardness.
    Additionally, most of our hardness results come with an almost tight conditional lower
    bound under the so-called Exponential Time Hypothesis,
    ruling out algorithms for $\#\edgesubsprob(\Phi)$ that run in time $f(k)\cdot
    |G|^{o(k/\log k)}$ for any computable function $f$.

    As a main technical result, we gain a complete understanding
    of the coefficients of toroidal grids and selected Cayley
    graph expanders in the homomorphism basis of $\#\edgesubsprob(\Phi)$.
    This allows us to establish hardness of exact counting using the
    Complexity Monotonicity framework due to Curticapean, Dell and Marx (STOC'17).
    This approach does not only apply to $\#\edgesubsprob(\Phi)$ but also to the more
    general problem of computing weighted linear combinations of subgraph counts.
    As a special case of such a linear combination, we introduce a
    parameterized variant of the Tutte Polynomial $T^k_G$ of a graph $G$, to which many
    known combinatorial interpretations of values of the (classical) Tutte Polynomial can
    be extended.
    As an example, $T^k_G(2,1)$ corresponds to the number of $k$-forests in the graph $G$.
    Our techniques allow us to completely understand the
    parameterized complexity of computing the evaluation of $T^k_G$
    at every pair of rational coordinates $(x,y)$.
    In particular, our results give a new proof for the
    $\#\W1$-hardness of the problem of counting $k$-forests in a graph.
\end{abstract}

\section{Introduction}

Be it searching for cliques in social networks or understanding protein-protein
interaction networks, many interesting real-life problems boil down to finding (or
counting) small patterns in large graphs.
Hence, to no surprise,
finding (and counting) small patterns in large graphs are among the most well-studied
computational problems in the fields of database
theory~\cite{ChandraM77,GroheSS01,DurandM15,ChenM16,DellRW19icalp}, molecular biology and
bioinformatics~\cite{GrochovK07,Nogaetat08,Rahatetal09,Schilleretal15}, and network
science~\cite{SchreiberS05,Miloetal02,Miloetal04}.
In fact, already in the 1970s, the relevance of \emph{finding} patterns became apparent
in the context of finding cliques, finding Hamiltonian paths, or finding specific
subgraphs in general~\cite{CorneilG70,Cook71,Ullmann76,ChandraM77}.
However, with the advent of motif counting
for the frequency analysis of small structures in complex networks~\cite{Miloetal02,Miloetal04},
it became evident that detecting the existence of a pattern graph
is not enough; we also need to \emph{count} all of the occurrences of the pattern.

In this work, our patterns are (not necessarily induced) edge subgraphs that satisfy a
certain graph property: for instance, given a graph, we want to count all occurrences of edge
subgraphs that are are planar or connected.

From a classical point of view, often the problem of \emph{finding} patterns is already \NP-hard:
prime examples include the aforementioned problems of finding (maximum) cliques or
Hamiltonian paths.
However, for the task of network motif counting, the patterns are (almost) always much smaller
than the network itself (see \cite{Miloetal02,Miloetal04,Nogaetat08}).
This motivates a \emph{parameterized} view:
can we obtain fast algorithms to compute the number of occurrences of ``small'' patterns?
If we cannot, can we at least obtain fast (randomized) algorithms to compute an
\emph{estimate} of this number? And if we cannot even do this,
can we at least obtain fast algorithms to \emph{detect} an occurrence?
In this work, we completely answer all of the above questions for patterns that are specified by
\emph{minor-closed} graph properties (such as planarity) or selected other graph properties
(such as connectivity).

As it turns out, the techniques we develop for answering the above questions are quite
powerful: they easily generalize to a parameterized version of the Tutte polynomial.
Specifically, our techniques allow us to completely understand at which rational
points we can  evaluate said parameterized Tutte polynomial in reasonable time,
and at which rational points this is not feasible.
This dichotomy turns out to be similar, but not equal, to the complexity landscape of the
classical Tutte polynomial due to Jaeger et al.~\cite{JaegerVW90}.

\paragraph*{Parameterized Counting and Hardness}

By now, \emph{counting complexity theory} is a well established subfield of theoretical
computer science.
Already in the~1970s, Valiant started a formal study of counting problems when
investigating the complexity of the permanent~\cite{Valiant79,Valiant79b}: counting the number of perfect
matchings in a graph is {\sf \#P}-complete, and hence harder than any problem in the
polynomial-time hierarchy {\sf PH} by Toda's Theorem~\cite{Toda91}.
In contrast, \emph{detecting} a perfect matching in a graph is much easier
and can be done in polynomial time~\cite{Edmonds65}.
Hence, counting problems can be much harder than their decision problem counterparts.

As an attempt to overcome the hardness of counting problems in general, the focus shifted
to a multivariate or \emph{parameterized} view on these problems.
Consider for example the following problem:
given a query $\varphi$ of size $k$ and a database $B$ of size $n$,
we want to count the number of answers to $\varphi$ in $B$.
If we make the very reasonable assumption that $k$ is much smaller than~$n$,
then we may consider an algorithm running in time $O(2^k \cdot n)$ as
\emph{tractable}. Note that in particular, such an algorithm may even outperform an
algorithm running in time $O(n^2)$. Also consider~\cite{Grohe02} for a more detailed
and formal discussion.

\noindent Formally, given a problem $P$ and a \emph{parameterization} $\kappa$ that maps each
instance $I$ of $P$ to a parameter~$\kappa(I)$, we say that $P$
is \emph{fixed-parameter tractable} (FPT) with respect to $\kappa$, if there is an algorithm that solves each
instance $I$ of size $n$ in time $f(\kappa(I)) \cdot n^{O(1)}$, for some computable
function $f$.
This notion was introduced by Downey and Fellows in the early
1990s~\cite{DowneyF95,DowneyF95b} and has itself spawned a rich body of
literature (see \cite{FlumG06,DowneyF13,CyganFKLMPPS15}).
In the context of the problems of detecting and counting small patterns in large networks,
we {parameterize} by the size of the pattern:
given a pattern of size $k$ and a network of size~$n$, we aim for
algorithms that run in time $f(k)\cdot n^{O(1)}$, for some computable function $f$.
However, for some patterns, even this goal is too ambitious:
it is widely believed that even finding a clique of size~$k$ is not fixed-parameter
tractable; in particular, an FPT algorithm for finding a clique of size $k$ would
also imply a breakthrough result for the Satisfiability Problem and
thereby refute the widely believed Exponential Time Hypothesis~\cite{Chenetal06,Chenetal09}.
If a problem $P$ is at least as hard as finding a clique (or counting all cliques)
of size $k$, we say that $P$ is $\W1$-hard (or $\#\W1$-hard, respectively).

For such a $(\#)\W1$-hard problem, the hope is to (significantly) improve upon the naive
brute-force algorithm, which runs in time $n^{O(k)}$ for the problems considered in this work.
However, in view of the aforementioned reduction from the Satisfiability Problem to the
problem of finding cliques of size $k$~\cite{Chenetal05,Chenetal06}, we can see that for
finding cliques this, too, would require a breakthrough for the Satisfiability Problem,
which, again, is believed to be unlikely~\cite{ImpagliazzoP01}.
In our paper, via suitable reductions from the problem of finding cliques, we establish
that exact algorithms significantly faster than the brute-force algorithms are unlikely for the problems
we study.

\paragraph*{Parameterized Detection and Counting of Edge Subgraphs}

\emph{Vertex-induced} subgraphs as patterns are notoriously hard to
detect or to count. The long line of research on this problem~\cite{KhotR02,ChenTW08,JerrumM15,JerrumM15density,Meeks16,JerrumM17,CurticapeanDM17,RothS18,DorflerRSW19,RothSW20}
showed that this holds even if the patterns are significantly smaller than the
host graphs, as witnessed by $\W{1}$ and $\#\W{1}$-hardness results and almost tight
conditional lower bounds. In case of exact counting, it is in fact an open question
whether there are non-trivial instances of induced subgraph counting that admit efficient
algorithms; recent work~\cite{RothSW20} supports the conjecture that no such instances exist.

In search for fast algorithms, in this work, we hence consider a related, but different
version of network-motif counting:
for a computable graph property~$\Phi$, in the problem $\#\edgesubsprob(\Phi)$ we are
given a graph~$G$ and a positive integer~$k$, and the task is to compute the number of
(not necessarily induced) edge subgraphs\footnote{Recall that an edge subgraph $G'$ of a
graph $G$ may have fewer edges than the subgraph of $G$ that is induced by the vertices
of~$G'$.} with $k$ edges in $G$ that satisfy $\Phi$.
Similarly, we write $\edgesubsprob(\Phi)$ for the corresponding decision problem.
Then, in contrast to the case of counting vertex-induced subgraphs, for
$(\#)\edgesubsprob(\Phi)$, we identify non-trivial properties $\Phi$
for which $(\#)\edgesubsprob(\Phi)$ is fixed-parameter tractable; we discuss this in more
detail later. First, however, let us take a detour to elaborate more on what is known
already for $(\#)\edgesubsprob(\Phi)$.

If the property $\Phi$ is satisfied by at most a single graph for each value of the
parameter~$k$, the decision problem $\edgesubsprob(\Phi)$ becomes the
subgraph isomorphism problem.
Hence, naturally there is a vast body of known techniques and results for special
properties $\Phi$:
for FPT algorithms, think of the Colour-Coding technique by Alon, Yuster and Zwick~\cite{AlonYZ95},
the  ``Divide and Colour''-technique~\cite{Chenetal09}, narrow sieving~\cite{Bjorklundetal17},
representative sets~\cite{FominLS14}, or ``extensor-coding''~\cite{BrandDH18} to name but
a few.
For hardness results, apart from the aforementioned example of detecting a clique,
Lin quite recently established that detecting a $k$-biclique is also
$\W1$-hard~\cite{Lin18}.
However, a complete understanding of the parameterized decision version of the
subgraph isomorphism is one of the major open problems of parameterized complexity
theory~\cite[Chapter~33.1]{DowneyF13}, that is still to be solved.


In the setting of parameterized \emph{counting}, the situation is much better understood:
Flum and Grohe~\cite{FlumG04} proved $\#\edgesubsprob(\Phi)$ to be $\#\W{1}$-hard when~$\Phi$
is the property of being a cycle, or the property of being a path.
Curticapean~\cite{Curticapean13} established the same result for the property of being a
matching. In~\cite{CurticapeanM14}, Curticapean and Marx established a complete
classification in case $\Phi$ does not hold on two different graphs with the same number
of edges, which is essentially the parameterized subgraph counting problem. In particular,
they identified a bound on the matching number as the tractability criterion.
In a later work, together with Dell~\cite{CurticapeanDM17}, they presented what is now
called the framework of Complexity Monotonicity, which can be considered to be one of the most
powerful tools in the field of parameterized counting problems.
Note that this does not classify the decision version, as $\#\W1$-hardness for a counting
problem does not imply $\W1$-hardness for the corresponding decision problem.

In contrast to the parameterized subgraph detection/counting problems, the problem
$(\#)\edgesubsprob(\Phi)$ allows to search for more general patterns.
For example, while the (parameterized) complexity of counting all subgraphs of a graph $G$
isomorphic to a \emph{fixed} connected graph $H$ with $k$ edges is fully
understood~\cite{CurticapeanM14}, the case of counting all connected $k$-edge subgraphs of
a graph $G$ remained open so far. As one of our main results, we completely understand the
problem $\#\edgesubsprob(\Phi)$ for the property $\Phi = \text{connectivity}$. In what follows, we present our results, followed by an exposition of the most important techniques.

\subsection*{Main Results}
In a first part, we present our results on $(\#)\edgesubsprob(\Phi)$; we continue with a
definition and our results for a parameterized Tutte polynomial in a second part.

Our main results on $(\#)\edgesubsprob(\Phi)$ can be categorized in roughly three
categories: (1) exact algorithms and hardness results for the counting problem;
(2) approximation algorithms for the counting problem; and (3) algorithms for the decision
problem. For minor-closed properties $\Phi$, we obtain exhaustive
results for all three categories, for other (classes of) properties that we study, we obtain partial criteria. For an overview over our results
on $\#\edgesubsprob(\Phi)$, also consider~\cref{table:results}; we go into more detail in
the following.

\begin{table}[p]
    \begin{tabularx}{\textwidth}{Xccc}
        \toprule
        Property $\Phi$ & Exact Counting & Apx. Counting & Decision\\\toprule
        \begin{tabular}{l}\!\!\!\!Minor-closed$^\dag$\\{\small (e.g. $\Phi =$ planarity)}
        \end{tabular}&  \begin{tabular}{c}$\#\W{1}$-hard\\not in
            $f(k)\cdot |G|^{o(k/\log k)}$\end{tabular} & FPTRAS & FPT\\
        & {\small (\cref{thm:minclose-hard})} &
        {\small (\cref{thm:minclose-hard})}
        & {\small (\cref{thm:minclose-hard})}\\[.8ex]
        $\Phi =$ connectivity &  \begin{tabular}{c}$\#\W{1}$-hard\\not in
            $f(k)\cdot |G|^{o(k/\log k)}$\end{tabular} & FPTRAS & FPT\\
        & {\small (\cref{cor:intro_further})} &
        {\small (follows from~\cite{DellLM20})}
        & {\small (easy)}\\[.8ex]
        $\Phi =$ Hamiltonicity &  \begin{tabular}{c}$\#\W{1}$-hard\\not in
            $f(k)\cdot |G|^{o(k/\log k)}$\end{tabular} &
        unknown & unknown\\
                & {\small (\cref{cor:intro_further})}\\[.8ex]
        $\Phi =$ Eulerianity &  \begin{tabular}{c}$\#\W{1}$-hard\\not in
            $f(k)\cdot |G|^{o(k/\log k)}$\end{tabular} &
        unknown & unknown\\
                & {\small (\cref{cor:intro_further})}\\[.8ex]
        $\Phi =$ claw-freeness &  \begin{tabular}{c}$\#\W{1}$-hard\\not in
            $f(k)\cdot |G|^{o(k/\log k)}$\end{tabular} &
        unknown & unknown\\
        & {\small (\cref{cor:intro_further})}\\
        \midrule
        Bounded matching & FPT & FPTRAS & FPT\\
        number & {\small (\cref{thm:smallmatchFPT})} &
        {\small (by exact counting)}
        & {\small (by exact counting)}\\[.8ex]
        Bounded treewidth & mixed$^\ddagger$ & FPTRAS & FPT\\
        & {~} &
        {\small (\cref{thm:approx_main_intro})}
        & {\small (follows from~\cite{PlehnV90})}\\[.8ex]
        Matching crit. {\bf and} & mixed$^\ast$ & FPTRAS & FPT\\
                                 star crit. &&
        {\small (\cref{thm:approx_main_intro})}
                                            & {\small
                                            (\cref{thm:dec_classification_intro})}\\[.8ex]
        Matching crit. {\bf or} & mixed$^\ddagger$ & mixed$^\S$ & FPT\\
                                star crit. &&
        {~}
        & {\small (\cref{thm:dec_classification_intro})}\\
        \midrule
        $\Phi = \Psi$ from Definition~\ref{Def:propertyPsi} & $\#\W{1}$-hard & no FPTRAS & FPT\\
        ~ & {\small (Theorem~\ref{lem:coef_special_case_intro})} &
        {\small (Theorem~\ref{thm:decapproxsep})}
                                    & {\small (Main Theorem~\ref{thm:dec_classification_intro})}\\[.8ex]
        $\Phi = \clique$ & $\#\W{1}$-hard & no FPTRAS & $\W{1}$-hard\\
        & {\small (\cite{FlumG04})} &
        {\small (implicitly by \cite{DowneyF95b})}
        & {\small (\cite{DowneyF95b})}\\\bottomrule\\
    \end{tabularx}
    \caption{\label{table:results}An overview of the complexity of $(\#)\edgesubsprob(\Phi)$ for
        different classes and examples of properties~$\Phi$, with respect to exact
        counting, approximate counting and decision. See further below for the definition of the matching
        and star criterion.
        All run-time lower bounds rely on the
        Exponential Time Hypothesis, and the absence of FPTRASes relies on the assumption
        that $\W{1}$ does not coincide with $\ccFPT$ under randomised parameterized
        reductions. We write ``mixed'' whenever the respective classes contain both
        tractable properties and hard properties.
        The known results about the clique problem are added for completeness; note that
        $\W1$-hardness of decision immediately rules out an FPTRAS for approximate
        counting under the previous assumptions.\\[.8ex]
        {$^\dag$\scriptsize We assume that the minor-closed property $\Phi$ does not
            have bounded matching number, is not trivially true and that each forbidden minor has a
        vertex of degree at least $3$.}\\
        {$^\ddagger$\scriptsize $\Phi=\mathsf{true}$ and $\Phi=\mathsf{false}$ always
            yield fixed-parameter tractability of exact counting. $\Phi(H)=1 \Leftrightarrow
            H$ is a matching yields $\#\W1$-hardness of exact counting~\cite{Curticapean13}; note
        that the latter property is of bounded treewidth and satisfies the matching criterion.}\\
        {$^\ast$\scriptsize $\Phi=\mathsf{true}$ always yields fixed-parameter
            tractability of exact counting. $\Phi(H)=1 \Leftrightarrow$ ($H$ is a matching or
            a star) yields $\#\W1$-hardness by~\cref{lem:coef_special_case_intro}; note that the
        latter property satisfies the matching criterion and the star criterion.}\\
        {$^\S$\scriptsize $\Phi=\mathsf{true}$ always yields an FPTRAS for approximate
            counting. $\Phi=\Psi$ (from Definition~\ref{Def:propertyPsi}) does not allow for an FPTRAS while
    satisfying the matching criterion.}}
\end{table}

\paragraph*{Complete Classification for Minor-Closed Properties}
Let us start with the case where the graph property $\Phi$ is closed under taking minors,
that is, if $\Phi$ holds for a graph, then $\Phi$ still holds after removing vertices or
edges, or after contracting edges.
For minor-closed properties $\Phi$, we obtain a complete picture of the complexity of
$\#\edgesubsprob(\Phi)$ and $\edgesubsprob(\Phi)$. In what follows, we say that a property $\Phi$ has \emph{bounded matching number} if there is a constant bound
on the size of a largest matching in graphs satisfying $\Phi$.

\begin{restatable}{mtheorem}{minclosehard}\label{thm:minclose-hard}
    Let $\Phi$ denote a minor-closed graph property.
    \begin{enumerate}
        \item {\sf Exact Counting}: If $\Phi$ is either trivially true or of
            bounded matching number,
            then the (exact) counting version $\#\edgesubsprob(\Phi)$ is fixed-parameter
            tractable. Otherwise, the problem $\#\edgesubsprob(\Phi)$ is $\#\W{1}$-hard.
            If, additionally, each forbidden minor of $\Phi$ has a vertex of degree at
            least $3$, and the Exponential Time Hypothesis holds, then
            $\#\edgesubsprob(\Phi)$ cannot be solved in time $f(k)\cdot |G|^{o(k/\log k)}$,
            for any function $f$.
        \item {\sf Approximate Counting}: The problem $\#\edgesubsprob(\Phi)$ always
            has a fixed-parameter tractable randomised approximation scheme
            (FPTRAS).\ifx\minclosehardlend\undefined\footnote{The
                formal definition is given in Section~\ref{sec:prelims_param}; intuitively an FPTRAS is the
            parameterized equivalent of a fully polynomial-time randomised approximation
            scheme (FPRAS).}\fi
        \item {\sf Decision:} The problem $\edgesubsprob(\Phi)$ is always fixed-parameter
            tractable.
        \ifx\minclosehardlend\undefined\lipicsEnd\fi
    \end{enumerate}
\end{restatable}
\def\minclosehardlend{1}\pagebreak

\noindent
Consider for example the property $\Phi$ of being planar: planar graphs do not have
bounded matching number. Additionally, by Kuratowski's Theorem, the forbidden minors of
planar graphs are the $3$-biclique $K_{3,3}$ and the $5$-clique $K_5$. Since both
$K_{3,3}$ and $K_5$ contain a vertex of degree at least $3$, we conclude that computing
the number of planar subgraphs with $k$ edges in a graph $G$ is $\#\W{1}$-hard and,
assuming ETH, cannot be solved in time $f(k)\cdot |G|^{o(k/\log k)}$ for any function $f$.
In sharp contrast, approximating the number of planar subgraphs with $k$ edges in a graph,
as well as deciding whether there is such a planar subgraph can be done efficiently.
We obtain \cref{thm:minclose-hard} as a combination of our (more general) results for each
of the three settings that we study; we discuss these results next.

\paragraph*{Results for Exact Counting}
Let us return to the case of arbitrary graph properties $\Phi$. Without any further
assumptions on $\Phi$, the naive algorithm for $\#\edgesubsprob(\Phi)$ on the input
$(k,G)$ proceeds by enumerating the $k$-edge subsets of~$G$ and counting the number of
cases where the corresponding subgraph satisfies $\Phi$. This leads to a running time of the
form $f(k) \cdot |V(G)|^{2k+O(1)}$.
However, at least the linear constant in the exponent can be substantially improved using
the currently fastest known algorithm for counting subgraphs with $k$ edges due to
Curticapean, Dell and Marx~\cite{CurticapeanDM17}.
We will show that it easily extends to the case of $\#\edgesubsprob(\Phi)$:

\begin{restatable}{proposition}{bestalgo}\label{cor:best_algo}
    Let $\Phi$ denote a computable graph property. Then $\#\edgesubsprob(\Phi)$ can be solved
    in time $f(k)\cdot |V(G)|^{0.174k +o(k)}$,
    where $f$ is some computable function.
    \ifx\corintrofurlend\undefined\lipicsEnd\fi
\end{restatable}

On the other hand, it was shown by Curticapean and Marx~\cite{CurticapeanM14} that for the
property $\Phi$ of being a matching, the problem $\#\edgesubsprob(\Phi)$ cannot be solved
in time $f(k)\cdot |V(G)|^{o(k/\log k)}$ for any function $f$, unless ETH fails. In other
words, asymptotically and up to a factor of $1/\log k$, the exponent of $|V(G)|$ in the
running time of $\#\edgesubsprob(\Phi)$ cannot be improved without posing any restriction
on $\Phi$.

The goal is hence to identify properties $\Phi$ for which the algorithm in
\cref{cor:best_algo} \emph{can} be (significantly) improved. In the best possible outcome, we
hope to identify the properties for which the exponent of $|V(G)|$ does not depend on $k$;
those cases are precisely the fixed-parameter tractable ones. An easy consequence of known
results for subgraph counting (see for instance \cite{CurticapeanM14}) establishes the following
tractability criterion; we will include the proof only for the sake of completeness:

\begin{restatable}{proposition}{smallmatchFPT}\label{thm:smallmatchFPT}
    Let $\Phi$ denote a computable graph property satisfying that there is $M > 0$ such that
    for all $k$ either the graphs with $k$ edges satisfying $\Phi$ or the graphs with $k$
    edges satisfying $\neg \Phi$ have matching number bounded by $M$. Then
    $\#\edgesubsprob(\Phi)$ is fixed-parameter tractable.
    \ifx\corintrofurlend\undefined\lipicsEnd\fi
\end{restatable}

\noindent Examples of properties satisfying the tractability criterion of \cref{thm:smallmatchFPT}
include, among others, the property of being a star, or the complement thereof.
We conjecture that all remaining properties induce $\#\W{1}$-hardness and rule out any
algorithm running in time $f(k)\cdot |G|^{o(k/\log k)}$ for any function $f$, unless ETH
fails.\footnote{Note that it does not matter whether we choose~$|G|$ or~$|V(G)|$ for the
size of the large graph since we care about the asymptotic behaviour of the exponent.} For
the case of minor-closed graph properties, we have seen above that this conjecture holds.

Further, the techniques we develop to prove hardness of $\#\edgesubsprob(\Phi)$ for
minor-closed properties $\Phi$ in \cref{thm:minclose-hard} can also be applied directly to
show hardness for other specific properties $\Phi$. Below, we record several natural
examples of such properties which are covered by our methods.

\begin{restatable}{mtheorem}{corintrofur}\label{cor:intro_further}
    Consider the following graph properties.
    \begin{itemize}
        \item $\Phi_C(H)=1$ if and only if $H$ is connected.
        \item $\Phi_H(H)=1$ if and only if $H$ is Hamiltonian.
        \item $\Phi_E(H)=1$ if and only if $H$ is Eulerian.
        \item $\Phi_{CF}(H)=1$ if and only if $H$ is claw-free.
    \end{itemize}
    For $\Phi \in \{ \Phi_C, \Phi_H, \Phi_E, \Phi_{CF} \}$,
    the problem $\#\edgesubsprob(\Phi)$ is $\#\W{1}$-hard.
    Further, unless ETH fails, the problem $\#\edgesubsprob(\Phi)$ cannot be solved in time
    $f(k)\cdot |G|^{o(k/\log k)}$ for any function $f$.
\ifx\corintrofurlend\undefined\lipicsEnd\fi
\end{restatable}

\paragraph*{Results for Approximate Counting and Decision}

Our results on exact counting indicate that we have to relax the problem if we aim for
tractability results for a larger variety of properties. One approach is to only ask for
an \emph{approximate} count of the number of $k$-edge subgraphs satisfying $\Phi$.
Tractability of approximation in the parameterized setting is given by the notion of a
\emph{fixed-parameter tractable randomized approximation scheme} (FPTRAS) as introduced by
Arvind and Raman~\cite{ArvindR02}. While we give the formal definition in
Section~\ref{sec:prelims_param}, it suffices for now to think of an FPTRAS as a fixed-parameter
tractable algorithm that can compute an arbitrarily good approximation of the answer with
high probability. Readers familiar with the classical notions of approximate counting
algorithms should think of an FPTRAS as an FPRAS in which we additionally allow a factor
of $f(k)$ in the running time, for any computable function $f$.

For the statement of our results, we say that a property $\Phi$ satisfies the
\emph{matching criterion} if it is true for all but finitely many matchings, and we say
that it satisfies the \emph{star criterion} if it is true for all but finitely many stars.
Furthermore, we say that $\Phi$ has bounded treewidth if there is a constant upper bound
on the treewidth of graphs that satisfy $\Phi$.

\begin{restatable}{mtheorem}{apxmain}\label{thm:approx_main_intro}
    Let $\Phi$ denote a computable graph property. If $\Phi$ satisfies the matching criterion
    \textbf{and} the star criterion, or if $\Phi$ has bounded treewidth, then
    $\#\edgesubsprob(\Phi)$ admits an FPTRAS.
    \lipicsEnd
\end{restatable}
For example, the property of being planar satisfies both, the star and the matching
criterion. Moreover, we can show that every minor-closed graph property $\Phi$ has either bounded
treewidth or satisfies matching and star criterion, and thus always admits an FPTRAS.
Additionally, if not only exact but also approximate counting is intractable, we ask
whether we can at least obtain an efficient algorithm for the decision version
$\edgesubsprob(\Phi)$. Again, we obtain a tractability criterion; observe the subtle
difference in the tractability criterion compared to \cref{thm:approx_main_intro}.

\begin{restatable}{mtheorem}{decclass}\label{thm:dec_classification_intro}\label{thm:dec_classification}
    Let $\Phi$ denote a computable graph property. If $\Phi$ satisfies the matching criterion
    \textbf{or} the star criterion, or if $\Phi$ has bounded treewidth, then
    $\edgesubsprob(\Phi)$ is fixed-parameter tractable.
    \lipicsEnd
\end{restatable}

As an easy corollary, we can conclude that for monotone, that is, subgraph-closed
properties~$\Phi$, the problem $\edgesubsprob(\Phi)$ is always fixed-parameter
tractable.\footnote{Every graph property has either bounded treewidth or unbounded
matching number. In the latter case, if the property is additionally monotone, it must
satisfy the matching criterion.}

For many previously studied problems, the complexity analysis of approximate
counting and decision were related: often an algorithm solving one setting can be used to
solve the other setting~\cite{Meeks16,DellLM20}.
However, in our results \cref{thm:approx_main_intro,thm:dec_classification_intro} we see
an asymmetry between the two settings: it suffices for $\Phi$
to satisfy only one of the star and the matching criterion to induce tractability of the
decision version, but we require satisfaction of both for approximate counting. One might
expect that this reflects a shortcoming of our proof methods (and that in fact it suffices
to check one of the criteria to have tractability of approximate counting). Interestingly,
this is \emph{not} the case:

\begin{restatable}{proposition}{sepapxdec}\label{thm:sep_approx_dec_intro}
    There is a computable graph property $\Psi$ (see Definition~\ref{Def:propertyPsi}) that
    satisfies the matching criterion, but not the star criterion, such that
    $\edgesubsprob(\Psi)$ is fixed-parameter tractable, but $\#\edgesubsprob(\Psi)$ does
    not admit an FPTRAS unless $\W{1}$ coincides with $\ccFPT$ (the class of all fixed-parameter
    tractable decision problems) under randomised parameterized reductions.
    \ifx\corintrofurlend\undefined\lipicsEnd\fi
\end{restatable}

\paragraph*{Dichotomy for Evaluating a parameterized Tutte Polynomial}

As a final part of the presentation of our main results, let us discuss our results
on a parameterized Tutte polynomial.

The classical Tutte polynomial (as well as its specializations like the chromatic, flow or
reliability polynomial) have received widespread attention, both from a combinatorial as
well as a complexity theoretic
perspective~\cite{JaegerVW90,AlonFW95,Vertigan98,BjorklundHKK08,GoldbergJ14,Delletal14,BrandDR16,BjorklundK21}.
The classical Tutte polynomial is of special interest from a complexity theoretic
perspective, as the Tutte polynomial encodes a plethora of properties of a graph:
prominent examples include the chromatic number, the number of acyclic orientations,
and the number of spanning trees; we refer the reader to the work of
Jaeger et al.\ \cite{JaegerVW90} for a comprehensive overview.
Formally, the Tutte polynomial is a bivariate graph polynomial defined as follows
(see \cite{JaegerVW90}):
\[T_G(x,y) := \sum_{A \subseteq E(G)} (x-1)^{k(A)-k(E(G))} \cdot (y-1)^{k(A)+\#A-\#V(G)} \,, \]
where $k(S)$ is the number of connected components of the graph $(V(G),S)$.
In the aforementioned work, Jaeger et al.\ \cite{JaegerVW90} also classified the complexity
of evaluating the Tutte Polynomial in every pair of (complex) coordinates, that is, for
every pair $(a,b)$, the complexity of computing the function $G \mapsto T_G(a,b)$ is fully
understood.

In this work, we consider the following parameterized version of the Tutte Polynomial by
restricting to edge-subsets~$A$ in~$G$ of size~$k$:
\[T^k_G(x,y) := \sum_{A \in \binom{E(G)}{k}} (x-1)^{k(A)-k(E(G))} \cdot (y-1)^{k(A)+k-\#V(G)} \,. \]
We observe that the parameterized Tutte polynomial can be seen as a weighted version of
counting small $k$-edge subgraph patterns by assigning to each $k$-edge subset $A$ of $G$
the weight \[(x-1)^{k(A)-k(E(G))} \cdot (y-1)^{k(A)+k-\#V(G)}\,.\]
Moreover, we point out that $T^k_G(x,y)$ is related to a generalization of the bases generating function for matroids~\cite{AnariLGV19}.
By establishing a so-called deletion-contraction recurrence, we show that $T^k_G(x,y)$ has
similar expressive power as its classical counterpart $T_G(x,y)$:

\begin{restatable}{mtheorem}{tutteprops}\label{mthm:tutte_props}
For any graph $G$ and positive integer $k$, the following graph invariants are encoded
    in $T^k_G(x,y)$:
\begin{enumerate}
    \item $T^k_G(2,1)$ is the number of $k$-forests in $G$. In other words
        $T^k_G(2,1)$ corresponds to the problem $\#\edgesubsprob(\Phi)$ for the property $\Phi$ of
        being a forest.
    \item For each positive integer $c$, the values of $T^k_G(1-c,0)$
        determine\footnote{They are equal up to trivial modifications; in particular,
        their complexities coincide.} the numbers of pairs $(A,\sigma)$, where $A$ is a
        $k$-edge subset of~$G$, and $\sigma$ is a proper $c$-colouring of $(V(G),A)$.
    \item From $T^k_G(2,0)$ we can compute the numbers of pairs $(A,\vec{\eta})$, where
        $A$ is a $k$-edge subset of $G$, and $\vec{\eta}$ is an acyclic orientation of
        $(V(G),A)$.
    \item $T^k_G(2,0)$ also determines the number of $k$-edge subsets $A$ of $G$, such
        that $(V(G),A)$ has even Betti Number (we give a formal definition of the Betti
        number in Section~\ref{sec:indpoints}).
    \item $T^k_G(0,2)$ determines the number of $k$-edge subsets $A$ of $G$, such that
        $(V(G),A)$ has an even number of components.
\end{enumerate} 
\end{restatable}
$~$\\
\noindent Note that, while $\#\edgesubsprob(\Phi)$ only allows us to count the number of subgraphs
with $k$ edges that satisfy~$\Phi$, the parameterized Tutte polynomial allows us to count
more intricate objects, such as tuples of an edge-subset and a colouring (or acyclic
orientation) on the induced graph.
From a complexity theoretic point of view, we obtain a similar result
as~\cite{JaegerVW90}, albeit only for rational coordinates:
for each \emph{fixed} pair $(x,y)$ of coordinates, we consider the problem receiving as
input a graph~$G$ and a positive integer $k$ and computing $T^k_G(x,y)$. Following the
paradigm of this work, we choose $k$ as a parameter, that is, we consider inputs in which~$k$ is significantly smaller than~$|G|$.

\begin{restatable}{mtheorem}{tuttemain}\label{thm:tutte_main_param_intro}\label{thm:tutte_main_param}
	Let $(x,y)$ denote a pair of rational numbers. The problem of computing $T^k_G(x,y)$ is
	solvable in polynomial-time if $x=y=1$ or $(x-1)(y-1)=1$, fixed-parameter tractable, but $\#\ccP$-hard, if $x=1$ and $y\neq 1$, and $\#\W{1}$-hard otherwise.
    \ifx\corintrofurlend\undefined\lipicsEnd\fi
\end{restatable}

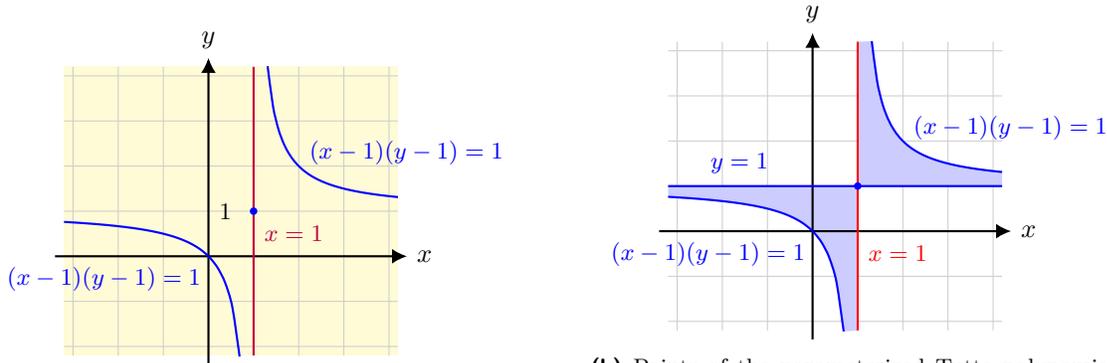
\begin{figure}[t]
    \centering
    \begin{subfigure}[b]{.48\textwidth}
        \begin{tikzpicture}[scale=0.6]
            \fill[yellow!20] (-3.2,-2.2) rectangle (4.2,4.2) {};
            \foreach\i in {-3,...,4}{
                \draw[gray!40,very thin] (\i,-2.2) -- (\i,4.2);
            }
            \foreach\i in {-2,...,4}{
                \draw[gray!40,very thin] (-3.2, \i) -- (4.2, \i);
            }

            \node[label=right:{\small 1}] at (0,1) {};
            \node[red,label=right:{\color{purple} \small $x = 1$}] at (1,.5) {};
            \node[red,label=right:{\color{blue} \small $(x - 1)(y - 1) = 1$}] at (2,2.3) {};
            \node[red,label=right:{\color{blue} \small $(x - 1)(y - 1) = 1$}] at
                (-4.7,-.5) {};

            \draw[thick,-{Latex[length=2mm, width=2mm]}] (-3.4, 0) -- (4.4, 0) node[right] {$x$};
            \draw[thick,-{Latex[length=2mm, width=2mm]}] (0, -2.4) -- (0, 4.4) node[above] {$y$};
            \draw[thick,domain=-2.2:4.2, smooth, variable=\x, purple] plot ({1}, {\x});
            \draw[thick,domain=-3.2:.69, smooth, variable=\x, blue]  plot ({\x}, {\x/(\x-1) });
            \draw[thick,domain=1.311:4.2, smooth, variable=\x, blue]  plot ({\x}, {\x/(\x-1) });

            \node[fill=blue,circle,inner sep=0.1em] at (1,1) {};
        \end{tikzpicture}
        \caption{Points of the parameterized Tutte polynomial that can be computed in
            polynomial-time (blue) and that are fixed-parameter tractable, but $\#\ccP$-hard
        (red). Exact computation at any other point (yellow) is $\#\W{1}$-hard.}
    \end{subfigure}
    ~
    \begin{subfigure}[b]{.48\textwidth}
        \begin{tikzpicture}[scale=0.6]
            \foreach\i in {-3,...,4}{
                \draw[gray!40,very thin] (\i,-2.2) -- (\i,4.2);
            }
            \foreach\i in {-2,...,4}{
                \draw[gray!40,very thin] (-3.2, \i) -- (4.2, \i);
            }

            \node[label=right:{\color{blue}\small $y = 1$}] at (-2.5,1.5) {};
            \node[red,label=right:{\color{red} \small $x = 1$}] at (1,-.5) {};
            \node[red,label=right:{\color{blue} \small $(x - 1)(y - 1) = 1$}] at (2,2.3) {};
            \node[red,label=right:{\color{blue} \small $(x - 1)(y - 1) = 1$}] at (-4.7,-.5) {};

            \fill[,domain=-3.2:.69, smooth, variable=\x, blue!20]
                (-3.2,1) --
                plot ({\x}, {\x/(\x-1) }) --
                (1,-2.2) --
                (1,1) --
                cycle;
            \fill[thick,domain=1.311:4.2, smooth, variable=\x, blue!20]
                (1,4.2) --
                plot ({\x}, {\x/(\x-1) })--
                (4.2,1) --
                (1,1) --
                cycle;

            \draw[thick,-{Latex[length=2mm, width=2mm]}] (-3.4, 0) -- (4.4, 0) node[right] {$x$};
            \draw[thick,-{Latex[length=2mm, width=2mm]}] (0, -2.4) -- (0, 4.4) node[above] {$y$};
            \draw[thick,domain=-2.2:4.2, smooth, variable=\x, red] plot ({1}, {\x});
            \draw[thick,domain=-3.2:4.2, smooth, variable=\x, blue] plot ({\x}, {1});

            \draw[thick,domain=-3.2:.69, smooth, variable=\x, blue]  plot ({\x}, {\x/(\x-1) });
            \draw[thick,domain=1.311:4.2, smooth, variable=\x, blue]  plot ({\x}, {\x/(\x-1) });

            \node[fill=blue,circle,inner sep=0.1em] at (1,1) {};
        \end{tikzpicture}
        \caption{Points of the parameterized Tutte polynomial that allow for an FPRAS (blue) and for an FPTRAS (red); all
        points on the boundary of the blue area are included. The complexity of approximation is
    open for all points outside of the coloured region.}
        \label{fig:tutteeasy_intro_apx}
    \end{subfigure}
    \caption{Points of the parameterized Tutte polynomial that can be computed in FPT-time
    (a)~exactly or (b)~approximately. We emphasize that a full classification for exact
counting is established, while the complexity of approximation remains open outside of the
coloured  area.\label{fig:tutteeasy_intro}}
\end{figure}
$~$\\
\noindent The class $\#\ccP$ is the counting version of $\NP$~\cite{Valiant79,Valiant79b} and, in particular, the $\#\ccP$-hard cases in the above classification are not polynomial-time tractable unless the polynomial-time hierarchy collapses to $\ccP$~\cite{Toda91}.
Consider \cref{fig:tutteeasy_intro} for a depiction of the classification.
Note that \cref{thm:tutte_main_param_intro} yields $\#\W{1}$-hardness for each of the
aforementioned problems from \cref{mthm:tutte_props}.
Note further, that the tractable cases are similar, but not equal to the classical
counterpart~\cite{JaegerVW90}.

Moreover, our proof uses entirely different tools than \cite{JaegerVW90}
and illustrates the power and utility of the method presented in the
subsequent discussion of our techniques.

Having fully classified the complexity of exact evaluation of the parameterized Tutte Polynomial, we also consider the complexity of approximate evaluation. We identify two regions bounded by the hyperbola $(x-1)(y-1)=1$ and the lines $x=1$ and $y=1$ as efficiently approximable; consider Figure~\ref{fig:tutteeasy_intro_apx} for a depiction.

\begin{restatable}{mtheorem}{tutteapprox}\label{thm:tutteapprox_intro}\label{thm:tutteapprox}
Let $(x,y)$ denote a pair of rational numbers. If $0\leq (x-1)(y-1) \leq 1$, then $T^k_G(x,y)$ has an FPTRAS. If additionally $x \neq 1$ or $y=1$ then $T^k_G(x,y)$ even has a fully polynomial-time randomized approximation scheme (FPRAS).
\end{restatable}

\subsection*{Techniques}\label{Sect:Techniques}

Our \cref{thm:approx_main_intro,thm:dec_classification_intro,thm:tutteapprox_intro} are obtained easily:
the proof of \cref{thm:approx_main_intro} is a standard application (see for instance
\cite{Meeks16}) of the Monte-Carlo approach, in combination with Ramsey's theorem, and
Arvind and Raman's algorithm for approximately counting subgraphs of bounded
treewidth~\cite{ArvindR02}.
The proof of \cref{thm:dec_classification_intro} uses a standard parameterized
Win-Win approach for graphs of bounded treewidth or bounded degree. Finally, the proof of \cref{thm:tutteapprox_intro} is an easy consequence of the work of Anari et al.~\cite{AnariLGV19} on approximate counting via log-concave polynomials.

Hence, in this technical discussion, we want to focus on the technique that enables us to
prove the lower bounds for \cref{thm:minclose-hard,cor:intro_further} and,
perhaps surprisingly, also for \cref{thm:tutte_main_param_intro}.

As a main component, we use the Complexity Monotonicity framework of Curticapean,
Dell and Marx~\cite{Curticapean13}. Given a property $\Phi$ and a positive integer $k$, we
write $\#\edgesubs{\Phi,k}{\star}$ for the function that maps a graph $G$ to the number of
$k$-edge subgraphs of $G$ that satisfy $\Phi$. Using a well-known transformation via
M\"obius inversion~\cite[Chapter 5.2]{Lovasz12}, we can show that there are rational
numbers $a_1,\dots,a_\ell$ and graphs $H_1,\dots,H_\ell$ such that for each graph $G$ we
have
\begin{equation}\label{eq:hombasis_intro}
    \#\edgesubs{\Phi,k}{G}=\sum_{i=1}^k a_i\cdot \#\homs{H_i}{G}\,,
\end{equation}
where $\#\homs{H_i}{G}$ is the number of graph homomorphisms from $H_i$ to $G$. In other
words, we can express $\#\edgesubs{\Phi,k}{\star}$ as a finite linear combination of
homomorphism counts. Here, we can then apply the Complexity Monotonicity
framework~\cite{CurticapeanDM17}, which asserts that computing a finite linear combination
of homomorphism counts is \emph{precisely as hard as} its hardest term
(among the terms with a non-zero coefficient).
However, the complexity of computing the number of homomorphisms from small pattern graphs
to large host graphs is very well-understood~\cite{DalmauJ04,Marx10}.
Roughly speaking, the higher the treewidth of the pattern graph, the harder the problem becomes;
we make this formal in Section~\ref{sec:prelims_param}.

Instead of our original problem $\#\edgesubsprob(\Phi)$, we can thus consider the
problem of computing linear combinations of graph homomorphism counts.
In particular, to obtain hardness, it suffices to understand
for which of the coefficients in equation~\eqref{eq:hombasis_intro} we have $a_i \neq 0$,
depending on $k$ and $\Phi$.

Relying on the well-known fact that the M\"obius function of
the partition lattice alternates in sign,
Curticapean, Dell, and Marx~\cite{CurticapeanDM17} observed that non-trivial cancellations
cannot occur in equation~\eqref{eq:hombasis_intro} if, for each $k$, every $k$-edge graph
that satisfies $\Phi$ must have the same number of vertices. Consequently, if the matching
number is unbounded, those properties yield $\#\W{1}$-hardness. An example for such a
property is the case of $\Phi(H)=1$ if and only if $H$ is a tree.
In contrast, the intractability result for the case of $\Phi=$ acyclicity (that is, being a
forest) turned out to be much harder to show~\cite{BrandR17}, indicated by connections to
parameterized counting problems in matroid theory.

In later work, the coefficients $a_i$ were shown to have even more interesting structure:
the coefficients~$a_i$  describe topological and
algebraic invariants of the set of pattern graphs.
For example, in~\cite{RothS18} it was shown that the coefficient of the $k$-clique in case of
counting vertex-induced subgraphs with property $\Phi$ is the reduced Euler characteristic
of a simplicial complex associated with $\Phi$ and can thus, if non-zero, be used to
establish evasiveness of certain graph properties~\cite{KahnSS84}.

In this work, we prove additional insights into said coefficients $a_i$.\footnote{For technical reasons, the approach we describe below requires us to consider a coloured version of $\#\edgesubsprob(\Phi)$, which is, however, shown to be interreducible with the uncoloured one.}
For any graph~$H$ we give an explicit formula for its coefficient $a_H$ in terms of
a sum over the \emph{fractures} on~$H$, an additional combinatorial structure on a graph
$H$ resembling, to some extent, a gadget construction used for the classification of the
subgraph counting problem~\cite{CurticapeanM14}.
Our most crucial insight is then that we can drastically simplify the expression of the
coefficient $a_H$ modulo a prime $\ell$ if $H$ admits a vertex-transitive action of a
group of order given by a power of $\ell$. In this case, we obtain an action of the group
on the set of fractures on $H$ and in the formula for $a_H$ all contributions from
fractures not fixed by the group cancel out modulo $\ell$.

In particular, we consider graphs $H$ which are Cayley graphs of a finite group of prime
power order and a symmetric set of generators.
Since the Cayley graph of a group always has a natural vertex-transitive action of this
group, such Cayley graphs always have the desired symmetry properties.
We exploit this by showing that there is a constant number of fractures fixed by
the group action. This in turn allows us to write $(a_H$ modulo $\ell)$ as a finite sum of terms
depending on the value of $\Phi$ on some explicit graphs.

Specifically, the first set of Cayley graphs we consider are the toroidal grids
$\torus_\ell$, which are depicted in \cref{fig:torgrid_intro}.
Since the treewidth of $\torus_\ell$ diverges
with $\ell$, we thus obtain a $\#\W{1}$-hardness result whenever
the coefficient $a_{\torus_\ell}$ does not vanish for infinitely many $\ell$.

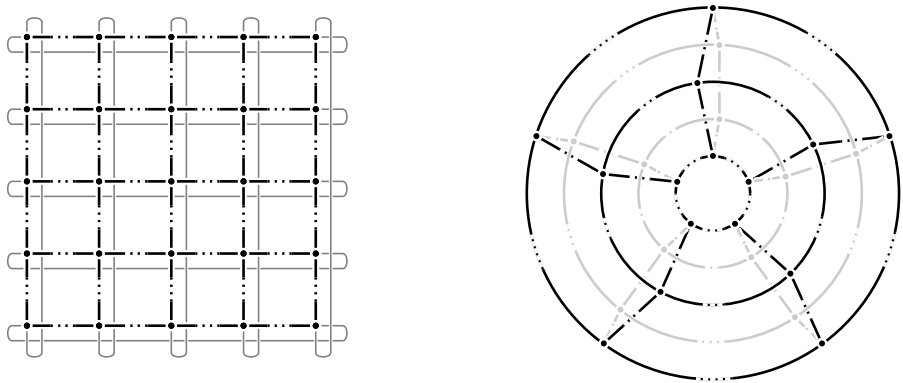
\begin{figure}[t]
    \centering
    \begin{tikzpicture}[scale=0.8,transform shape]
        \pic[sty=10] at (0,0) {grid=5/14};
        \pic[sty=10] at (9,-.2) {torus=5/14};
    \end{tikzpicture}
    \caption{Two isomorphic representations of the toroidal grid $\torus_\ell$: On the left hand side as
    a grid with connected endpoints, on the right hand side as a stylized torus.
}\label{fig:torgrid_intro}\vspace{-2ex}
\end{figure}

\noindent Writing  $M_k$ for the matching of size $k$,
$P_2$ for the path consisting of $2$ edges, $C_k$ for the cycle of length $k$,
$S_k$ for a  sun (a cycle with dangling edges) of size $k$, and $\torus_k$ for the
toroidal grid of size $k$, our first main technical result reads as follows:

\begin{theorem}[Simplified version]\label{lem:coef_special_case_intro}
    Let $\Phi$ denote a computable graph property and assume that  infinitely many primes
    $\ell$ satisfy the equation\footnote{We write $+$ for (disjoint) graph union and $\ell
        H$ for the graph consisting of $\ell$ disjoint copies of $H$.
        Further, we set $\Phi(H)=1$ if $H$ satisfies $\Phi$ and $\Phi(H)=0$ otherwise.} \begin{small}
    \begin{equation} \label{eqn:coef_special_case_intro}
        -6\Phi(M_{2\ell^2}) +4 \Phi(M_{\ell^2} + \ell C_\ell)+ 8\Phi(\ell^2P_2)  -
        \Phi(2\ell C_\ell) -2\Phi(\ell C_{2\ell}) -4\Phi(\ell S_\ell) + \Phi(\torus_\ell)
        \neq 0 \mod \ell \,.
    \end{equation} \end{small}
    Then $\#\edgesubsprob(\Phi)$ is $\#\W{1}$-hard.\lipicsEnd
\end{theorem}

As a toy example for an application of \cref{lem:coef_special_case_intro}, let us consider
the property $\Phi$ of being connected. Observe that among the graphs in
\eqref{eqn:coef_special_case_intro}, only $\torus_\ell$ is connected, and thus the sum is
always $1$ for $\ell \geq 2$. Thus, indeed the left-hand side of
\eqref{eqn:coef_special_case_intro} is nonzero, proving that $\#\edgesubsprob(\Phi)$ is
$\#\W{1}$-hard.

Using \cref{lem:coef_special_case_intro}, we can prove most of the $\#\W{1}$-hardness
results of \cref{thm:minclose-hard}. However, using the toroidal grid~$\torus_\ell$ we
cannot prove (almost) tight conditional lower bounds: the treewidth of~$\torus_\ell$ grows only
with the \emph{square-root} of the parameter $k$ (that is the number of edges of the graph).
To address this problem, we consider a second family of $4$-regular Cayley graphs,
constructed explicitly by Peyerimhoff and Vdovina~\cite{PeyerimhoffV11}, which have the
additional property of being \emph{expander graphs}.
In particular, for these graphs, the treewidth grows \emph{linearly} in the number of
edges. This allows us to obtain almost tight conditional lower bounds.
The variant of \cref{lem:coef_special_case_intro} for these Cayley graph expanders is given by \cref{lem:expander_coeffs} (in combination with \cref{lem:hardness_basis}).

The only drawback of the Cayley graphs from \cite{PeyerimhoffV11} is that the
corresponding groups always have orders given by powers of $2$ (in contrast to having
arbitrary primes $\ell$ in \cref{lem:coef_special_case_intro}). Hence, our criterion for
hardness is the nonvanishing of some expression modulo $2$. Ultimately, this is the
reason why for the conditional lower bounds in \cref{thm:minclose-hard} we need to exclude
forbidden minors having a vertex of degree $2$ or less.

Finally, to obtain \cref{thm:tutte_main_param_intro}, we express the parameterized Tutte polynomial
at a rational point $(x,y)$ as a linear combination of (fractures of) toroidal grids;
the proof of \cref{lem:coef_special_case_intro} then essentially shows that this linear
combination always contains a graph with unbounded treewidth, yielding $\#\W1$-hardness.

\def\corintrofurlend{1}

\newpage
\tableofcontents

\newpage

\section{Preliminaries}\label{sec:prelims}
Given a finite set $S$, we write $|S|$ and $\#S$ for the cardinality of $S$. Further,
given a function $f: X\times Y \rightarrow Z$ and an element $x\in X$, we write
$f(x,\star) : B \rightarrow C$ for the function $y \mapsto f(x,y)$.

\subsection{Graphs and Homomorphisms}
Graphs in this work are simple and irreflexive, that is, we do not allow multiple edges or
self-loops. Given a graph~$G$ and a subset $A$ of $E(G)$, we write $G(A)$ for the graph
$(V(G),A)$, and we write $G[A]$ for the graph obtained from $G(A)$ by deleting all
isolated vertices.

Given graphs $F$ and $G$, a \emph{homomorphism} from $F$ to $G$ is a mapping $\varphi:
V(F) \rightarrow V(G)$ which is edge-preserving, that is, for each edge $\{u,v\}\in E(F)$
we have $\{\varphi(u),\varphi(v)\}\in E(G)$. A homomorphism is called an \emph{embedding}
if it is injective (on the vertices). We write $\homs{F}{G}$ and $\embs{F}{G}$ for the set
of all homomorphisms and embeddings, respectively, from $F$ to $G$.

An \emph{isomorphism} between two graphs $F$ and $G$ is a bijective embedding $\varphi$
satisfying the stronger constraint $\{u,v\}\in E(F) \Leftrightarrow
\{\varphi(u),\varphi(v)\}\in E(G)$. We say that $F$ and $G$ are isomorphic, denoted by
$F\cong G$, if an isomorphism from $F$ to $G$ exists. An isomorphism from a graph $F$ to
itself is called an \emph{automorphism} and we write $\auts{F}$ for the group formed by
such automorphisms (where the group operation is the composition of automorphisms).

A graph $F$ is a \emph{minor} of a graph $G$ if it can be obtained from $G$ by a sequence
of edge- and vertex-deletions and edge-contractions (with multiple edges and self-loops
deleted).

A graph $G$ is called $k$\emph{-edge-coloured} if the edges of $G$ are coloured with (at
most) $k$ pairwise different colours. Given a homomorphism $\varphi\in\homs{G}{H}$ for
some graphs $G$ and $H$, we also call $\varphi$ an $H$\emph{-colouring}. Moreover, an
$H$\emph{-coloured graph} is a pair of a graph $G$ and an $H$-colouring $\varphi$.
For ease of readability, we say that a graph~$G$ is $H$-coloured if the $H$-colouring is
implicit or clear from the context. Observe that an $H$-colouring $\varphi$ of a graph $G$
induces a $\#E(H)$-edge-colouring by mapping an edge $\{u,v\}\in E(G)$ to the colour
$\{\varphi(u),\varphi(v)\}$. Throughout this work, we use the following notion of
homomorphisms between $H$-coloured graphs:

\begin{defn}
    Let $F$ and $G$ denote two $H$-coloured graphs and let $c_F$ and $c_G$ denote their
    $H$-colourings. A homomorphism $\varphi$ from $F$ to $G$ is called
    \emph{colour-preserving} if $c_G(\varphi(v)) = c_F(v)$ for every $v \in V(F)$. We
    write $\necphoms{F}{H}{G}$ for the set of all colour-preserving homomorphisms from $F$
    to $G$. \emph{Colour-preserving embeddings} and the set $\necpembs{F}{H}{G}$ are
    defined similarly.

    Further, we say that two $H$-coloured graphs $F$ and $G$ are isomorphic \emph{as
    $H$-coloured graphs}, denoted by $F\cong_H G$, if there is a
    colour-preserving isomorphism from $F$ to $G$.
    \lipicsEnd
\end{defn}
Note that, given two $H$-coloured graphs $F$ and $G$, we write $F\cong G$ (rather than
$F\cong_H G$) if the underlying \emph{uncoloured} graphs are isomorphic.

For the treatment of decision and approximate counting, we introduce the following
classification criteria on (computable) graph properties. To this end, we write
$K_{\ell,r}$ for the biclique with $\ell$ vertices on the left and $r$ vertices on the
right side, respectively. In particular, $K_{1,r}$ denotes the star of size $r$.
\begin{defn}
    Let $\Phi$ denote a computable graph property. We say that
    \begin{itemize}
        \item $\Phi$ satisfies the \emph{matching criterion} if $\Phi(M_k)=1$ for all but
            finitely many $k$.
        \item $\Phi$ satisfies the \emph{star criterion} if $\Phi(K_{1,k})=1$ for all but
            finitely many $k$.
        \item $\Phi$ has \emph{bounded treewidth} if there is a constant $B$ such that
            $\Phi$ is false on all graphs of treewidth at least~$B$.
    \lipicsEnd
    \end{itemize}
\end{defn}

\noindent For example, the properties of being bipartite or planar satisfy both, the matching and
the star criterion. Furthermore, the property of being $2$-regular has bounded treewidth,
while the criterion of just being regular satisfies only the matching criterion. Further,
the property of being a tree is an example that satisfies both, the star criterion and is
of bounded treewidth, while the property of being a forest satisfies all three criteria.

\paragraph*{Expander Graphs}
All (almost-tight) conditional lower bounds in this work rely on the existence of certain (families
of) expander graphs. Given a positive integer $d$, a rational $c>0$, and a class of graphs
$\mathcal{G}=\{G_1,G_2,\dots\}$ with $\#V(G_i)=n_i$, we call $\mathcal{G}$ a family of
$(n_i,d,c)$\emph{-expanders} if, for all $i$, the graph $G_i$ is $d$-regular and satisfies
\[\forall X \subseteq V(G_i): |S(X)| \geq c\left( 1-\frac{|X|}{n_i} \right)|X| \,,\]
where $S(X)$ is the set of all vertices in $V(G_i)\setminus X$ that are adjacent to a
vertex in $X$.

While being sparse due to $d$-regularity, expander graphs have treewidth\footnote{We use
the graph parameter ``treewidth'' ($\mathsf{tw}$) in a black-box manner only,
and refer the reader to, for instance, Chapter~7 of~\cite{CyganFKLMPPS15} for a detailed
exposition.} linear in the number of vertices (see for instance Proposition~1 in~\cite{GroheM09}
and set $\alpha=1/2$). Furthermore, they admit arbitrarily large clique
minors~\cite{KleinbergR96}. Formally, we have:
\begin{fact}\label{fac:niceExpanders}
    Fix a rational $c$ and an integer $d$, and let $\mathcal{G}$ denote a family of
    $(n_i,d,c)$-expanders. Then, $\#E(G_i)\in \Theta(n_i)$ and $\mathsf{tw}(G_i)\in
    \Theta(n_i)$. Furthermore, for each positive integer $k$ there is an index $j$
    such that for all $i\geq j$, the graph~$G_i$ contains the complete graph on $k$
    vertices as a minor.
    \lipicsEnd
\end{fact}

\subsection{Parameterized Complexity Theory}\label{sec:prelims_param}
A \emph{parameterized counting problem} is a pair of a counting problem $P:\Sigma^\ast
\rightarrow \mathbb{N}$ and a \emph{parameterization} $\kappa: \Sigma^\ast \rightarrow
\mathbb{N}$. A \emph{parameterized decision problem} is a pair $(P,\kappa)$ of a decision
problem $P:\Sigma^\ast \rightarrow \{0,1\}$ and a parameterization $\kappa$. Consider for
example the problems $\#\textsc{Clique}$ and $\textsc{Clique}$: on input a graph~$G$
and a positive integer $k$, the task is to either compute the number of $k$-cliques in $G$ or
to detect the mere existence of a $k$-clique, respectively.
The parameterization is given by $\kappa(G,k):=k$ for both problems.

A parameterized problem $(P,\kappa)$ is called \emph{fixed-parameter tractable} (FPT) if
there is a computable function $f$ such that $P$ can be computed in time
$f(\kappa(x))\cdot |x|^{O(1)}$. For historic reasons, the class of all fixed-parameter
tractable \emph{decision} problems is called $\ccFPT$.\footnote{In some literature
$\ccFPT$ is used for both, parameterized decision and counting problems, while some
authors write $\mathsf{FFPT}$ for the class of all fixed-parameter tractable parameterized
counting problems.}
Furthermore, a \emph{parameterized Turing reduction} from $(P,\kappa)$ to
$(\hat{P},\hat{\kappa})$ is a Turing reduction from $P$ to $\hat{P}$ that, on input $x$,
runs in time $f(\kappa(x))\cdot |x|^{O(1)}$ and additionally satisfies
$\hat{\kappa}(y)\leq f(\kappa(x))$ for each oracle query $y$. Again, $f$ only needs to be
some fixed computable function. We write $(P,\kappa)\fptred (\hat{P},\hat{\kappa})$ if a
parameterized Turing reduction exists.

A parameterized counting problem is $\#\W{1}$\emph{-hard} if it can be reduced from
$\#\textsc{Clique}$, and, similarly, a parameterized decision problem is
$\W{1}$\emph{-hard} if it can be reduced from $\textsc{Clique}$, both with respect to
parameterized Turing reductions.

\noindent Under reasonable assumptions, such as the Exponential
Time Hypothesis (ETH)~\cite{ImpagliazzoP01}, defined below, $\#\W{1}$- and $\W{1}$-hard
problems are not fixed-parameter tractable.\footnote{In fact, Chen et al.\
\cite{Chenetal05,Chenetal06} proved the much stronger statement that $\#\textsc{Clique}$
cannot be solved in time $f(k)\cdot|V(G)|^{o(k)}$ for any function $f$, unless ETH fails.}
\begin{conjecture}[Exponential Time Hypothesis]
    The \emph{Exponential Time Hypothesis} (ETH) asserts that $3$-$\textsc{SAT}$ cannot be
    solved in time $\exp(o(n))$, where $n$ is the number of variables of the input formula.
    \lipicsEnd
\end{conjecture}

Our hardness results in this paper are obtained by reducing from the problem
$\#\homsprob(\mathcal{H})$. Given a fixed class of graphs $\mathcal{H}$,
in the problem $\#\homsprob(\mathcal{H})$ the input is
a graph $H\in \mathcal{H}$ and an arbitrary graph $G$ and the task is
to compute the number of homomorphisms from $H$ to $G$; the parameter is $|H|$. Dalmau and
Jonsson~\cite{DalmauJ04} established an exhaustive classification for this problem,
stating that $\#\homsprob(\mathcal{H})$ is fixed-parameter tractable if the treewidth of
graphs in~$\mathcal{H}$ is bounded by a constant, and $\#\W{1}$-hard, otherwise.

Let $\Phi$ denote a graph property, that is, a function from (isomorphism classes) of
graphs to $\{0,1\}$. Setting
\[\edgesubs{\Phi,k}{G} :=\{ A\subseteq E(G)~|~ \#A=k \wedge \Phi(G[A])=1 \} \,,\]
we define $\#\edgesubsprob(\Phi)$ as the parameterized counting problem in which
on input a graph $G$ and a positive integer $k$, the task is to compute
the number $\#\edgesubs{\Phi,k}{G}$; the parameter is $k$.

In this paper, we often rely on the following important, but easy observation: write
$\Phi_k$ for the set of graphs $H$ with $k$ edges and without isolated vertices, that
satisfy $\Phi$. Then we have
\begin{equation}\label{eq:unlcolouredSubgraphCounts}
    \#\edgesubs{\Phi,k}{G} = \sum_{H\in \Phi_k} \#\subs{H}{G} \,,
\end{equation}
where $\#\subs{H}{G}$ is the number of subgraphs of $G$ that are isomorphic to $H$.

Using the aforementioned transformation, both, \cref{cor:best_algo} and
\cref{thm:smallmatchFPT} are easy consequences of known results regarding the subgraph
counting problem. We add their proofs only for the sake of completeness:

\bestalgo*
\begin{proof}
    The fastest known algorithm for computing $\#\subs{H}{G}$ for a $k$-edge graph $H$
    runs in time $k^{O(k)}\cdot |V(G)|^{0.174k+o(k)}$ and is due to Curticapean, Dell and
    Marx~\cite{CurticapeanDM17}.  Now observe that the size of $\Phi_k$ is bounded by a
    function in $k$, since graphs in $\Phi_k$ have $k$ edges and no isolated vertices and
    thus have at most $2k$ vertices.
    Consequently, their algorithm extends to $\#\edgesubsprob(\Phi)$ by computing
    $\#\edgesubs{\Phi,k}{G}$ as given in Equation~\eqref{eq:unlcolouredSubgraphCounts}.
\end{proof}
Note that the growth of $f$ in the previous result depends, among other factors, on the
complexity of verifying $\Phi$.

For the following, recall that a property $\Phi$ has \emph{bounded matching number} if
there is a constant $c$ such that $\Phi$ is false on all graphs containing a matching of
size at least $c$. Furthermore, write $\neg \Phi$ for the complement of $\Phi$.

\smallmatchFPT*
\begin{proof}
    Applying Equation~\eqref{eq:unlcolouredSubgraphCounts}, we observe that
    counting subgraphs isomorphic to $H$ is fixed-parameter tractable (even
    polynomial-time solvable) if there is a constant upper bound on the size of the
    largest matching of~$H$~\cite{CurticapeanM14}. This allows us to compute
    $\#\edgesubs{\Phi,k}{G}$ in the desired running time if the graphs in~$\Phi_k$ have
    matching number bounded by $M$. In case the latter is true for $\neg \Phi_k$ instead,
    we use the fact that
    $\#\edgesubs{\Phi,k}{G}=\binom{\#E(G)}{k}-\#\edgesubs{\neg\Phi,k}{G}$ and proceed
    similarly.
\end{proof}

\subsection{Combinatorial Commutative Algebra}
We assume familiarity with the notions of basic group theory and refer the reader to
for instance \cite{Lang93} for a detailed introduction.
Given a positive integer $\ell$, we write $\mathbb{Z}_\ell$ for the group of integers
modulo $\ell$, and we write $\mathbb{Z}^k_\ell$ for the $k$-fold \emph{direct product} of
$\mathbb{Z}_\ell$; recall that the binary operation of the direct product is defined
coordinate-wise.

For a prime $p$, recall that a finite group $G$ is called a $p$-group if the order $\#G$
of $G$ is a power of $p$. Recall that by Lagrange's theorem, this implies that the order
of any subgroup $H$ of $G$ is likewise a power of $p$.

Given a group $G$ and a subgroup $H \subseteq G$, we write $G/H$ for the set of \emph{left
cosets} of $H$. Formally, a left coset of $H$ is an equivalence class
of the following equivalence relation on $G$: two elements $g, g' \in G$
are equivalent if and only if $g' = g h$ for some $h \in H$.
We write $g H$ for the equivalence class of $g \in G$.
We define the \emph{index} of $H$ in $G$ as the cardinality $[G : H] = \# G/H$
of the set of left cosets. Note that $[G:H]$ might be finite even though $\#G$ is
infinite; in fact, we encounter this case when we treat Cayley graph expanders of
$2$-groups. The index satisfies the basic identity
$\#G = [G:H] \cdot \# H$, and again, with a slight abuse of notation, we observe that
the identity remains well-defined in the infinite case: $\#G$ is infinite if and only if
one of $[G:H]$ or $\#H$ is infinite.
If the subgroup $H$ is \emph{normal} in $G$ (that is for each $g \in G$ we have that the
subset $g H g^{-1} = \{g h g^{-1} | h \in H\}$ of $G$ is equal to $H$), then the set $G/H$
naturally carries the structure of a group, with the group operation defined by $(g_1 H)
\cdot (g_2 H) = (g_1 \cdot g_2) H$. In this case, we call $G/H$ a \emph{quotient group}.

Given an element $g \in G$ we write
\[
    \langle g \rangle = \{ g^{a} : a \in \mathbb{Z} \} \subseteq G
\]
for the \emph{subgroup generated by $g$} (recall that for a negative integer $a$ we define
$g^{a}$ as the $|a|$-th power of the inverse element of $g$). If there is a positive
integer $a$ such that $g^a$ equals the neutral element of $G$, we define the \emph{order}
$\mathrm{ord}_G(g)$ of $g$ as the smallest such positive integer $a$ (and set
$\mathrm{ord}_G(g)= \infty$ otherwise). If $g$ has finite order, the subgroup $\langle g
\rangle$ of $G$ generated by $g$ is isomorphic to the cyclic group
$\mathbb{Z}_{\mathrm{ord}_G(g)}$.

Given a finite group $A$ and a set $S$ of generators of $A$, the \emph{Cayley graph} of
$A$ and $S$, denoted by $\Gamma(A,S)$ has as vertices the elements of $A$, and two
vertices $u$ and $v$ are adjacent\footnote{Note that, in some literature, Cayley
graphs are coloured and directed. However, we only consider the underlying uncoloured and
simple graph.} if there is an $s\in S$ such that $v=us$. For example, the Cayley graph
$\Gamma(\mathbb{Z}_\ell,\{1, -1\})$ is isomorphic to the cycle of length $\ell$.

\paragraph*{M\"obius Inversion and the Partition Lattice}
We follow the notation of the standard textbook of Stanley~\cite{Stanley11}.
Given a finite partially ordered set $(L,\leq)$, and a function $f: L \rightarrow
\mathbb{Q}$, the \emph{zeta transformation} $\zeta f: L\rightarrow \mathbb{Q}$ is defined
as
\[ \zeta f(\sigma) := \sum_{\rho \geq \sigma} f(\rho) \,. \]

The principle of M\"obius inversion allows us to invert a zeta transformation; a proof of
the following theorem can be found in~\cite[Chapt.~3]{Stanley11}
\begin{theorem}\label{thm:mobiusinv}
    Given a partially ordered set $(L,\leq)$, there is a computable function $\mu:L
    \times L \rightarrow \mathbb{Z}$, called the M\"obius function, such that for all $f:
    L \rightarrow \mathbb{Q}$ and $\sigma\in L$ we have
    \[f(\sigma) =\sum_{\rho \geq \sigma}\mu(\sigma,\rho)\cdot \zeta f (\rho) \,.
    \lipicsEnd
\]
\end{theorem}

\noindent We use M\"obius inversion on the ordering of partitions of finite sets. Given two
partitions $\sigma$ and $\rho$ of a finite set $S$, we say that $\sigma$ \emph{refines}
$\rho$ if every block of $\sigma$ is a subset of a block of $\rho$, and in this case we
write $\sigma \leq \rho$. This induces a partial order, called the \emph{partition
lattice}\footnote{A lattice is a partially ordered set in which every pair of elements has
a least upper bound and a greatest lower bound. For a formal definition we refer
to~\cite{Stanley11}.} of $S$. The explicit formula of the M\"obius function over
the partition lattice is of particular importance in this work:

\begin{theorem}[see~Chapter 3 in~\cite{Stanley11}]\label{thm:mobiusexpl}
    Let $\sigma$ denote a partition of a finite set $S$. We have
    \[\mu(\sigma,\top) = (-1)^{|\sigma|-1} \cdot (|\sigma|-1)!\,,\]
    where $\top=\{S\}$ denotes the coarsest partition.
    \lipicsEnd
\end{theorem}

\paragraph*{Fractured graphs}
\begin{defn}[Fractures]\label{def:fracGeneral}
    Let $H$ denote a graph. A \emph{fracture} of $H$ is a tuple \[\rho=(\rho_v)_{v\in
    V(H)}\,,\] where $\rho_v$ is a partition of the set of edges $E_H(v)$ of $H$ incident to $v$.
    \lipicsEnd
\end{defn}

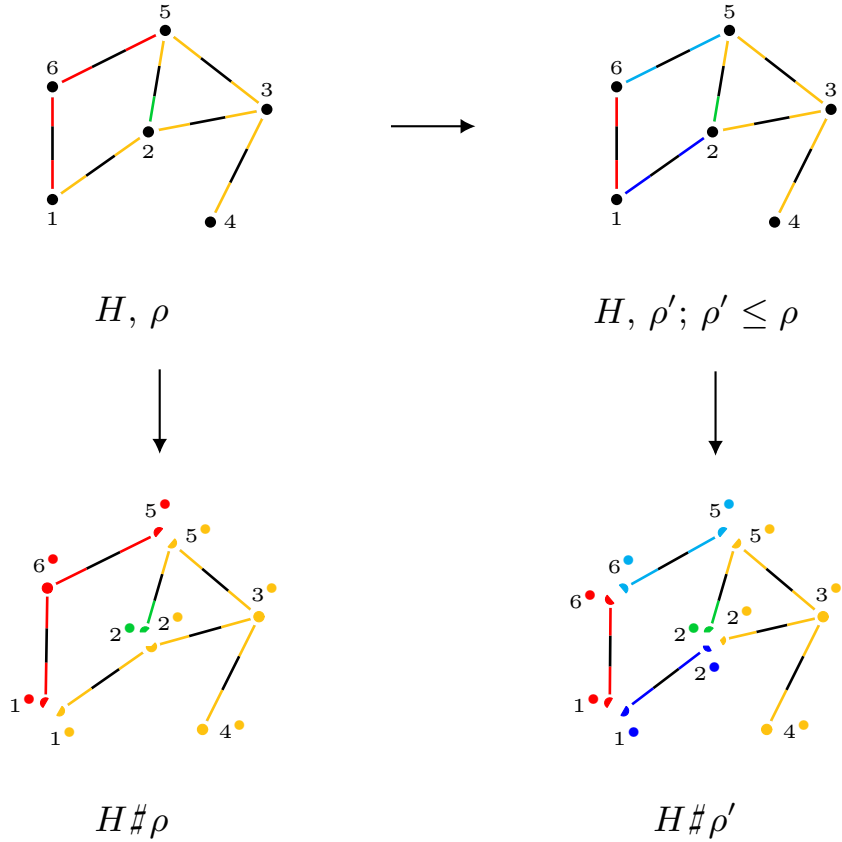
\begin{figure}[t]
    \centering
    \begin{tikzpicture}[scale=1.5,transform shape]
        \begin{scope}[local bounding box=scope1]
            \node[vertex,label={6}] (6) at (0, 0) {};
            \node[vertex,label=5] (5) at (1, .5) {};
            \node[vertex,label=below:1] (1) at (0, -1) {};
            \node[vertex,label=below:2] (2) at (.85, -.4) {};
            \node[vertex,label=above:3] (3) at (1.9, -.2) {};
            \node[vertex,label=right:4] (4) at (1.4, -1.2) {};

            \coledge(6|red)--(5|red);
            \coledge(5|yellow!50!orange)--(2|green!80!blue);
            \coledge(5|yellow!50!orange)--(3|yellow!50!orange);
            \coledge(2|yellow!50!orange)--(3|yellow!50!orange);
            \coledge(4|yellow!50!orange)--(3|yellow!50!orange);
            \coledge(2|yellow!50!orange)--(1|yellow!50!orange);
            \coledge(6|red)--(1|red);

            \node at (.7,-2) {\small $H$, $\rho$};
        \end{scope}
        \draw[-{Latex[length=2mm, width=2mm]},thick] (3,-.35) -- (3.75,-.35);
        \draw[-{Latex[length=2mm, width=2mm]},thick] ($(scope1.south)+(0,-.25)$) -- ++(0,-.75);
        \begin{scope}[shift={(5,0)}, local bounding box=scope2]
            \node[vertex,label={6}] (6) at (0, 0) {};
            \node[vertex,label=5] (5) at (1, .5) {};
            \node[vertex,label=below:1] (1) at (0, -1) {};
            \node[vertex,label=below:2] (2) at (.85, -.4) {};
            \node[vertex,label=above:3] (3) at (1.9, -.2) {};
            \node[vertex,label=right:4] (4) at (1.4, -1.2) {};

            \coledge(6|cyan)--(5|cyan);
            \coledge(5|yellow!50!orange)--(2|green!80!blue);
            \coledge(5|yellow!50!orange)--(3|yellow!50!orange);
            \coledge(2|yellow!50!orange)--(3|yellow!50!orange);
            \coledge(4|yellow!50!orange)--(3|yellow!50!orange);
            \coledge(2|blue)--(1|blue);
            \coledge(6|red)--(1|red);

            \node at (.7,-2) {\small$H$, $\rho'$; $\rho' \le \rho$};
        \end{scope}
        \draw[-{Latex[length=2mm, width=2mm]},thick] ($(scope2.south)+(0,-.25)$) -- ++(0,-.75);
        \begin{scope}[shift={(0,-4.5)}]
            \kowaen{0,0}{-50/310/red}{6};
            \kowaen{1,.5}{-130/50/yellow!50!orange,50/230/red}{5};
            \kowaen{0,-1}{-120/60/yellow!50!orange,60/240/red}{1};
            \kowaen{.85,-.4}{-180/45/yellow!50!orange,45/180/green!80!blue}{2};
            \kowaen{1.9,-.2}{0/360/yellow!50!orange}{3};
            \kowaen{1.4,-1.2}{0/360/yellow!50!orange}{4};

            \coledge(6-1|red)--(5-2|red);
            \coledge(5-1|yellow!50!orange)--(2-2|green!80!blue);
            \coledge(5-1|yellow!50!orange)--(3-1|yellow!50!orange);
            \coledge(2-1|yellow!50!orange)--(3-1|yellow!50!orange);
            \coledge(4-1|yellow!50!orange)--(3-1|yellow!50!orange);
            \coledge(2-1|yellow!50!orange)--(1-1|yellow!50!orange);
            \coledge(6-1|red)--(1-2|red);

            \node[label={$6^{\color{red}\bullet}$}] at (6-1) {};
            \node[label={above:$5^{\color{red}\bullet}$}] at (5-2) {};
            \node[label=right:{$~5^{\color{yellow!50!orange}\bullet}\!$}] at (5-2) {};
            \node[label=above:{$~3^{\color{yellow!50!orange}\bullet}\!$}] at (3-1) {};
            \node[label=right:{$4^{\color{yellow!50!orange}\bullet}\!$}] at (4-1) {};
            \node[label={$\quad~2^{\color{yellow!50!orange}\bullet}\!$},
            label distance=-10] at (2-1) {};
            \node[label=left:{$\!2^{\color{green!80!blue}\bullet}\!\!$}] at (2-2) {};
            \node[label=left:{$1^{\color{red}\bullet}$\!\!}] at (1-2) {};
            \node[label=below:{$1^{\color{yellow!50!orange}\bullet}$\!\!}] at (1-1) {};

            \node at (.7,-2) {\small $\fracture{H}{\rho}$};
        \end{scope}
        \begin{scope}[shift={(5,-4.5)}]
            \kowaen{0,0}{-50/130/cyan,130/310/red}{6};
            \kowaen{1,.5}{-130/50/yellow!50!orange,50/230/cyan}{5};
            \kowaen{0,-1}{-120/60/blue,60/240/red}{1};
            \kowaen{.85,-.4}{-60/45/yellow!50!orange,45/180/green!80!blue,180/300/blue}{2};
            \kowaen{1.9,-.2}{0/360/yellow!50!orange}{3};
            \kowaen{1.4,-1.2}{0/360/yellow!50!orange}{4};

            \coledge(6-1|cyan)--(5-2|cyan);
            \coledge(5-1|yellow!50!orange)--(2-2|green!80!blue);
            \coledge(5-1|yellow!50!orange)--(3-1|yellow!50!orange);
            \coledge(2-1|yellow!50!orange)--(3-1|yellow!50!orange);
            \coledge(4-1|yellow!50!orange)--(3-1|yellow!50!orange);
            \coledge(2-3|blue)--(1-1|blue);
            \coledge(6-2|red)--(1-2|red);

            \node[label={$6^{\color{cyan}\bullet}$}] at (6-1) {};
            \node[label=left:{$6^{\color{red}\bullet}\!$}] at (6-2) {};
            \node[label={above:$5^{\color{cyan}\bullet}$}] at (5-2) {};
            \node[label=right:{$~5^{\color{yellow!50!orange}\bullet}\!$}] at (5-2) {};
            \node[label=above:{$~3^{\color{yellow!50!orange}\bullet}\!$}] at (3-1) {};
            \node[label=right:{$4^{\color{yellow!50!orange}\bullet}\!$}] at (4-1) {};
            \node[label={$\quad~2^{\color{yellow!50!orange}\bullet}\!$},
            label distance=-10] at (2-1) {};
            \node[label=below:{$\!2^{\color{blue}\bullet}\!$}] at (2-3) {};
            \node[label=left:{$\!2^{\color{green!80!blue}\bullet}\!\!$}] at (2-2) {};
            \node[label=left:{$1^{\color{red}\bullet}$\!\!}] at (1-2) {};
            \node[label=below:{$1^{\color{blue}\bullet}$\!\!}] at (1-1) {};

            \node at (.7,-2) {\small $\fracture{H}{\rho'}$};
        \end{scope}
    \end{tikzpicture}
    \caption{Two fractures $\rho$ and $\rho' \le \rho$ of a graph $H$
        (depicted by the colours of the outgoing edges of the vertices)
        and the corresponding fractured graphs $\fracture{H}{\rho}$ and
        $\fracture{H}{\rho'}$.
    }\label{fig:fracture}
\end{figure}

Note that the set of all fractures of $H$, denoted by $\mathcal{L}(H)$,
is a lattice that is isomorphic to the (point-wise) product of the partition lattices of $E_H(v)$
for each $v\in V(H)$. In particular, we write $\sigma \leq \rho$ if, for each $v\in V(H)$,
the partition $\sigma_v$ refines the partition $\rho_v$.
Consider \cref{fig:fracture} for a visualization of a fracture and its refinement.

Note further that a fracture describes how to split (or \emph{fracture}) each vertex of a
given graph: for each vertex~$v$, create a vertex $v^B$ for each block $B$ in the partition
$\rho_v$; edges originally incident to $v$ are made incident to $v^B$ if and only if they
are contained in $B$. We call the resulting graph the \emph{fractured graph $\fracture{H}{\rho}$};
a formal definition is given in \cref{def:mrhoGeneral}, a visualization is given in
\cref{fig:fracture}.

\begin{defn}[Fractured Graph $\fracture{H}{\rho}$]\label{def:mrhoGeneral}
    Given a graph $H$, consider the matching $M_H$ containing one edge for each edge of $H$;
    formally,
    \[
        V(M_H) := \bigcup_{e = \{u,v\} \in E(H)} \{ u_e, v_e \}\quad\text{and}\quad
        E(M_H) := \{ \{ u_e, v_e \} \mid e = \{u,v\} \in E(H)\}.
    \]

    For a fracture $\rho$ of $H$, we define\footnote{The notation $\fracture{H}{\rho}$
    stems from the fact that the symbol ``$\fracture$'' is commonly used for medical
    fractures.} the graph $\fracture{H}{\rho}$ as the quotient
    graph of $M_H$ under the equivalence relation on $V(M_H)$ which identifies two
    vertices $v_e, w_f$ of $M_H$ if and only if $v=w$ and $e,f$ are in the same block $B$
    of the partition $\rho_v$ of $E_H(v)$. We write $v^B$ for the vertex of
    $\fracture{H}{\rho}$ given by the equivalence class of the vertices $v_e$ (for which
    $e \in B$) of $M_H$.
    \lipicsEnd
\end{defn}

The fractured graph $\fracture{H}{\rho}$ comes with a natural $H$-colouring. Indeed,
the homomorphism $M_H \to H$ which sends $v_e \in V(M_H)$ to $v \in V(H)$ descends to
$\fracture{H}{\rho}$ so that we always have a canonical diagram
$M_H \to \fracture{H}{\rho} \to H$ of graph homomorphisms.
For example, for any graph $H$, the fracture $\bot$, with $\bot_v := \{ \{e\}\mid e \in E_H(v)
\}$, induces the fractured graph $\fracture{H}{\bot} = M_H$;
the fracture $\top$, with $\top_v := \{ E_H(v) \}$ induces the fractured graph
$\fracture{H}{\top} = H$. The fractures $\bot, \top$ are the minimal and maximal elements
of the lattice $\mathcal{L}(H)$, respectively.

Given a graph property $\Phi$ and a graph $H$, we write $\mathcal{L}(\Phi,H)$ for the set
of all fractures $\rho$ of $H$ that satisfy $\Phi(\fracture{H}{\rho})=1$.\pagebreak

\section{The Colour-Preserving Homomorphism Basis}\label{sec:CountingWithColours}
It turns out that the analysis of the complexity of $\#\edgesubsprob(\Phi)$ is much easier
if a \emph{colourful} version of the problem is considered. For our hardness results, we
then show that the colourful version reduces to the uncoloured version. To this
end, recall that an $H$-colouring of a graph $G$ is a homomorphism from~$G$ to~$H$, and
that a graph $G$ is $H$\emph{-coloured} if $G$ is equipped with an $H$-colouring $c$.
Recall further the implicitly defined $\#E(H)$-edge-colouring of $G$.
In the colourful version of $\#\edgesubsprob(\Phi)$, denoted by $\#\coledgesubsprob(\Phi)$,
the task is to compute the cardinality of the set
\[\ncoledgesubs{\Phi}{H}{G} :=\{ A\subseteq E(G)~|~ \#A=\#E(H) \wedge c(A)=E(H) \wedge \Phi(G[A])=1 \} \,.\]
In particular, we write $\#\ncoledgesubs{\Phi}{H}{\star}$ for the function that maps an
$H$-coloured graph $G$ to the number $\#\ncoledgesubs{\Phi}{H}{G}$.
Note that $\#A=\#E(H) \wedge c(A)=E(H)$ implies that the $A$ contains each of the $\#E(H)$
colours precisely once. Further, note that $\Phi(G[A])=1$ if and only if $\Phi$ holds on
the (uncoloured) graph $G[A]$.

Each element $A \in \ncoledgesubs{\Phi}{H}{G}$ induces a fracture $\rho$ of $H$, where for
$v \in V(G)$ two edges $e,f \in E_H(v)$ are in the same block of $\rho_v$ if and only if
their (unique) preimages $\widehat e, \widehat f \in A$ under $c : A \to E(H)$ are
connected to the \emph{same} endpoint in $V_v$. From the construction, it immediately
follows that $G[A]$ and $\fracture{H}{\rho}$ are canonically isomorphic as $H$-coloured
graphs, that is, $G[A]\cong_H\fracture{H}{\rho}$.

Our goal is to express $\#\ncoledgesubs{\Phi}{H}{\star}$ as a linear combination of
(colour-preserving) homomorphism counts from graphs only depending on $\Phi$ and $H$. In
case $H$ is a torus, we establish an explicit criterion sufficient for the term
$\#\necphoms{\torus_\ell}{\torus_\ell}{\star}$ to survive with a non-zero coefficient in
this linear combination.
The existence of the linear combination is given by the following lemma:

\begin{lemma}\label{lem:col_lincomb}
    Let $H$ denote a graph. We have
    \[\#\ncoledgesubs{\Phi}{H}{\star} = \sum_{\sigma\in\mathcal{L}(\Phi,H) } ~~~\sum_{{\rho} \geq {\sigma}} \mu({\sigma},{\rho}) \cdot \#\necphoms{\fracture{H}{\rho}}{H}{\star} \,,\]
    where the relation $\leq$ and the M\"obius function $\mu$ are over the lattice of fractures $\mathcal{L}(H)$.
\end{lemma}
\begin{proof}
    Let $G$ denote an $H$-coloured graph.
    We first partition the elements $A$ of the set $\ncoledgesubs{\Phi}{H}{G}$ according to their
    induced fractures. Writing $\ncoledgesubs{\Phi}{H}{G}[{\sigma}]$ for the set of $A$
    inducing the fracture $\sigma \in \mathcal{L}(H)$, we obtain
    \[\#\ncoledgesubs{\Phi}{H}{G} = \sum_{\sigma\in\mathcal{L}(\Phi,H) } \#\ncoledgesubs{\Phi}{H}{G}[{\sigma}]\,, \]
    since $\#\ncoledgesubs{\Phi}{H}{G}[{\sigma}]=0$ for all $\sigma \notin \mathcal{L}(\Phi,H)$.
    From the fact that $G[A]$ is canonically isomorphic to $\fracture{H}{\sigma}$ as an
    $H$-coloured graph, for $\sigma$ associated to $A \in \ncoledgesubs{\Phi}{H}{G}$, it
    follows that
    \[\#\ncoledgesubs{\Phi}{H}{G}[{\sigma}] = \#\necpembs{\fracture{H}{\sigma}}{H}{G}\,.\]
    Note that we are using that graphs of the form $\fracture{H}{\sigma}$ can have no
    non-trivial automorphisms as $H$-coloured graphs (since all edges must be fixed).
    It remains to show that
    \begin{equation}\label{eq:colorful_mobius_new}
        \#\necpembs{\fracture{H}{\sigma}}{H}{G} = \sum_{{\rho} \geq {\sigma}} \mu({\sigma},{\rho}) \cdot \#\necphoms{\fracture{H}{\rho}}{H}{G}
    \end{equation}
    To this end, we establish the following zeta transformation, which should be considered
    as a colour-preserving version of the standard transformation in the uncoloured case
    (see e.g.\ \cite[Section 5.2.3]{Lovasz12}).
    \begin{claim}
        For every fracture $\sigma$ of $H$, we have
        \[\#\necphoms{\fracture{H}{\sigma}}{H}{G} = \sum_{{\rho} \geq {\sigma}} \#\necpembs{\fracture{H}{\rho}}{H}{G}   \]
    \end{claim}
    \begin{claimproof}
        Every colour-preserving homomorphism $\varphi$ from $\fracture{H}{\sigma}$ to $G$,
        induces a fracture $\rho\geq \sigma$, that is, $\rho_{v}$ is a coarsening of
        $\sigma_{v}$ for every $v\in V(H)$.
        Indeed, recall that the vertices of $\fracture{H}{\sigma}$ over $v \in V(H)$
        correspond to the blocks $B$ of the partition $\sigma_v$ of the edges $E_H(v)$.
        Then the partition $\rho_v$ of $E_H(v)$ is obtained from $\sigma_v$ by joining
        those blocks $B,B'$ whose associated vertices in $\fracture{H}{\sigma}$ map to the
        same vertex of $G$ under $\varphi$. We have that the subgraph of $G$ given by the
        image of $\fracture{H}{\sigma}$ under $\varphi$ is canonically isomorphic to
        $\fracture{H}{\rho}$ as an $H$-coloured graph.

        Let us call two homomorphisms in $\necphoms{\fracture{H}{\sigma}}{H}{G}$
        equivalent if they induce the same fracture and write
        $\necphoms{\fracture{H}{\sigma}}{H}{G}[\rho]$ for the equivalence class of all
        homomorphisms inducing $\rho$.
        The claim then follows by partitioning the set
        $\necphoms{\fracture{H}{\sigma}}{H}{G}$ into those equivalence classes and
        observing that
        \[ \#\necphoms{\fracture{H}{\sigma}}{H}{G}[\rho] =
        \#\necpembs{\fracture{H}{\rho}}{H}{G}\,.\]
    \end{claimproof}
    Equation~\eqref{eq:colorful_mobius_new} is now obtained by using M\"obius inversion
    (\cref{thm:mobiusinv}) on the zeta-transformation given by the previous claim. This
    concludes the proof.
\end{proof}

Let us now collect for the coefficient of the term
$\#\necphoms{\fracture{H}{\top}}{H}{G}$, where $\top$ is the maximum fracture of $H$ with
respect to the ordering $\leq$. In particular, each partition of $\top$ only consists of a
single block and thus $\fracture{H}{\top} \cong H$, where the isomorphism is given by the
$H$-colouring of $\fracture{H}{\top}$.
\begin{corollary}\label{cor:collect_coeffs}
    Let $\Phi$ denote a computable graph property and let $H$ denote a graph.
    There is a unique computable function $a_{\Phi,H}: \mathcal{L}(H) \rightarrow
    \mathbb{Z}$ such that
    \[\#\ncoledgesubs{\Phi}{H}{\star} = \sum_{\rho\in \mathcal{L}(H)} a_{\Phi,H}(\rho)\cdot \#\necphoms{\fracture{H}{\rho}}{H}{\star} \,,\]
    For $\rho=\top$ we have
    \[ a_{\Phi,H}(\top) = \sum_{\sigma \in \mathcal{L}(\Phi,H)} ~\prod_{v\in V(H)} (-1)^{|{\sigma}_{v}|-1} \cdot (|{\sigma}_{v}|-1)! \]
    Here, $|{\sigma}_{v}|$ denotes the number of blocks of partition ${\sigma}_{v}$.
\end{corollary}
\begin{proof}
    The first claim follows immediately from \cref{lem:col_lincomb} by collecting
    coefficients; note that $\Phi$ and $\mu$ are computable, and that the image of $\mu$
    is a subset of $\mathbb{Z}$.
    For the second claim, we collect the coefficients of
    $\#\necphoms{\fracture{H}{\top}}{H}{\ast}$ in \cref{lem:col_lincomb} and obtain
    \[ a_{\Phi,H}(\top) = \sum_{\sigma \in \mathcal{L}(\Phi,H)} \mu({\sigma},\top) \,.\]
    Recall that $\mu$ is the M\"obius function of $\mathcal{L}(H)$ and that the latter is
    the product of the partition lattices of $N_H(v)$ for each $v\in V(H)$. Using that the
    M\"obius function is multiplicative with respect to the product (see for instance
    \cite[Proposition 3.8.2]{Stanley11}) and applying the explicit formula for the
    partition lattice (\cref{thm:mobiusexpl}), we obtain the second claim.
\end{proof}

In the remainder of the paper, given $\Phi$ and $H$, we refer to the function $a_{\Phi,H}$
from \cref{cor:collect_coeffs} as the \emph{\cfunction} of $\Phi$ and $H$.

\subsection{Complexity Monotonicity for Colour-Preserving Homomorphisms}

Our next goal is to prove that computing a finite linear combination of colour-preserving
homomorphism counts, as given by \cref{cor:collect_coeffs}, is precisely as hard as
computing its hardest term.
While the proof strategy follows the approach used in~\cite{CurticapeanDM17}
and~\cite{DorflerRSW19}, we need to adapt to the setting of colour-preserving
homomorphisms between fractured graphs.

We rely on the tensor product of $H$-coloured graphs in the following way: let $H$ denote a
fixed graph, and let $G$ and $F$ denote $H$-coloured graphs with colourings $c_G$ and $c_F$.
The \emph{colour-preserving tensor product}~$G\times_H F$ has vertices $\{(x,y)\in
V(G)\times V(F)~|~ c_G(x) = c_F(y)\}$, and two vertices $(x,y)$ and $(x',y')$ are made
adjacent in $G\times_H F$ if (and only if) $\{x,x'\}\in E(G)$ and $\{y,y'\}\in E(F)$.
Observe that the graph $G\times_H F$ is $H$-coloured as well by the colouring $(x,y)
\mapsto c_G(x) = c_F(y)$.

\begin{lemma}\label{lem:tensorlinear}
    Let $H$ denote a graph, and let $F$, $G_1$ and $G_2$ denote $H$-coloured graphs. We have
    \[\#\necphoms{F}{H}{G_1 \times_H G_2} = \#\necphoms{F}{H}{G_1} \cdot \#\necphoms{F}{H}{G_2} \,.\]
\end{lemma}
\begin{proof}
    The function
    \begin{align*}
        &b: \necphoms{F}{H}{G_1} \times \necphoms{F}{H}{G_2}
        \to \necphoms{F}{H}{G_1 \times_H G_2},\\
        &b(\varphi,\psi)(u):= (\varphi(u),\psi(u)) \text{ for } u\in V(F)
    \end{align*}
    is the canonical bijection.
\end{proof}

\noindent The reduction for isolating terms with non-zero coefficient requires to solve a system of
linear equations. For the definition of the corresponding matrix, we fix a linear
extension $\preccurlyeq$ of the order $\leq$ of the $H$-fractures.
Recall that $\leq$ is also the order of the product of the partition lattices of the set
$E(v)$ for all $v\in V(H)$. In particular, $\sigma \leq \rho$ if and only if $\sigma_v$
refines $\rho_v$ for all $v\in V(H)$.
As a consequence, we observe that $\rho \succ \sigma$, that is, $\neg (\rho \preccurlyeq
\sigma)$, implies the existence of a vertex $v\in V(H)$ such that $\rho_v$ does not refine
$\sigma_v$.
Now let $\mathcal{M}_H$ denote the matrix whose columns and rows are associated with the
set of all $H$-fractures, ordered by $\preccurlyeq$, and whose entries are given by
\[ \mathcal{M}_H[\rho,\sigma] := \#\ecphoms{\fracture{H}{\rho}}{\fracture{H}{\sigma}} \,.\]

\begin{lemma}\label{lem:fullrank}
    For each $H$, the matrix $\mathcal{M}_H$ is upper triangular with entries $1$ on the diagonal.
\end{lemma}
\begin{proof}
    Let us first consider the diagonal. We claim that
    $\#\ecphoms{\fracture{H}{\rho}}{\fracture{H}{\rho}}=1$. Due to the trivial (identity)
    homomorphism we have $\#\ecphoms{\fracture{H}{\rho}}{\fracture{H}{\rho}} \geq 1$. On
    the other hand, the canonical colouring $\fracture{H}{\rho} \to H$ induces a bijection
    from the edges of $\fracture{H}{\rho}$ to the edges of~$H$ that preserves the
    colouring. Since a colour-preserving homomorphism $\fracture{H}{\rho} \to
    \fracture{H}{\rho}$ must commute with this map, it must act as the identity on all
    edges of $\fracture{H}{\rho}$ and is thus equal to the identity.

    It remains to prove that $\mathcal{M}_H[\rho,\sigma]=0$ for every $\rho \succ \sigma$.
    Recall that the latter implies the existence of a vertex $v\in V(H)$ such that
    $\rho_{v}$ does not refine $\sigma_{v}$, that is, there is a block~$B$ of~$\rho_{v}$
    which is not a subset of any block of $\sigma_{v}$.
    Thus, there are edges $e, f \in B \subseteq E_H(v)$ such that $e,f$ are in different
    blocks of $\sigma_v$. Identifying $E(\fracture{H}{\sigma}) = E(\fracture{H}{\rho}) =
    E(H)$ using the colouring, we see that $e,f$ are adjacent to the \emph{same} vertex in
    $\fracture{H}{\rho}$ (corresponding to the block $B$), but to \emph{different}
    vertices in $\fracture{H}{\sigma}$. This implies that there cannot be a colour
    preserving homomorphism $\varphi: \fracture{H}{\rho} \to \fracture{H}{\sigma}$ since
    $e,f$ being incident at $v^B$ in $\fracture{H}{\rho}$ would imply that
    $e=\varphi(e),f=\varphi(f)$ must be incident at $\varphi(v^B)$ in
    $\fracture{H}{\sigma}$.
\end{proof}

\noindent We are now able to prove a version of the Complexity Monotonicity principle which is
sufficient for the purposes in this work. In what follows, given a graph property $\Phi$,
we write $\#\ncoledgesubs{\Phi}{\star}{\star}$ for the function that expects as input a
graph $H$ and an $H$-coloured graph $G$, and outputs $\#\ncoledgesubs{\Phi}{H}{G}$.

\begin{lemma}\label{lem:newcomplexitymonotonicity}
    Let $\Phi$ denote a computable graph property. There exists a deterministic algorithm~$\mathbb{A}$ which has oracle access to the function
    $\#\ncoledgesubs{\Phi}{\star}{\star}$,
    and computes, on input a graph $H$ and an $H$-coloured graph $G$, the numbers
    $\#\necphoms{\fracture{H}{\rho}}{H}{G}$ for every $H$-fracture~$\rho$ satisfying that
    $a_{\Phi,H}(\rho)\neq 0$, where $a_{\Phi,H}$ is the \cfunction{} of $\Phi$ and $H$.

    Furthermore, there is a computable function $f$ such that $\mathbb{A}$ runs in time
    $f(|H|)\cdot |G|^{O(1)}$ and every posed oracle query $(\hat{H},\hat{G})$ satisfies
    $\hat{H}=H$ and $|\hat{G}|\leq f(|H|)\cdot |G|$.
\end{lemma}
\begin{proof}
    Given $H$ and $G$, we can obtain the numbers
    $\#\ncoledgesubs{\Phi}{H}{G \times_H (\fracture{H}{\sigma})}$ for all
    $H$-fractures~$\sigma$ via access to the oracle. By
    \cref{cor:collect_coeffs,lem:tensorlinear}, we have
    \begin{align*}
        &\#\ncoledgesubs{\Phi}{H}{G \times_H (\fracture{H}{\sigma})}\\
        &\qquad = \sum_{\rho\in \mathcal{L}(H)} a_{\Phi,H}(\rho)\cdot \#\necphoms{\fracture{H}{\rho}}{H}{G \times_H (\fracture{H}{\sigma})}\\
        &\qquad = \sum_{\rho\in \mathcal{L}(H)} a_{\Phi,H}(\rho)\cdot \#\necphoms{\fracture{H}{\rho}}{H}{G} \cdot \#\necphoms{\fracture{H}{\rho}}{H}{\fracture{H}{\sigma}}
    \end{align*}
    \noindent Observe that the latter yields a system of linear equations for the numbers
    \[a_{\Phi,H}(\rho)\cdot \#\necphoms{\fracture{H}{\rho}}{H}{G}\] with matrix
    $\mathcal{M}_H$ which has full rank according to \cref{lem:fullrank}. Consequently
    $\mathbb{A}$ can compute the number
    $a_{\Phi,H}(\rho)\cdot \#\necphoms{\fracture{H}{\rho}}{H}{G}$
    for each $H$-fracture $\rho$. Therefore, $\#\necphoms{\fracture{H}{\rho}}{H}{G}$
    can be computed whenever $a_{\Phi,H}(\rho)\neq 0$.

    Now observe that $a_{\Phi,H}$, which is computable, only depends on $\Phi$, which is
    fixed, and $H$. Furthermore $\mathcal{L}(H)$ and all $\fracture{H}{\rho}$ only depend
    on $H$. Thus the computation of
    $\#\necphoms{\fracture{H}{\rho}}{H}{\fracture{H}{\sigma}}$, takes time only depending
    on $H$ as well.
    Consequently, the system of linear equations can be solved in time $f'(|H|)\cdot
    |G|^{O(1)}$ for some computable function $f'$. Furthermore, the size of $G\times_H
    (\fracture{H}{\sigma})$ is bounded by $|\fracture{H}{\sigma}|\cdot |G|$. Setting
    $f(|H|):= \max\{f'(|H|),\max_{\sigma \in \mathcal{L}(H)} |\fracture{H}{\sigma}| \}$
    concludes the proof; since each fractured graph $\fracture{H}{\sigma}$ has only~$\#E(H)$ many edges.
\end{proof}

\subsection{Intractability of Counting Homomorphisms from Tori and Expanders}
The final step of this section is to establish $\#\W{1}$-hardness of the (uncoloured)
problem $\#\edgesubsprob(\Phi)$ whenever $a_{\Phi,\torus_\ell}(\top)\neq 0$ for infinitely
many $\ell$. Essentially, we rely on the fact that tori have high treewidth and that
the problem of counting (colour-preserving) homomorphisms from high-treewidth graphs is
hard. We can proceed similarly in case of expanders, and due to the fact that expanders
have high treewidth and are sparse (see \cref{fac:niceExpanders}), we even obtain an
almost tight conditional lower bound.

In both cases, we use Complexity Monotonicity, which yields hardness of the
(edge-) colourful version of $\#\edgesubsprob(\Phi)$. Consequently, we need to show that
the colourful version reduces to the uncoloured version.
This can be achieved by a standard inclusion-exclusion argument:

\begin{lemma}\label{lem:nocolours}
    Let $\Phi$ denote a computable graph property. There exists a deterministic algorithm
    $\mathbb{A}$, equipped with oracle access to the function
    \[ (k,\hat{G})\mapsto \#\edgesubs{\Phi,k}{\hat{G}}\,, \]
    which expects as input a graph $H$ and an $H$-coloured graph $G$, and computes in time $2^{|E(H)|}\cdot |G|^{O(1)}$ the cardinality
    $\#\ncoledgesubs{\Phi}{H}{G}$. Furthermore, every
    oracle query $(k,\hat{G})$ posed by $\mathbb{A}$ satisfies $k=|E(H)|$ and
    $|\hat{G}|\leq|G|$.
\end{lemma}
\begin{proof}
    Given $H$ and an $H$-coloured graph $G$, we write $c:E(G)\rightarrow E(H)$ for the
    induced edge-colouring of $G$. Given a set of edge-colours $J\subseteq E(H)$, we write
    $G\setminus J$ for the graph obtained from $G$ by deleting all edges $e$ with $c(e)\in
    J$. Now recall that
    \[\edgesubs{\Phi,|E(H)|}{G} =\{ A\subseteq E(G)~|~ \#A=|E(H)| \wedge \Phi(G[A])=1 \} \,,\text{ and}\]
    \[\ncoledgesubs{\Phi}{H}{G} =\{ A\subseteq E(G)~|~ \#A=|E(H)| \wedge c(A)=E(H) \wedge \Phi(G[A])=1 \} \,.\]
    Next set $k:=|E(H)|$, then we have
    \begin{align*}
        &\#\ncoledgesubs{\Phi}{H}{G}\\
         &\quad= \#\edgesubs{\Phi,k}{G}\\
          &\qquad- \#\left(\bigcup_{e\in E(H)} \{A\in \edgesubs{\Phi,k}{G}~|~e \notin c(A) \}\right)\\
        &\quad=\#\edgesubs{\Phi,k}{G}\\
          &\qquad- \sum_{\emptyset \neq J \subseteq E(H)}  (-1)^{|J|+1} \cdot \#\left(\bigcap_{e\in J} \{A\in \edgesubs{\Phi,k}{G}~|~e \notin c(A) \} \right)\\
        &\quad=\#\edgesubs{\Phi,k}{G}\\
          &\qquad- \sum_{\emptyset \neq J \subseteq E(H)}  (-1)^{|J|+1} \cdot \#\{A\in \edgesubs{\Phi,k}{G}~|~\forall e\in J: e \notin c(A) \}\\
        &\quad=\#\edgesubs{\Phi,k}{G}\\
          &\qquad- \sum_{\emptyset \neq J \subseteq E(H)}  (-1)^{|J|+1} \cdot \#\edgesubs{\Phi,k}{G\setminus J}\\
          &\quad=\sum_{J \subseteq E(H)}  (-1)^{|J|} \cdot \#\edgesubs{\Phi,k}{G\setminus J}
    \end{align*}

    \noindent Note that the second equation is due to the inclusion-exclusion principle. We conclude
    that the desired number $\#\ncoledgesubs{\Phi}{H}{G}$ can be computed using
    $2^{|E(H)|}$ many oracle calls of the form $\#\edgesubs{\Phi,|E(H)|}{G\setminus J}$.
    The claim of the lemma follows since $|G\setminus J|\leq |G|$.
\end{proof}

For the formal statement of this section's main lemma, we define
$\mathcal{H}[\Phi,\torus]$ as the set of all $\torus_\ell$ such that
$a_{\Phi,\torus_\ell}(\top)\neq 0$. Furthermore, given a family
$\mathcal{G}=\{G_1,G_2,\dots\}$ of $(n_i,d,c)$-expanders, we write
$\mathcal{H}[\Phi,\mathcal{G}]$ for the set of all $G_i$ such that $a_{\Phi,G_i}(\top)\neq
0$

\begin{lemma}\label{lem:hardness_basis}
    Let $\Phi$ denote a computable graph property,
    fix a rational $c$ and an integer $d$, and
    let $\mathcal{G}=\{G_1,G_2,\dots\}$ denote a family of $(n_i,d,c)$-expanders.
    If at least one of $\mathcal{H}[\Phi,\torus]$ and $\mathcal{H}[\Phi,\mathcal{G}]$ is
    infinite, then $\#\edgesubsprob(\Phi)$ is $\#\W{1}$-hard. Moreover, if
    $\mathcal{H}[\Phi,\mathcal{G}]$ is infinite, then $\#\edgesubsprob(\Phi)$ cannot be
    solved in time
    $f(k)\cdot |G|^{o(k/\log k)}$ for any function $k$, unless the ETH fails.
\end{lemma}
\begin{proof}
    We start with the case of $\mathcal{H}[\Phi,\torus]$ being infinite. If the latter is
    true, then $\mathcal{H}[\Phi,\torus]$ has unbounded treewidth since it contains graphs
    with arbitrary large grid minors~\cite{RobertsonS86-ExGrid}.
    This allows us to reduce from the problem $\#\homsprob(\mathcal{H}[\Phi,\torus])$
    which is known to be $\#\W{1}$-hard since $\mathcal{H}[\Phi,\torus]$ has unbounded
    treewidth~\cite{DalmauJ04}.
    It is convenient to consider the following intermediate problem: given a graph $H\in
    \mathcal{H}[\Phi,\torus]$ and an $H$-coloured graph $G$ with colouring~$c$, in the
    problem $\#\cphomsprob(\mathcal{H}[\Phi,\torus])$ the task is to compute the number
    $\#\cphoms{H}{G}$ of homomorphisms $\varphi \in \homs{H}{G}$ such that
    $c(\varphi(v))=v$ for each vertex~$v$ of~$H$. It is well-known that
    $\#\homsprob(\mathcal{H})$ reduces to $\#\cphomsprob(\mathcal{H})$ for every class of
    graphs $\mathcal{H}$; see for instance
    \cite{Roth19,DorflerRSW19,DellRW19icalp,Curticapean15}---note that, in the latter, the
    problem is referred to as
    $\#\textsc{PartitionedSub}(\mathcal{H})$. Thus we have
    \begin{equation}\label{eq_redchain1}
        \#\homsprob(\mathcal{H}[\Phi,\torus])\fptred \#\cphomsprob(\mathcal{H}[\Phi,\torus])\,.
    \end{equation}

    Now observe that $\#\cphoms{H}{G} = \#\necphoms{\fracture{H}{\top}}{H}{G}$ for every
    graph $H$ and $H$-coloured graph $G$, since $\fracture{H}{\top}=H$. By definition of
    $\mathcal{H}[\Phi,\torus]$, we have that $a_{\Phi,\torus_\ell}(\top)\neq 0$ whenever
    $\torus_\ell\in \mathcal{H}[\Phi,\torus]$. Thus we can use Complexity Monotonicity
    (\cref{lem:newcomplexitymonotonicity}) which yields the reduction
    \begin{equation}\label{eq_redchain2}
        \#\cphomsprob(\mathcal{H}[\Phi,\torus])\fptred \#\coledgesubsprob(\Phi)\,.
    \end{equation}
    Finally, we can reduce to the uncoloured version by \cref{lem:nocolours} and obtain
    \begin{equation}\label{eq_redchain3}
        \#\coledgesubsprob(\Phi) \fptred \#\edgesubsprob(\Phi) \,.
    \end{equation}
    Consequently, $\#\edgesubsprob(\Phi)$ is $\#\W{1}$-hard by \eqref{eq_redchain1} -
    \eqref{eq_redchain3} in combination with $\#\W{1}$-hardness of
    $\#\homsprob(\mathcal{H}[\Phi,\torus])$.

     In case of $\mathcal{H}[\Phi,\mathcal{G}]$, we reduce from the homomorphism counting
     problem as well and obtain $\#\W{1}$-hardness analogously. However, for the almost
     tight conditional lower bound, we rely on a result of Marx~\cite{Marx10} implying
     that for any class $\mathcal{H}$ of unbounded treewidth, the problem
     $\#\homsprob(\mathcal{H})$ cannot be solved in time
    \[f(|H|)\cdot |G|^{o(\mathsf{tw}(H)/\log \mathsf{tw}(H))}\]
    for any function $f$, unless ETH fails.\footnote{Observe that this result follows
    only implicitly from~\cite{Marx10}, but we made it explicit in~\cite{RothSW20}.} Let
    us use the aforementioned lower bound for the case of
    $\mathcal{H}=\mathcal{H}[\Phi,\mathcal{G}]$. We observe that the reduction sequence
    from $\#\homsprob(\mathcal{H}[\Phi,\mathcal{G}])$ to $\#\edgesubsprob(\Phi)$ as
    illustrated before only leads to a linear blow up of the parameter: given an input
    $(G_i,G)$ for which we wish to compute $\#\homs{G_i}{G}$, we only query the oracle for
    $\#\edgesubsprob(\Phi)$ on instances $(k,G')$ where $k=\#E(G_i)$ and $|G'|\leq
    f''(\#E(G_i)) \cdot  |G|$ for some function $f''$.
    Since both $\#E(G_i)$ and the treewidth of~$G_i$ are linear in $|V(G_i)|$ (see
    \cref{fac:niceExpanders}), any algorithm that, for some function $f'$, solves
    $\#\edgesubsprob(\Phi)$ in time
    \[f'(k)\cdot |G'|^{o(k/\log k)}\]
    yields an algorithm for $\#\homsprob(\mathcal{H}[\Phi,\mathcal{G}])$ running in time
    \[f(|G_i|)\cdot |G|^{o(\mathsf{tw}(G_i)/\log \mathsf{tw}(G_i))}\,,\]
    for some function $f$ (depending only on $f'$ and $f''$), contradicting ETH by Marx'
    lower bound.
\end{proof}
Regarding the previous proof, observe that we cannot obtain a similar conditional
lower bound if only $\mathcal{H}[\Phi,\torus]$ is infinite, since in that case the
parameter grows quadratically: while~$\torus_\ell$ has treewidth $O(\ell)$, it has
$2\ell^2$ edges.

\section{Coefficients of Tori and Cayley Graph Expanders}\label{sec:main_section}
The previous section allows us to establish hardness of $\#\edgesubsprob(\Phi)$ by the
purely combinatorial problem of determining whether one of the sets
$\mathcal{H}[\Phi,\torus]$ and $\mathcal{H}[\Phi,\mathcal{G}]$, for some family of
expanders $\mathcal{G}$, is infinite. Still, this is a challenging combinatorial problem
and we consider the treatment of the coefficients of the tori and Cayley graph expanders
to be our main technical contribution in this work.

Recall from \cref{cor:collect_coeffs} that the \cfunction~of $\Phi$ and $H$ satisfies
\[ a_{\Phi,H}(\top) = \sum_{\sigma \in \mathcal{L}(\Phi,H)} ~\prod_{v\in V(H)} (-1)^{|{\sigma}_{v}|-1} \cdot (|{\sigma}_{v}|-1)! \,.\]
In case that $H$ satisfies certain symmetry properties, we obtain that it suffices to
consider only those fractures in the previous sum that are fixed-points under suitable
group actions. More formally, we obtain the desired symmetries from the structure of the
groups underlying the Cayley graph constructions for tori and expanders as introduced in
the subsequent subsections

\subsection{Symmetries of the Torus}
We start with a simple Cayley graph given by the direct product of two cyclic groups:
\begin{defn}[The Torus]
    Let $\ell\geq 3$ denote an integer. The \emph{torus}, also called the \emph{toroidal
    grid}, $\torus_\ell$ of size $\ell$ is the Cayley graph of $\mathbb{Z}_\ell^2$ with
    generators $(\pm 1,0),(0,\pm 1)$, that is,
    \[\torus_\ell := \Gamma\left(\mathbb{Z}_\ell^2,\{(1,0),(-1,0),(0,1),(0,-1)\}\right) \,.\]

    Equivalently, the vertices of $\torus_\ell$ are $\mathbb{Z}_\ell^2$ and two vertices
    $(x,y)$ and $(x',y')$ are adjacent if and only if
    \begin{center}
        $x=x'$ and $y' = y \pm 1 \mod \ell$, or $y=y'$ and $x' = x\pm 1 \mod \ell$.
    \end{center}
    Consult \cref{fig:torgrid_intro} for a visualization.
    \lipicsEnd
\end{defn}

In the following, for simplicity, we write $+$ for (point-wise) addition modulo $\ell$.
Our goal is to understand the symmetries of $\torus_\ell$.
Consider the following action of $\ztwol$ on the vertices of $\torus_\ell$. Let $(i,j)\in
\ztwol$ and let $(x,y)\in V(\torus_\ell)$. We set $(i,j)\vdash (x,y) := (i,j) + (x,y)$.
The following is immediate:

\begin{fact}
    The action of $\ztwol$ on $V(\torus_\ell)$ is transitive. In particular, for every
    $(i,j)\in \ztwol$, the function $(i,j)\vdash \star$ is an automorphism of
    $\torus_\ell$.
    \lipicsEnd
\end{fact}
The fact above allows us to consider the set $\mathbb{Z}_\ell^2$ of all $(i,j)$-``shifts''
is a subgroup of the automorphism group of~$\torus_\ell$. We remark that not all
automorphisms are given by such shifts, but for our arguments we will not need to consider
the full group of automorphisms.

\paragraph*{Fractures of the Torus}

Recall that a fracture $\rho$ of a graph $H$ is a tuple $\rho = (\rho_v)_{v \in V(H)}$
where  $\rho_v$ is a  partition of the set $E_H(v)$ of edges of $H$ incident to $v$.
Now given an automorphism $\varphi : H \to H$ of $H$, it gives a bijection from the edges
$E_H(v)$ at $v$ to the edges~$E_H(\varphi(v))$ at~$\varphi(v)$. Thus, given a fracture
$\rho$ of $H$, we obtain a fracture~$\varphi(\rho)$ of~$H$, such that two edges $e_1, e_2
\in E_H(\varphi(v))$ are in the same block of $\varphi(\rho)_{\varphi(v)}$ if and only if
their preimages $\varphi^{-1}(e_1), \varphi^{-1}(e_2) \in E_H(v)$ are in the same block of
$\rho_v$.

\noindent We claim that that the two fractured graphs $\fracture{H}{\rho} \cong
\fracture{H}{\varphi(\rho)}$ are isomorphic.
To see this, note that the automorphism $\varphi: H \to H$ lifts to an automorphism
$\widehat \varphi : M_H \to M_H$ of the matching $M_H$ associated to $H$, where~$\widehat\varphi$ sends the vertex $v_e$ of $M_H$ to $\varphi(v)_{\varphi(e)}$. The map $\widehat
\varphi$ sends the equivalence relation on~$M_H$ associated to $\rho$ (with quotient
$\fracture{H}{\rho}$) to the equivalence relation associated to $\varphi(\rho)$ (with
quotient $\fracture{H}{\varphi(\rho)}$). Thus $\widehat \varphi : M_H \to M_H$ descends to
an isomorphism $\fracture{H}{\rho} \to \fracture{H}{\varphi(\rho)}$ fitting in a diagram
of graph homomorphisms, depicted in \cref{fig:ghomcd}.
\begin{figure}[t]
    \centering
    \begin{tikzcd}
        M_H \arrow[r] \arrow[d,"\widehat \varphi"] & \fracture{H}{\rho} \arrow[r] \arrow[d] & H \arrow[d,"\varphi"]\\
        M_H \arrow[r] & \fracture{H}{\varphi(\rho)} \arrow[r] & H
    \end{tikzcd}
    \caption{Each automorphism $\varphi$ of $H$ lifts to an automorphism $\hat{\varphi}$ of $M_H$. The latter descends to an isomorphism from $\fracture{H}{\rho}$ to $\fracture{H}{\varphi(\rho)}$, for every fracture $\rho$ of $H$.}\label{fig:ghomcd}
\end{figure}

Given a finite group $G$ acting on the graph $H$ by graph isomorphisms $\varphi_g : H \to
H$ (for $g \in G$), we obtain an action $\Vdash$ of $G$ on the lattice $\mathcal{L}(H)$ of
fractures on $H$, where $g \in G$ acts by $g \Vdash \rho = \varphi_g(\rho)$. Clearly, this
action respects the order of the lattice ($\rho  \leq \rho'$ if and only if $g \Vdash \rho
\leq g \Vdash \rho'$) and as seen above, for any $\widetilde \rho$ in the $G$-orbit of
$\rho$ we have $\fracture{H}{\widetilde \rho} \cong \fracture{H}{\rho}$.

We now return to the special case when $H=\torus_\ell$ is a torus. Here, given a
vertex $(i,j)$ of $H$ it is convenient to identify the edges incident to the
vertex (connecting it to $(i,j+1),(i,j-1),(i-1,j)$, and $(i+1,j)$) with the four
``directions'' $\tu,\td,\tl$, and $\tr$, respectively, so that each $\rho_{(i,j)}$
is a partition of the set
$\{\tu,\td,\tl,\tr\}$.

We have seen that $\ztwol$ acts transitively on the vertices of $\torus_\ell$ in such a
way that every element of $\ztwol$ induces an automorphism of $\torus_\ell$.
Thus, by the discussion above, we obtain an action $\Vdash$ of $\ztwol$ on the set of
fractures of $\torus_\ell$. Let us make this action explicit: $(i,j) \Vdash \rho :=
\hat{\rho}$, where $\hat{\rho}_{((i,j)\vdash (i',j'))} = \rho_{(i',j')}$ for all
$(i',j')\in \ztwol$.

\paragraph*{Analysis of the Fixed-points}
We proceed with the fixed-points of the action $\Vdash$ of $\ztwol$ on the fractures $\torus_\ell$. Since this action consists of (all possible)
$(i,j)$-shifts, the fixed-points are precisely those fractures $\rho$ for which all
partitions $\rho_{(i,j)}$
are equal --- recall that we assumed every $\rho_{(i,j)}$ to be
a partition of $\{\tu,\td,\tl,\tr\}$. Fortunately, there are only $15$ partitions of the
four-element set, and thus we can analyse the fixed-points by hand. Indeed, one special
case of our main result, as well as the classification of the parameterized Tutte
polynomial, rely on the understanding of all of those $15$ fixed-points.
However, while there are 15 fixed-points $\rho$, we can group those into 7 types according
to the isomorphism class of $\fracture{\torus_\ell}{\rho}$; an illustration of all
fixed-points is given in \cref{fig:fp}.

\begin{obs}\label{obs:fixedPointsbasic} The fixed-points of the action of $\ztwol$ on the
    fractures of $\torus_\ell$ are as follows.

    \noindent \textbf{\sffamily Matching:} $\fracture{\torus_\ell}{\rho}\cong M_{2\ell^2}$, the matching of size $2\ell^2$.
    \begin{enumerate}
        \item $\rho_{(i,j)}=\{\{\tu\},\{\td\},\{\tl\},\{\tr\}\}$ for all $(i,j)\in \ztwol$, that is, $\rho=\bot$.
            \medskip
    \end{enumerate}
    \textbf{\sffamily Matching and cycles:} $\fracture{\torus_\ell}{\rho}\cong M_{\ell^2}+\ell
    C_\ell$, the union of a matching of size $\ell^2$ and $\ell$ disjoint cycles of length $\ell$.
	\begin{enumerate}
		\setcounter{enumi}{1}
        \item $\rho_{(i,j)}=\{\{\tu,\td\},\{\tl\},\{\tr\}\}$ for all $(i,j)\in \ztwol$.
        \item $\rho_{(i,j)}=\{\{\tu\},\{\td\},\{\tl,\tr\}\}$ for all $(i,j)\in \ztwol$.
        \medskip
    \end{enumerate}
    \textbf{\sffamily Wedge packing:} $\fracture{\torus_\ell}{\rho}\cong \ell^2 P_2$, the union of
    $\ell^2$ disjoint paths of length $2$.
 	\begin{enumerate}
 		\setcounter{enumi}{3}
        \item $\rho_{(i,j)}=\{\{\tu,\tr\},\{\td\},\{\tl\}\}$ for all $(i,j)\in \ztwol$.
        \item $\rho_{(i,j)}=\{\{\tu,\tl\},\{\td\},\{\tr\}\}$ for all $(i,j)\in \ztwol$.
        \item $\rho_{(i,j)}=\{\{\td,\tl\},\{\tu\},\{\tr\}\}$ for all $(i,j)\in \ztwol$.
        \item $\rho_{(i,j)}=\{\{\td,\tr\},\{\tu\},\{\tl\}\}$ for all $(i,j)\in \ztwol$.
        \medskip
    \end{enumerate}
    \textbf{\sffamily Cycle packing I:} $\fracture{\torus_\ell}{\rho}\cong 2\ell C_\ell$,
    the union of $2\ell$ disjoint cycles of length $\ell$.
    \begin{enumerate}
        \setcounter{enumi}{7}
        \item $\rho_{(i,j)}=\{\{\tu,\td\},\{\tl,\tr\}\}$ for all $(i,j)\in \ztwol$.
            \medskip
    \end{enumerate}
    \textbf{\sffamily Cycle packing II:} $\fracture{\torus_\ell}{\rho}\cong \ell
    C_{2\ell}$, the union of $\ell$ disjoint cycles of length $2\ell$.
    \begin{enumerate}
        \setcounter{enumi}{8}
        \item $\rho_{(i,j)}=\{\{\tu,\tr\},\{\td,\tl\}\}$ for all $(i,j)\in \ztwol$.
        \item $\rho_{(i,j)}=\{\{\tu,\tl\},\{\td,\tr\}\}$ for all $(i,j)\in \ztwol$.
            \medskip
    \end{enumerate}
    \textbf{\sffamily Sun packing:} $\fracture{\torus_\ell}{\rho}\cong \ell S_{\ell}$, the
    union of $\ell$ suns of size $\ell$. Here a a sun of size $\ell$ is obtained from a
    cycle of length~$\ell$ by adding one ``dangling'' edge at every vertex of the cycle.
    \begin{enumerate}
        \setcounter{enumi}{10}
        \item $\rho_{(i,j)}=\{\{\tu\},\{\td,\tl,\tr\}\}$ for all $(i,j)\in \ztwol$.
        \item $\rho_{(i,j)}=\{\{\td\},\{\tu,\tl,\tr\}\}$ for all $(i,j)\in \ztwol$.
        \item $\rho_{(i,j)}=\{\{\tl\},\{\tu,\td,\tr\}\}$ for all $(i,j)\in \ztwol$.
        \item $\rho_{(i,j)}=\{\{\tr\},\{\tu,\tl,\td\}\}$ for all $(i,j)\in \ztwol$.
            \medskip
    \end{enumerate}
    \textbf{\sffamily Torus:} $\fracture{\torus_\ell}{\rho}\cong \torus_\ell$, the torus of size $\ell$.
    \begin{enumerate}
        \setcounter{enumi}{14}
    \item $\rho_{(i,j)}=\{\{\tu,\td,\tl,\tr\}\}$ for all $(i,j)\in \ztwol$, that is $\rho=\top$.
        \lipicsEnd
    \end{enumerate}
\end{obs}

\begin{figure}[t!]
   \centering
   \begin{tikzpicture}[scale=.87,transform shape]
       \pic[sty=10] at (0,0) {grid=3/0};
       \node at (0,-2.25) {\small ${\{\{\tu\},\{\td\},\{\tl\},\{\tr\}\}}$};
       \pic[sty=10] at (4,0) {grid=3/1};
       \node at (4,-2.25) {\small ${\{\{\tu, \td\},\{\tl\},\{\tr\}\}}$};
       \pic[sty=10] at (8,0) {grid=3/2};
       \node at (8,-2.25) {\small ${\{\{\tu\},\{\td\},\{\tl, \tr\}\}}$};
       \pic[sty=10] at (12,0) {grid=3/3};
       \node at (12,-2.25) {\small $\{\{\tr, \tu\},\{\tl\},\{\td\}\}$};
       \pic[sty=10] at (0,-5) {grid=3/4};
       \node at (0,-7.25) {\small $\{\{\tu\},\{\tl\},\{\td, \tr\}\}$};
       \pic[sty=10] at (4,-5) {grid=3/5};
       \node at (4,-7.25) {\small $\{\{\tu, \tl\},\{\td\},\{\tr\}\}$};
       \pic[sty=10] at (8,-5) {grid=3/6};
       \node at (8,-7.25) {\small $\{\{\tu\},\{\tl, \td\},\{\tr\}\}$};
       \pic[sty=10] at (12,-5) {grid=3/7};
       \node at (12,-7.25) {\small $\{\{\tr, \tu\},\{\tl, \td\}\}$};
       \pic[sty=10] at (0,-10) {grid=3/8};
       \node at (0,-12.25) {\small $\{\{\tu,\tl\},\{\td,\tr\}\}$};
       \pic[sty=10] at (4,-10) {grid=3/9};
       \node at (4,-12.25) {\small ${\{\{\tu, \td\},\{\tl, \tr\}\}}$};
       \pic[sty=10] at (8,-10) {grid=3/10};
       \node at (8,-12.25) {\small $\{\{\tl\},\{\tl, \td, \tr\}\}$};
       \pic[sty=10] at (12,-10) {grid=3/11};
       \node at (12,-12.25) {\small $\{\{\tu\},\{\tl, \td, \tr\}\}$};
       \pic[sty=10] at (2,-15) {grid=3/12};
       \node at (2,-17.25) {\small $\{\{\tu,\tl,\td\},\{\tr\}\}$};
       \pic[sty=10] at (6,-15) {grid=3/13};
       \node at (6,-17.25) {\small $\{\{\tu,\tl,\tr\},\{\td\}\}$};
       \pic[sty=10] at (10,-15) {grid=3/14};
       \node at (10,-17.25) {\small ${\{\{\tu, \tl, \td, \tr\}\}}$};
   \end{tikzpicture}
   \caption{The 15 fixed points of the action $\Vdash$ of $\mathbb{Z}_\ell^2$ on the fractures of
       $\torus_\ell$. Vertices in the fractured graphs that correspond to the same vertex in
       the original graph are encircled with a dashed line; the corresponding fixed point is
   denoted below its graphical representation. }\label{fig:fp}
\end{figure}
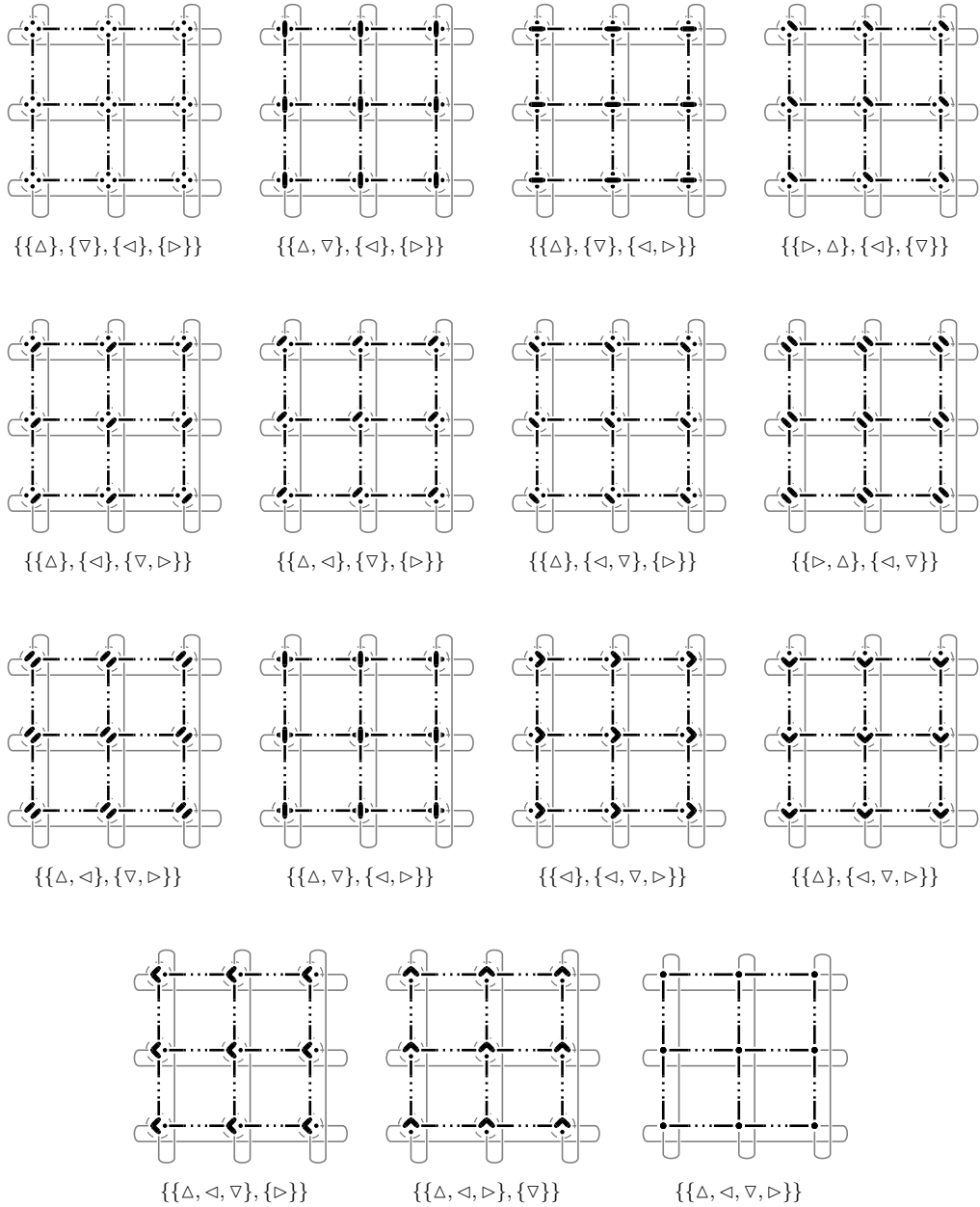
\afterpage{\clearpage}

While it might be surprising at first glance, we observe that, for many properties $\Phi$,
our analysis of the complexity of $\#\edgesubsprob(\Phi)$ \emph{only} depends on which
of the previous $15$ fixed-points $\rho$ satisfy that $\fracture{\torus_\ell}{\rho}$ has
the property $\Phi$.

\subsection{Symmetries of Cayley Graph Expanders of 2-groups}\label{sec:tori_exp}
For the second family of Cayley graphs, we rely on an explicit construction of $4$-regular
Cayley graph expanders due to Peyerimhoff and Vdovina~\cite{PeyerimhoffV11}. They are
constructed from an explicit infinite group $G$ with generators $x_0, x_1$ and a sequence
\[G \supseteq N_0 \supseteq N_1 \supseteq N_2 \supseteq \cdots\]
of normal subgroups $N_i$ of $G$, such that the indices $[G : N_i] = 2^{t_i} = n_i$ are
powers of $2$ converging to infinity as $i$ increases. Moreover, writing $\Kgroup_i$ for
the quotient group $G/N_i$, the set of Cayley graphs $G_i := \Gamma(\Kgroup_i,\{v_0^{\pm
1}, v_1^{\pm 1}\})$ is a family $\mathcal{G}$ of $(n_i,4,c)$-expanders for some constant
$c>0$. Here, $v_0 = x_0 N_i$ and $v_1 = x_1 N_i$ are generators of $\Kgroup_i$.

Similar to the case of the toroidal grid, we obtain an action of the group $\Kgroup_i$ on the
graph $G_i$, where an element $g \in \Kgroup_i$ acts on a vertex $v \in V(G_i) =
\Kgroup_i$ sending it to $g \vdash v = gv$, where the latter is the product of $g$ and $v$
in the group $\Kgroup_i$. The action of $g$ defines a graph automorphism of $G_i$ since
the four edges $\{v, v v_j^{\pm 1}\}$ (for $j=0,1$) at $v$ are sent to the four edges
$\{gv, gv v_j^{\pm 1}\}$ incident to $gv$.

\paragraph*{Fractures of the Cayley Graph Expanders}
For $v \in V(G_i) \cong \Kgroup_i$, the edges adjacent to $v$ connect $v$ to the vertices
$vs$ for $s \in S$ and thus can be uniquely labelled\footnote{We use the same notation
here as in the section about the torus grid $\torus_\ell$.
Many of the intuitions that we gained so far are still valid,
for instance the edge $\tr$ going out from the
vertex $v$ to $v v_0$ is equal to the edge $\tl$ associated to the vertex $v v_0$. On the
other hand, we also need to be more careful in our proofs, since e.g. going along an edge
$\tr$ followed by $\tu$ does not necessarily go to the same vertex as the path $\tu$
followed by $\tr$, since the group $\Kgroup_i$ is in general not abelian.} by
\begin{equation*}
    \tr = \{v, v v_0\},\ \tl = \{v, v v_0^{-1}\},\ \tu = \{v, v v_1\},\ \td = \{v, v v_1^{-1}\}.
\end{equation*}
Thus a fracture $\rho \in \mathcal{L}(G_i)$ is a collection $\rho = (\rho_v)_{v \in V(G_i)}$ of partitions of the set $\{\tr, \tl, \tu, \td\}$.

\paragraph*{Analysis of the Fixed-points}
As seen before, the action $\vdash$ of $\Kgroup_i$ on $G_i$ induces an action $\Vdash$ of $\Kgroup_i$ on the lattice of partitions $\mathcal{L}(G_i)$.
A fracture $\rho = (\rho_v)_{v \in V(G_i)}$ is invariant under the action of $\Kgroup_i$ if and only if $\rho_v$ does not depend on $v$.

Later we want to compute the coefficient $a_{\Phi,G_i}(\top)$ modulo two. As before we
observe that only fixed points of the action of $\Kgroup_i$  contribute, and
additionally we observe that only such fixed points $\rho$ can contribute where $\rho_v$
has at most two blocks. Thus in the following we consider only such fixed points.

\begin{lemma}\label{lem:fixedPointsexpander}
    Fix $i \geq 2$ and denote
    $b_0 = \mathrm{ord}_{\Kgroup_i}(v_0), b_1 = \mathrm{ord}_{\Kgroup_i}(v_1),  b_2 =
    \mathrm{ord}_{\Kgroup_i}(v_1 ^{-1} v_0), b_3 = \mathrm{ord}_{\Kgroup_i}(v_1 v_0),$
    and $a_j = \#\Kgroup_i / b_j$ for $j=0, \ldots, 3$.
    Then the fixed-points $\rho = (\rho_v)_{v \in V(G_i)}$ of the action of $\Kgroup_i$ on
    the fractures of $G_i$ satisfying that all $\rho_v$ have at most two blocks
    are as follows:

\noindent \textbf{\sffamily Cycle packing I:} 
	\begin{enumerate}
		\setcounter{enumi}{0}
        \item $\rho_{v}=\{\{\tu,\td\},\{\tl,\tr\}\}$ for all $v \in V(G_i)$ and $\fracture{G_i}{\rho} \cong a_0 \cdot C_{b_0} + a_1 \cdot C_{b_1}$.
        \medskip
    \end{enumerate}
 \textbf{\sffamily Cycle packing II:} 
	\begin{enumerate}
		\setcounter{enumi}{1}
        \item $\rho_{v}=\{\{\tu,\tr\},\{\td,\tl\}\}$ for all $v \in V(G_i)$ and $\fracture{G_i}{\rho} \cong a_2 \cdot C_{2 b_2}$.
        \item $\rho_{v}=\{\{\tu,\tl\},\{\td,\tr\}\}$ for all $v \in V(G_i)$ and $\fracture{G_i}{\rho} \cong a_3 \cdot C_{2 b_3}$.
            \medskip
    \end{enumerate}
 \textbf{\sffamily Sun packing:} 
	\begin{enumerate}
		\setcounter{enumi}{3}
        \item $\rho_{v}=\{\{\tu\},\{\td,\tl,\tr\}\}$ for all $v \in V(G_i)$ and $\fracture{G_i}{\rho} \cong a_0 \cdot S_{b_0}$.
        \item $\rho_{v}=\{\{\td\},\{\tu,\tl,\tr\}\}$ for all $v \in V(G_i)$ and $\fracture{G_i}{\rho} \cong a_0 \cdot S_{b_0}$.
        \item $\rho_{v}=\{\{\tl\},\{\tu,\td,\tr\}\}$ for all $v \in V(G_i)$ and $\fracture{G_i}{\rho} \cong a_1 \cdot S_{b_1}$.
        \item $\rho_{v}=\{\{\tr\},\{\tu,\tl,\td\}\}$ for all $v \in V(G_i)$ and $\fracture{G_i}{\rho} \cong a_1 \cdot S_{b_1}$.
        \medskip
    \end{enumerate}
 \textbf{\sffamily Full graph:} 
\begin{enumerate}
	\setcounter{enumi}{7}
        \item $\rho_{v}=\{\{\tu,\td,\tl,\tr\}\}$ for all $v \in V(G_i)$, that is $\rho=\top$ and $\fracture{G_i}{\rho} \cong G_i$.
    \end{enumerate}

Moreover, the numbers $a_j, b_j$ are all powers of $2$ and $a_j \geq 8$.
\end{lemma}
\begin{proof}
In cases 1,2,3 it follows from the definition of the fractured graph that
$\fracture{G_i}{\rho}$ is $2$-regular and thus a union of circles.
In case 1 the first type of circles (associated to the directions $\tl, \tr$) is given by
\begin{equation} \label{eqn:w0circle} w_0 \to w_0 v_0 \to w_0 v_0^2 \to \cdots \to w_0 v_0^{b_0-1} \to w_0 v_0^{b_0} = w_0\end{equation}
and thus isomorphic to $C_{b_0}$, with one circle for each $w_0 \in K_i / \langle v_0 \rangle$ giving a total number of $\#(\Kgroup_i / \langle v_0 \rangle) = \#\Kgroup_i / b_i = a_i$. Analogously we obtain $a_1$ copies of $C_{b_1}$ associated to the directions $\tu, \td$.

In case 2 the circles are of the form
\begin{align*}
w_0 &\to w_0 v_1^{-1} \to w_0 v_1^{-1} v_0 \to w_0 v_1^{-1} v_0 v_1^{-1} \to w_0 (v_1^{-1} v_0)^{2} \to \cdots \\ &\to w_0 (v_1^{-1} v_0)^{b_2-1} v_1^{-1} \to w_0 (v_1^{-1} v_0)^{b_2} = w_0.\end{align*}
Thus they are isomorphic to $C_{2 b_2}$ and since the total number of edges of $\fracture{G_i}{\rho}$ is equal to $\#E(G_i) = 4 \#\Kgroup_i/2 = 2 \#\Kgroup_i$, the number of copies of $C_{2 b_2}$ is given by $2 \#\Kgroup_i / (2 b_2) =  a_2$. The case 3 is treated analogously.

In case 4, the connected component of a vertex $w_0 \in \fracture{G_i}{\rho}$ associated to the directions $\td,\tl,\tr$  certainly contains the circle $C_{b_0}$ given by \eqref{eqn:w0circle} and in addition, each vertex $w_0 v_0^{\ell}$ is connected to $w_0 v_0^{\ell} v_1^{-1}$, which forms a leaf of $\fracture{G_i}{\rho}$. Thus, these are the only additional vertices connected to the circle and thus the connected component of each vertex in $\fracture{G_i}{\rho}$ forms a sun $S_{b_0}$. The total number of suns is $\#E(G_i) / \#E(S_{b_0}) = (2 \#\Kgroup_i)/(2 b_0) = a_0$. The cases 5,6,7 are treated completely analogously.

Finally, case 8 follows from the general property $\fracture{H}{\top} \cong H$. The fact that $a_j, b_j$ divide the order of $\#\Kgroup_i$ together with the property that $\Kgroup_i$ is a $2$-group, imply that $a_j, b_j$ are powers of $2$. Finally, we show the inequality $a_j \geq 8$ by induction on $i \geq 2$. Note that in the case $i=2$ this can be checked by hand.
For this one uses the explicit description of the group law of $\Kgroup_2$ presented in \cite[Section 3]{PeyerimhoffV11} and verifies that the orders $b_j$ of elements $v_0, v_1, v_1^{-1} v_0, v_1 v_0$ are precisely $4$, so that $a_j = \#\Kgroup_2 / b_j = 8$.

To conclude the general case for $i \geq 2$, denote $V_0^{i} = \langle v_0 \rangle \subseteq \Kgroup_i$ the subgroup generated by $v_0$, so that $a_0 = [\Kgroup_i : V_0^i]$. Recalling the facts from the start of the section, we saw that $\Kgroup_i = G / N_i$ with $N_2 \supseteq N_i$ for $i \geq 3$. Thus we have a surjective group homomorphism
\[\varphi_i : \Kgroup_i = G / N_i \to G / N_2 = \Kgroup_2\,,~ xN_i \mapsto xN_2 \]
sending $V_0^i \subseteq \Kgroup_i$ to $V_0^2 \subseteq \Kgroup_2$ (this follows since the generator $v_0 = x_0 N_i$ of $V_0^i$ maps to the generator $v_0 = x_0 N_2$ of $V_0^2$). As mentioned above, we checked by hand that $V_0^2$ has index $8$ in $\Kgroup_2$. Then by \cref{Lem:indexmonotone} we have that $8 = [\Kgroup_2 : V_0^2] = [\varphi(\Kgroup_i) : \varphi(V_0^i)]$ divides $[\Kgroup_i : V_0^i] = a_i$ and thus $a_0 \geq 8$. The bounds for $a_1, a_2, a_3$ work exactly the same way.
\end{proof}


\begin{lemma} \label{Lem:indexmonotone}
    Let $G,G'$ denote finite groups and $\varphi : G \to G'$ a group homomorphism. Then,
    for any subgroup $H \subseteq G$ we have that $[\varphi(G) : \varphi(H)]$ divides $[G
    : H]$.
\end{lemma}
\begin{proof}
    Let $K = \mathrm{ker}(\varphi)$ and $K_H = K \cap H = \mathrm{ker}(\varphi|_{H})$,
    then by the First isomorphism theorem we have $\varphi(G) \cong G/K$ and $\varphi(H)
    \cong H/K_H$. Using this, we have
    \begin{align*}
        [G:H] &= \frac{\#G}{\#H} = \frac{\#G / \#K}{\#H / \#K_H} \cdot  \frac{\#K}{\#K_H} = \frac{\#(G / K)}{\#(H / K_H)}\cdot  \frac{\#K}{\#K_H} = \frac{\#\varphi(G)}{\#\varphi(H)}\cdot \frac{\#K}{\#K_H} \\
        ~&= [\varphi(G):\varphi(H)] \cdot \frac{\#K}{\#K_H}.
    \end{align*}
    But $K_H \subseteq K$ is a subgroup, so by Lagrange's theorem, the order of $K_H$
    divides the order of $K$, so that $\#K/\# K_H$ is an integer. Thus the above equality
    shows that $[\varphi(G) : \varphi(H)]$ divides $[G : H]$.
\end{proof}

\subsection{Analysis of the Coefficient Function via Fixed-points}
While the value $a_{\Phi,H}(\top)$ of the \cfunction~seems to be very difficult to handle
for arbitrary graphs~$H$, we now use our observations on the symmetries of the
torus and the Cayley graph expanders to prove that the \cfunction~does not vanish
infinitely often under specific constraints on the behaviour of~$\Phi$ on the fixed-points
presented in the preceding section.

We start with the case of $a_{\Phi,\torus_\ell}(\top)$, which, while being simple, turns
out to be required for one of the special cases in our main classification for
minor-closed graph properties:
\begin{lemma}\label{lem:coef_special_case}
    Let $\ell$ denote a prime and let $\Phi$ denote a computable graph property. We have that
\begin{align*}
    a_{\Phi,\torus_\ell}(\top) &= -6\Phi(M_{2\ell^2}) +4 \Phi(M_{\ell^2} + \ell C_\ell)
    + 8\Phi(\ell^2P_2)\\ &\quad- \Phi(2\ell C_\ell) -2\Phi(\ell C_{2\ell}) -4\Phi(\ell S_\ell) +
\Phi(\torus_\ell) \mod \ell \,.\end{align*}
\end{lemma}
\begin{proof}
    By \cref{cor:collect_coeffs} we have
    \[ a_{\Phi,\torus_\ell}(\top) = \sum_{\sigma \in \mathcal{L}(\Phi,\torus_\ell)} ~\prod_{v\in V(\torus_\ell)} (-1)^{|{\sigma}_{v}|-1} \cdot (|{\sigma}_{v}|-1)! \]
    Setting $f(\sigma):= \prod_{v\in V(\torus_\ell)} (-1)^{|{\sigma}_{v}|-1} \cdot (|{\sigma}_{v}|-1)!$, this rewrites to
    \[a_{\Phi,\torus_\ell}(\top) = \sum_{\sigma \in \mathcal{L}(\Phi,\torus_\ell)} f(\sigma)\,. \]
    We now use the action $\Vdash$ of $\ztwol$ on the subset
    $\mathcal{L}(\Phi,\torus_\ell)$ of $\mathcal{L}(\torus_\ell)$, given by permuting the
    elements of a fracture $\rho$ according to the coordinate shift induced by an element
    $(i,j)\in \ztwol$. Restricting this action to $\mathcal{L}(\Phi,\torus_\ell)$ is
    well-defined since the action does not change the isomorphism class\footnote{Note that
    while the action \emph{can} change the isomorphism class \emph{as a
    $\torus_\ell$-coloured graph}, the property $\Phi$ only depends on the underlying
    uncoloured graph, which is unchanged, and thus $\mathcal{L}(\Phi,\torus_\ell)$ is indeed
    invariant under the action.} of $\fracture{\torus_\ell}{\rho}$. In particular, we have
    that $f(\sigma)=f(\rho)$ whenever $\sigma$ and $\rho$ are in the same orbit of the action.
    This allows us to rewrite as follows; the sum is taken over all orbits $[\sigma]$ of the
    group action:
    \[ a_{\Phi,\torus_\ell}(\top) = \sum_{[\sigma]} \#[\sigma]  \cdot f(\sigma) \]
    Since $\ell$ is a prime, the group order of $\ztwol$ is a power of $\ell$. As the size
    of every orbit must divide the group order, we can ignore all orbits which are not
    fixed-points, that is $\sigma$ for which $\#[\sigma]=1$, if we take the sum
    modulo~$\ell$. All $15$ fixed-points are explicitly given in
    \cref{obs:fixedPointsbasic}. Let us now compute the coefficients of each collection of
    fixed-points that induce the same graph, up to isomorphism; we use Fermat's Little
    Theorem---recall that $\ell$ is a prime.

    \medskip

     \noindent \textbf{\sffamily Matching:} One fixed-point $\rho$ satisfies
     $\fracture{\torus_\ell}{\rho}\cong M_{2\ell^2}$. The contribution to
     $a_{\Phi,\torus_\ell}(\top)$ is thus
    \[1\cdot f(\rho)\cdot \Phi(M_{2\ell^2})=((-1)^{4-1} \cdot (4-1)!)^{\ell^2} \Phi(M_{2\ell^2}) = -6 \Phi(M_{2\ell^2}) \mod \ell \,.\]
     \textbf{\sffamily Matching and cycles:} Two fixed-points $\rho$ satisfy $\fracture{\torus_\ell}{\rho}\cong M_{\ell^2} +\ell C_\ell$. The contribution to $a_{\Phi,\torus_\ell}(\top)$ is thus
    \[2\cdot f(\rho)\cdot \Phi(M_{\ell^2}+\ell C_\ell) = 2\cdot ((-1)^{3-1} \cdot (3-1)!)^{\ell^2} \Phi(M_{\ell^2}+\ell C_\ell) = 4 \Phi(M_{\ell^2}+\ell C_\ell) \mod \ell \,. \]
     \textbf{\sffamily Wedge packing:} Four fixed-points $\rho$ satisfy $\fracture{\torus_\ell}{\rho}\cong \ell^2 P_2$.
    The contribution to $a_{\Phi,\torus_\ell}(\top)$ is thus
    \[4\cdot f(\rho)\cdot \Phi(\ell^2 P_2) = 4\cdot ((-1)^{3-1} \cdot (3-1)!)^{\ell^2} \Phi(\ell^2 P_2) = 8 \Phi(\ell^2 P_2) \mod \ell \,. \]
     \textbf{\sffamily Cycle packing I:} One fixed-point $\rho$ satisfies $\fracture{\torus_\ell}{\rho}\cong 2\ell C_\ell$.
    The contribution to $a_{\Phi,\torus_\ell}(\top)$ is thus
    \[1\cdot f(\rho)\cdot \Phi(2\ell C_\ell) = ((-1)^{2-1} \cdot (2-1)!)^{\ell^2} \Phi(2\ell C_\ell) = - \Phi(2\ell C_\ell) \mod \ell \,. \]
     \textbf{\sffamily Cycle packing II:} Two fixed-points $\rho$ satisfy $\fracture{\torus_\ell}{\rho}\cong \ell C_{2\ell}$.
    The contribution to $a_{\Phi,\torus_\ell}(\top)$ is thus
    \[2\cdot f(\rho)\cdot \Phi(\ell C_{2\ell}) = 2\cdot ((-1)^{2-1} \cdot (2-1)!)^{\ell^2} \Phi(\ell C_{2\ell}) = -2 \Phi(\ell C_{2\ell}) \mod \ell \,. \]
     \textbf{\sffamily Sun packing:} Four fixed-points $\rho$ satisfy $\fracture{\torus_\ell}{\rho}\cong \ell S_{\ell}$.
    The contribution to $a_{\Phi,\torus_\ell}(\top)$ is thus
    \[4\cdot f(\rho)\cdot \Phi(\ell S_{\ell}) = 4\cdot ((-1)^{2-1} \cdot (2-1)!)^{\ell^2} \Phi(\ell S_{\ell}) = -4 \Phi(\ell S_{\ell}) \mod \ell \,. \]
     \textbf{Torus:} One fixed-point $\rho$ satisfies $\fracture{\torus_\ell}{\rho}\cong \torus_\ell$.
    The contribution to $a_{\Phi,\torus_\ell}(\top)$ is thus
    \[1\cdot f(\rho)\cdot \Phi(\torus_\ell) = ((-1)^{1-1} \cdot (1-1)!)^{\ell^2} \Phi(\torus_\ell) =  \Phi(\torus_\ell) \mod \ell \,. \]
     Taking the sum of the previous terms (modulo $\ell$) concludes the proof.
\end{proof}

 We proceed with a similar lemma for the Cayley graph expanders.

\begin{lemma}\label{lem:expander_coeffs}
    Let $\mathcal{G}=\{G_1,G_2,\dots \}$ denote the family of Cayley graph expanders given in
    \cref{sec:tori_exp} and let $\Phi$ denote a computable graph property. For $i \geq 2$ we have
    \[
        a_{\Phi,G_i}(\top) = \Phi(a_0 \cdot C_{b_0} + a_1 \cdot C_{b_1}) + \Phi(a_2 \cdot C_{2 b_2}) + \Phi(a_3 \cdot C_{2 b_3}) + \Phi(G_i)  \mod 2\,.
    \]
    Moreover, the numbers $a_j, b_j$ are all powers of $2$ and $a_j \geq 8$.
\end{lemma}
\begin{proof}
    By \cref{cor:collect_coeffs} we have
    \[ a_{\Phi,G_i}(\top) = \sum_{\rho \in \mathcal{L}(\Phi,G_i)} ~\prod_{v\in V(G_i)} (-1)^{|{\rho}_{v}|-1} \cdot (|{\rho}_{v}|-1)! \]
    Setting $f(\rho):= \prod_{v\in V(G_i)} (-1)^{|{\rho}_{v}|-1} \cdot (|{\rho}_{v}|-1)!$,
    this rewrites to
    \[a_{\Phi,G_i}(\top) = \sum_{\rho \in \mathcal{L}(\Phi,G_i)} f(\rho)\,. \]
    As before, the action of the $2$-group $\Kgroup_i$ leaves the set
    $\mathcal{L}(\Phi,G_i)$ invariant and modulo $2$ the contribution of all elements
    $\rho$ not fixed under $\Kgroup_i$ vanishes. Thus we only consider the fixed points
    $\rho = (\rho_v)_{v \in V(G_i)}$, for which $\rho_v$ is independent of $v$.

    From the formula of $f(\rho)$ it is easy to see that $f(\rho) = 1 \mod 2$ if $\rho$ has at
    most two blocks and $f(\rho) = 0 \mod 2$ otherwise. Thus only the fractures $\rho$ from
    cases 1 to 8 of \cref{lem:fixedPointsexpander}  can give a nontrivial contribution to
    $a_{\Phi,G_i}(\top)$. The fixed-point $\rho$ contributes if and only if
    $\Phi(\fracture{G_i}{\rho})=1$. Finally, since the pairs of cases 4,5 and 6,7 lead to
    isomorphic graphs $\fracture{G_i}{\rho}$, any possible contributions from these cancel
    modulo $2$ and we are left with the four summands above.
\end{proof}

\section{Exact Counting of Small Subgraph Patterns}
Building upon our analysis of the \cfunction{} of the torus and the Cayley graph expanders above, we are now able to present the proofs of our results on exact counting.

\subsection{Hardness for Minor-closed Properties}\label{sec:minor_classification}
We present an exhaustive and explicit complexity dichotomy of $\#\edgesubsprob(\Phi)$ for properties~$\Phi$ that are minor-closed. Recall that a graph $H$ is a minor of a graph $G$ if it
can be obtained from $G$ by a sequence of vertex-deletions, edge-deletions and
edge-contractions (where multiple edges and self-loops are deleted). A property $\Phi$ is
\emph{minor-closed} if, for all graphs $G$ with $\Phi(G)=1$, we have that $\Phi$ is true
for all minors of~$G$ as well.

Given a minor-closed property $\Phi$, by the celebrated Robertson-Seymour
Theorem~\cite{RobertsonS04}, there is a \emph{finite} set~$\mathcal{F}$ such that for
all graphs~$H$ we have that $\Phi(H)=1$ if and only if no graph in~$\mathcal{F}$ is a
minor of~$H$. Recall that $\#\edgesubsprob(\Phi)$ is fixed-parameter tractable if $\Phi$
has bounded matching number or if $\Phi$ is trivially true. We show that
$\#\edgesubsprob(\Phi)$ is $\#\W{1}$-hard in all other cases, given that $\Phi$ is
minor-closed.
It turns out that we need to distinguish\footnote{For example, if $\mathcal{F}$ only
contains the path of two edges, then $\#\edgesubsprob(\Phi)$ is the problem of counting
$k$-matchings. If we would be able to use the Cayley graph expanders for this property as
given by \cref{lem:expander_coeffs}, then the \cfunction{} would satisfy $a_{\Phi,k}(\top)=1 \mod 2$, which could be used to establish that counting $k$-matchings modulo $2$ is hard, contradicting the fact that the latter problem is known to be polynomial time solvable~\cite[Section~1.4]{BjorklundDH15}.} whether $\mathcal{F}$ contains a graph of degree at most $2$.

\begin{lemma}\label{lem:minor_dicho_deg2}
    Let $\Phi$ denote a minor-closed graph property with unbounded matching number and assume that $\mathcal{F}$ contains a graph $F$ of degree at most $2$. Then $\mathcal{H}[\Phi,\torus]$ is infinite.
\end{lemma}
\begin{proof}
    The assumption that $\Phi$ has unbounded matching number implies that $\Phi$ is
    satisfied by graphs containing arbitrarily large matchings. Since $\Phi$ is closed
    under taking minors, this implies that $\Phi(M_k)=1$ for all $k$.

    Now observe that since $F$ has degree at most $2$, the graph $F$ is a union of paths
    and cycles. Therefore, there is a constant $c$ (only depending on $F$) such that for
    all $\ell>c$, the graph $F$ is a minor of each of the following graphs:
    \[ M_{\ell^2} + \ell C_\ell, 2\ell C_\ell, \ell C_{2\ell}, \ell S_\ell, \text{ and } \torus_\ell \,. \]
    Indeed, any finite union of paths and cycles can be obtained as a minor of
    sufficiently large cycle packings, sun packings and tori.

    Now assume that $\ell$ is additionally a prime and greater than $3$. By
    \cref{lem:coef_special_case}, we thus have that
    \[a_{\Phi, \torus_\ell}(\top) = -6\Phi(M_{2\ell^2}) + 8\Phi(\ell^2 P_2) = -6 +
    8\Phi(\ell^2 P_2) \mod \ell \,.\]
    The claim follows by observing that $-6 + 8\Phi(\ell^2 P_2) \neq 0 \mod \ell$ for
    every prime $\ell>3$, regardless on whether $\Phi(\ell^2 P_2)=1$ or $\Phi(\ell^2
    P_2)=0$.
\end{proof}

Recall that $\mathcal{G}$ is the family of Cayley graph expanders introduced in
\cref{sec:tori_exp}.
\begin{lemma}\label{lem:minor_dicho_nodeg2}
    Let $\Phi$ denote a minor-closed graph property which is not trivially true, and
    assume that~$\mathcal{F}$ does not contain a graph $F$ of degree at most $2$. Then
    $\mathcal{H}[\Phi,\mathcal{G}]$ is infinite.
\end{lemma}
\begin{proof}
    Since $\Phi$ is not trivially true, the set $\mathcal{F}$ is non-empty. Thus let $F$
    denote an arbitrary graph in~$\mathcal{F}$. By \cref{fac:niceExpanders}, there is an
    index $j$ such that for all $i\geq j$, the graph $G_i$ contains the complete graph on
    $\#V(F)$ vertices (and thus also $F$) as a minor. In other words, $\Phi(G_i)=0$ for
    all $i\geq j$. By \cref{lem:expander_coeffs}, we have that for all $i\geq 2$
    \[
        a_{\Phi,G_i}(\top) = \Phi(a_0 \cdot C_{b_0} + a_1 \cdot C_{b_1}) + \Phi(a_2 \cdot C_{2 b_2}) + \Phi(a_3 \cdot C_{2 b_3}) + \Phi(G_i)  \mod 2\,.
    \]
    Hence, for $i\geq \max\{2,j\}$, we have
    \[
        a_{\Phi,G_i}(\top) = \Phi(a_0 \cdot C_{b_0} + a_1 \cdot C_{b_1}) + \Phi(a_2 \cdot C_{2 b_2}) + \Phi(a_3 \cdot C_{2 b_3})  \mod 2\,.
    \]
    Finally, we rely on the premise of the lemma, implying that each graph in
    $\mathcal{F}$ has a vertex of degree at least $3$. Consequently, no graph in
    $\mathcal{F}$ can be a minor of a cycle-packing. Thus $\Phi(a_0 \cdot C_{b_0} + a_1
    \cdot C_{b_1}) = \Phi(a_2 \cdot C_{2 b_2}) = \Phi(a_3 \cdot C_{2 b_3}) =1$, and,
    consequently, $a_{\Phi,G_i}(\top) =1 \mod 2$ for each $i\geq \max\{2,j\}$.
\end{proof}

\noindent We are finally able to prove our main result for exact counting; note that all minor-closed graph properties are (polynomial-time) computable due to the finite set of forbidden minors.

\begin{theorem}\label{thm:main_minor_exact_only}
    Let $\Phi$ denote a minor-closed graph property.
    If $\Phi$ is either trivially true or of bounded matching number,\footnote{We say that
        a property has bounded matching number if there is a constant bound on the size of a
    largest matching in graphs satisfying $\Phi$.} then the (exact) counting version
    $\#\edgesubsprob(\Phi)$ is fixed-parameter tractable. Otherwise $\#\edgesubsprob(\Phi)$ is
    $\#\W{1}$-hard. If, additionally, each forbidden minor of $\Phi$ has a vertex of degree at
    least $3$, then $\#\edgesubsprob(\Phi)$ cannot be solved in time
    \[ f(k)\cdot |G|^{o(k/\log k)} \,,\]
    for any function $f$, unless the Exponential Time Hypothesis fails.
\end{theorem}
\begin{proof}
    The (fixed-parameter) tractability part is given by \cref{thm:smallmatchFPT}. If
    $\Phi$ has unbounded matching number, but at least one forbidden minor is of degree at
    most $2$, then, by \cref{lem:minor_dicho_deg2}, the set $\mathcal{H}[\Phi,\torus]$ is
    infinite, which implies $\#\W{1}$-hardness by \cref{lem:hardness_basis}.

    If $\Phi$ is not trivially true and each forbidden minor has a vertex of degree at least
    $3$, then, by \cref{lem:minor_dicho_nodeg2}, the set $\mathcal{H}[\Phi,\mathcal{G}]$ is
    infinite. Again by \cref{lem:hardness_basis}, this implies both, $\#\W{1}$-hardness
    and the conditional lower bound.
\end{proof}

\subsection{Hardness for Selected Natural Graph Properties}

In addition to classifying
$\#\edgesubsprob(\Phi)$ for minor-closed properties $\Phi$, we can also use the criteria for
establishing $\#\W{1}$-hardness and an almost tight conditional lower bound of
$\#\edgesubsprob(\Phi)$ directly to some natural, but non-minor-closed properties. With this
we aim to illustrate the simplicity of applying our fixed-points result for Cayley graph
expanders (\cref{lem:expander_coeffs}) to explicitly given graph properties.

\corintrofur*
\begin{proof}
    Our proof proceeds by applying \cref{lem:expander_coeffs} to show $a_{\Phi,G_i}(\top) \neq 0$ for each of the properties $\Phi$, allowing us to conclude using \cref{lem:hardness_basis}.

    For $\Phi \in  \{\Phi_C, \Phi_H, \Phi_E\}$, observe that the graphs
    $ a_0 \cdot C_{b_0} + a_1 \cdot C_{b_1}$, $a_2 \cdot C_{2 b_2}$, and $a_3 \cdot
    C_{2 b_3}$
    are each disconnected (and hence not Hamiltonian, nor Eulerian either) if $a_i
    \geq 8$ for $i=1,2,3$.
    Further, the graphs $G_i$ are connected since Cayley graphs
    are connected\footnote{Recall that our definition of Cayley graphs enforces the set $S$ to be a set of generators
    of the group.} Thus, the graphs $G_i$ are also Eulerian, since they are $4$-regular.
    Moreover, Cayley graphs of $p$-groups are Hamiltonian~\cite{Witte86}.
    Thus, by \cref{lem:expander_coeffs}, we have that for $i\geq 2$:
    \[
        a_{\Phi,G_i}(\top) = \Phi(a_0 \cdot C_{b_0} + a_1 \cdot C_{b_1}) + \Phi(a_2 \cdot C_{2 b_2}) + \Phi(a_3 \cdot C_{2 b_3}) + \Phi(G_i) = 1  \mod 2\,.
    \]
    Consequently, $\mathcal{H}[\Phi,\mathcal{G}]$ is infinite if $\Phi \in \{\Phi_C,
    \Phi_H, \Phi_E\}$.
    By \cref{lem:hardness_basis}, we obtain both, $\#\W{1}$-hardness and the conditional
    lower bound.

    \noindent For $\Phi = \Phi_{CF}$ we can perform a similar analysis: observe that cycle-packings are
    always claw-free. On the other hand, for each $i> 2$, the graphs $G_i$ do contain an
    (induced) claw. To see this,
    let $e_{\Kgroup_i}$ denote the neutral element of $\Kgroup_i$ and consider the vertices of $G_i$ associated to  $e_{\Kgroup_i}$, $v_0$,
    $v_1$ and $v_1^{-1}$. While $e_{\Kgroup_i}$ is adjacent to the remaining three cosets,
    it is easy to check by hand that $v_0$, $v_1$ and $v_1^{-1}$ constitute an
    independent set in $G_i$.

    Consequently, by \cref{lem:expander_coeffs}, we have that for $i>2$:
    \[
        a_{\Phi,G_i}(\top) = \Phi(a_0 \cdot C_{b_0} + a_1 \cdot C_{b_1}) + \Phi(a_2 \cdot C_{2 b_2}) + \Phi(a_3 \cdot C_{2 b_3}) + \Phi(G_i) = 3+0 = 1  \mod 2\,.
    \]
    Thus, $\mathcal{H}[\Phi_{CF},\mathcal{G}]$ is infinite.
    By \cref{lem:hardness_basis}, we hence obtain both $\#\W{1}$-hardness and the
    conditional lower bound.
\end{proof}

\section{Approximate Counting of Small Subgraph Patterns}\label{sec:approx}
Recall that we identified $\#\edgesubsprob(\Phi)$ as an inherently hard problem in case we
aim for \emph{exactly} counting the solutions. In particular, we established
$\#\W{1}$-hardness for any non-trivial minor-closed property~$\Phi$ of unbounded matching
number. For this reason, the section below deals with the complexity of
\emph{approximating} the number of solutions. Tractability of approximating the solutions
of parameterized counting problems is given by the notion of a fixed-parameter tractable
randomized approximation scheme.
\begin{defn}[FPTRAS~\cite{ArvindR02,Meeks16}]\label{def:FPTRAS}
    Let $(P,\kappa)$ denote a parameterized counting problem. A \emph{fixed-parameter
    tractable randomized approximation scheme} ``FPTRAS'' for $(P,\kappa)$ is a randomized
    algorithm $\mathbb{A}$ that, given $x\in \Sigma^\ast$ and rational numbers
    $\varepsilon >0$ and $0< \delta< 1$ computes an integer $z$ such that
    \[ \pr{(1-\varepsilon) P(x) \leq z \leq (1+\varepsilon) P(x)} \geq 1-\delta \,.\]
    The running time of $\mathbb{A}$ must be bounded by $f(\kappa(x))\cdot
    \mathsf{poly}(|x|,1/\varepsilon,\log(1/\delta))$
    for some computable function~$f$.
    \lipicsEnd
\end{defn}

Indeed, we can show that $\#\edgesubsprob(\Phi)$ allows an FPTRAS for every minor-closed
property $\Phi$. In fact, we prove the following general criterion, which implies the
existence of an FPTRAS for minor-closed properties.

\apxmain*

We start with the case of $\Phi$ satisfying both the matching and the star criterion. For
readers familiar with the meta-theorem of Dell, Lapinskas and Meeks~\cite{DellLM20}, we
point out that their method cannot be used to achieve the desired goal in the current
setting: the results in~\cite[Section~1.3]{DellLM20} imply that $\#\edgesubsprob(\Phi)$
admits an FPTRAS whenever the \emph{edge-colourful decision version} of
$\edgesubsprob(\Phi)$ is fixed-parameter tractable; in the latter, we expect as input a
graph $G$ with $k$ different edge-colours and the goal is to decide whether there is a
subset $A$ of edges containing each colour exactly once such that $G[A]$ satisfies $\Phi$
(w.r.t.\ the underlying uncoloured graph). Thus, if we could show that the edge-colourful
decision version is fixed-parameter tractable for properties satisfying the matching and
the star criterion, \cref{thm:approx_main_intro} would follow.

However, the latter cannot be true (unless $\ccFPT=\W{1}$) since the following property
$\Phi$ induces a $\W{1}$-hard colourful decision version, while satisfying both the
matching and the star criterion: $\Phi(H)=1$ if and only if $H$ is either a star, a
matching, or the union of a clique and a triangle. $\W{1}$-hardness follows from a
reduction from finding edge-colourful $k$-cliques in a graph, which is known to be
$\W{1}$-hard.\footnote{The \emph{vertex-colourful} clique problem is $\W{1}$-hard~(see
Chapter 13 in~\cite{CyganFKLMPPS15}) and reduces to the \emph{edge-colourful} clique
problem by assigning an edge $\{u,v\}$ the colour $\{c(u),c(v)\}$, where $c(u)$ and $c(v)$
are the vertex-colours of $u$ and $v$.} The reduction is straightforward: given a graph
$G$ with $\binom{k}{2}$ edge colours, we construct a graph $G'$ by adding a triangle with
three fresh colours to the graph. Then $G'$ contains a colourful
$\binom{k}{2}+3$-edge-subset $A$ that satisfies $\Phi$ if and only if $G$ contains an
edge-colourful $k$-clique. The latter is true since any colourful
$\binom{k}{2}+3$-edge-subset must contain the triangle with the three fresh colours and
can thus neither induce a star, nor a matching.

Being unable to rely on the colourful decision version, we thus use a different approach
using Ramsey's Theorem, similarly to the one in~\cite{Meeks16}. More precisely, we
use the following consequence:

\begin{lemma}
    Let $k\geq 4$ denote a positive integer and let $G$ denote a graph with at least
    $R(k,k)$ edges. Then $G$ contains either $K_{1,k}$ or $M_k$ as a subgraph.
\end{lemma}
\begin{proof}
    We apply Ramsey's Theorem to the line graph $L(G)$ of $G$: The vertices of $L(G)$ are
    the edges of~$G$, and two vertices $e$ and $e'$ of $L(G)$ are adjacent if and only if
    $e\cap e'\neq 0$.
    Sine $L(G)$ contains at least $R(k,k)$ vertices, Ramsey's Theorem implies that $L(G)$
    either contains an independent set of a clique of size~$k$. Note that a
    $k$-independent set of $L(G)$ corresponds to a $k$-matching in $G$, and that a
    $k$-clique in~$L(G)$ corresponds to a star $K_{1,k}$ in $G$; the latter requires that
    $k\geq 4$ since the line graph of a triangle is a triangle (and thus a clique) as
    well.
\end{proof}

The subsequent observation enables our Monte-Carlo algorithm to only rely on ``FPT-many'' samples:
\begin{lemma}\label{lem:few_samples}
    Let $k\geq 4$ denote a positive integer and let $G$ denote a graph with at least
    $R(k,k)$ edges. Assume a subset~$A$ of $k$ edges is sampled uniformly at random. We
    have \[ \pr{G[A] \cong M_k \vee G[A] \cong K_{1,k}} \geq \binom{R(k,k)}{k}^{-1} \,.\]
\end{lemma}
\begin{proof}
    Set $m=|E(G)|$ and $r=R(k,k)$. It is convenient to assume that $A$ is sampled as
    follows: we first choose $r$ edges u.a.r., denote this set by $S$, and afterwards we
    obtain $A$ by choosing $k$ edges among $S$ u.a.r.; of course, we need to show that
    this yields a uniform distribution. Let $B$ denote any $k$-edge subset of $G$. By the law
    of total probability, we have that
    \begin{align*}
    	 \pr{A=B} &= \sum_{T \in \binom{E(G)}{r}} \pr{S=T} \cdot \pr{A=B\mid S=T}\\
    	 ~&= \sum_{T
    		\in \binom{E(G)}{r}} \binom{m}{r}^{-1} \cdot \pr{A=B\mid S=T}\,.
    \end{align*}
    Note that $\pr{A=B\mid S=T}=\binom{r}{k}^{-1}$ if $B\subseteq T$, and $\pr{A=B\mid
    S=T}=0$ otherwise. Consequently
\begin{align*}
	 \pr{A=B} &= \#\{T\subseteq E(G) ~|~ B\subseteq T ~\wedge~ \#T=r \}\cdot \binom{m}{r}^{-1}  \binom{r}{k}^{-1}\\
	 ~&= \binom{m-k}{r-k}  \binom{m}{r}^{-1}  \binom{r}{k}^{-1} = \binom{m}{k}^{-1} \,.
\end{align*}
    Now let $\mathcal{E}$ denote the event $G[A] \cong M_k \vee G[A] \cong K_{1,k}$ and note
    that for every $r$-edge subset $T$ of $G$ we have that $\pr{\mathcal{E}\mid S=T}\geq
    \binom{r}{k}^{-1}$ since, by the previous lemma, $G[T]$ contains either $M_k$ or
    $K_{1,k}$ as a subgraph. We conclude that
    \[\pr{\mathcal{E}} =  \sum_{ T \in \binom{E(G)}{r}} \pr{S=T} \cdot \pr{\mathcal{E}\mid
    S=T}  = \binom{m}{r} \binom{m}{r}^{-1} \cdot \pr{\mathcal{E}\mid S=T} \geq \binom{r}{k}^{-1}\,,  \]
    which concludes the proof.
\end{proof}

\noindent For our FPTRAS, we use the following (consequence of a) Chernoff bound:
\begin{theorem}[see Theorem 11.1 in~\cite{MitzenmacherU17}]
    Let $X_1,\dots,X_t$ denote independent and identically distributed indicator random variables
    with expectation $\eta = E[X_i]$, and let $0< \varepsilon,\delta < 1$ denote
    positive rationals.
	If $t \geq (3 \ln(2/\delta))/(\varepsilon^2 \eta)$, then
    \[ \pr{ \left| \frac{1}{t}\cdot \sum_{i=1}^t X_i - \eta\right| < \varepsilon \eta }
    \geq 1-\delta \,.\lipicsEnd\]
\end{theorem}

\begin{lemma}\label{lem:MS_approx}
    Let $\Phi$ denote a computable graph property satisfying both, the matching criterion
    and the star criterion. Then $\#\edgesubsprob(\Phi)$ has an FPTRAS.
\end{lemma}
\begin{proof}
    By assumption, there is a constant $c'$ such that $\Phi$ is true for all matchings
    and stars of size at least $c'$; we set $c=\max({c',4})$. Our FPTRAS $\mathbb{A}$ is
    constructed as follows: If $k<c$ or if $|E(G)|\leq R(k,k)$, then we solve the problem
    (exactly) by the naive brute-force algorithm. Otherwise, we take
    \[\binom{R(k,k)}{k} \cdot \frac{3\ln(2/\delta)}{\varepsilon^2}\]
    many independent samples of $k$-edge sets $A$ of $G$, each taken uniformly at random.
    Finally, we output the fraction of those samples $A$ such that $\Phi(G[A])=1$. Consult
    \cref{alg1} for a visualization as pseudo-code.

    \begin{algorithm}[t]
        \SetKwBlock{Begin}{}{end}
        \SetKwFunction{esf}{MatchingsAndStarsFPTRAS}
        \esf{$G$, $k$, $\varepsilon$, $\delta$}\Begin{
            \If{$k < c$ \KwSty{or} $|E(G)|\leq R(k,k)$}{
                Solve the problem exactly by brute force.\;
                }\Else{
                $X \gets 0$; $t \gets \binom{R(k,k)}{k} \cdot
                \frac{3\ln(2/\delta)}{\varepsilon^2}$\;
                \For{$i \gets 1$ \KwSty{to} $t$}{
                    Sample a $k$-edge subset $A$ of $G$ uniformly at random.\;
                    \lIf{$\Phi(G[A])=1$}{$X \gets X + 1$}
                }
                \Return{$\frac{X}{t}\cdot \binom{|E(G)|}{k}$}\;
            }
        }
        \caption{An FPTRAS for $\#\edgesubsprob(\Phi)$ if $\Phi$ satisfies the matching and the star criterion.}\label{alg1}
    \end{algorithm}

    \noindent Let us first argue about the running time: if $k<c$ then the brute force algorithm
    takes time at most~$|G|^c$,\footnote{$|G|$ rather than $|E(G)|$ since $G$ might
    contain many isolated vertices.} and if $|E(G)|\leq R(k,k)$ then the brute force
    algorithm takes time at most $|G| + R(k,k)^k$. Otherwise, we iterate through the loop
    $t$ times, and each iteration can clearly be done in time $f'(k)\cdot
    \mathsf{poly}(|G|)$ for some computable function $f'$ --- note that the factor $f'(k)$
    depends on the complexity of verifying whether $\Phi(G[A])$ holds, which might require
    super-polynomial time in $|G[A]|\in O(k)$. The overall running time is thus bounded by
    \[ \max \biggl\{ |G|^c, |G|+R(k,k)^k, \binom{R(k,k)}{k} \cdot (3 \ln(2/\delta))/\varepsilon^2 \cdot f'(k) \cdot \mathsf{poly}(|G|) \biggr\}\,,\]
    which is bounded by $f(k) \cdot \mathsf{poly}(|G|,1/\varepsilon,\log(1/\delta))$
    for some computable function $f$.

	Next note that correctness is trivial in case the brute force algorithm is executed.
    Hence assume that~$k\geq c$ and $|E(G)|> R(k,k)$. To avoid notational clutter, we set
    $r:=R(k,k)$ and $m:= |E(G)|$. Now let $X_i$ denote the indicator variable defined to
    be $1$ if the $i$-th sample, denoted $A_i$, satisfies $\Phi(G[A_i])=1$, and $X_i=0$
    otherwise. Observe that $E[X_i]= \#\edgesubs{\Phi,k}{G} \cdot \binom{m}{k}^{-1}$ for
    all $i$. In what follows, we thus just set $\eta:=E[X_i]$. Since $\Phi$ is true
    for $M_k$ and $K_{1,k}$, and by \cref{lem:few_samples} we furthermore have
    \[ \eta= \pr{\Phi(G[A])=1} \geq \pr{G[A]\cong M_k \vee G[A] \cong K_{1,k}} \geq \binom{r}{k}^{-1} \,.\]
    Consequently, $t\geq  (3 \ln(2/\delta))/(\varepsilon^2 \eta)$. By the previous Chernoff bound, we thus have
    \[ \pr{ \left| \frac{1}{t}\cdot \sum_{i=1}^t X_i - \eta\right| < \varepsilon \eta } \geq 1-\delta \,.\]
    Finally, recall that $X=\sum_{i=1}^t X_i$ and observe that
    \begin{align*}
        ~&\left| \frac{1}{t}\cdot \sum_{i=1}^t X_i - \eta\right| < \varepsilon \eta\\
         \Rightarrow &\left|\frac{X}{t} -\frac{\#\edgesubs{\Phi,k}{G}}{\binom{m}{k}}\right| < \varepsilon \cdot \frac{\#\edgesubs{\Phi,k}{G}}{\binom{m}{k}} \\
        \Rightarrow & \left|\frac{X}{t}\cdot \binom{m}{k} -\#\edgesubs{\Phi,k}{G}\right| < \varepsilon \cdot \#\edgesubs{\Phi,k}{G} \,.
    \end{align*}
    \noindent We conclude the proof by pointing out that the latter implies
    \[ (1-\varepsilon)\cdot \#\edgesubs{\Phi,k}{G} \leq  \frac{X}{t}\cdot \binom{m}{k} \leq (1+\varepsilon)\cdot \#\edgesubs{\Phi,k}{G}\,.\]
\end{proof}

For the case of $\Phi$ having bounded treewidth, we rely on the following result of Arvind and Raman; to this end, given a fixed positive integer $T$, let $\#\subsprob(T)$ denote the problem that, on input a graph $H$ of treewidth at most $T$ and an arbitrary graph $G$, requires to compute $\#\subs{H}{G}$.

\begin{theorem}[\cite{ArvindR02}]\label{thm:ArvindRFPTRAS}
    For each positive integer $T$, there is an FPTRAS for $\#\subsprob(T)$ if parameterized by the size of the graph $H$.
    \lipicsEnd
\end{theorem}

\begin{lemma}\label{lem:tw_approx}
    Let $\Phi$ denote a computable graph property. If $\Phi$ has bounded treewidth, then $\#\edgesubsprob(\Phi)$ admits an FPTRAS.
\end{lemma}
\begin{proof}
    By assumption, there is a constant $T$ such that the treewidth of each graph $H$ with
    $\Phi(H)=1$ is at most $T$. Define $g(k):= |\Phi_k|$ and observe that $g$ is
    computable as $\Phi$ is.

    Recall from equation~\eqref{eq:unlcolouredSubgraphCounts} that for each $G$ and $k$ we have
    \[\#\edgesubs{\Phi,k}{G} = \sum_{H\in \Phi_k} \#\subs{H}{G} \]
	We thus just use the FPTRAS from \cref{thm:ArvindRFPTRAS} to approximate (with
    probability $1-\delta/g(k)$) each term $\#\subs{H}{G}$ with $H\in \Phi_k$ and output
    the sum given by the previous equation.

    Observe that approximating each term $\#\subs{H}{G}$ takes time at most \[
        f'(|H|)\cdot \mathsf{poly}(|G|,1/\varepsilon,\log(g(k)/\delta)),
    \] for some computable function $f'$.

    Since each $H\in \Phi_k$ has $k$ edges, the overall running time is thus clearly
    bounded by
    \[ f(k)\cdot\mathsf{poly}(|G|,1/\varepsilon,\log(\delta))  \]
    for some computable function $f$---note that $f$ depends on $\Phi$, $f'$ and $g$,
    but the latter three are independent of the input.
    Now let $r$ denote the output of our algorithm. It remains to show that
    \[ \pr{ (1-\varepsilon)\cdot \#\edgesubs{\Phi,k}{G} \leq r \leq (1+\varepsilon)\cdot
    \#\edgesubs{\Phi,k}{G} } \geq 1-\delta \,.\]
    Write $r_H$ for the output of the FPRAS from \cref{thm:ArvindRFPTRAS} on input $G$,
    $H$, $\varepsilon$, and $\delta/g(k)$. Then \[r=\sum_{H\in \Phi_k}r_H\,,\] and the
    following holds for each $H\in \Phi_k$
    \[ \pr{(1-\varepsilon)\cdot  \#\subs{H}{G} \leq r_H \leq (1+\varepsilon) \cdot \#\subs{H}{G} } \geq 1-\delta/g(k) \,.\]
    Since the outcomes $r_H$ are independent and $g(k)=|\Phi_k|$, we have
    \[ \pr{\forall H\in \Phi_H : (1-\varepsilon)\cdot \#\subs{H}{G} \leq r_H \leq
    (1+\varepsilon)\cdot \#\subs{H}{G}} \geq (1-\delta/g(k))^{g(k)}\,,\]
	which is at most $(1-\delta)$ by Bernoulli's inequality.\pagebreak

	\noindent Consequently, with probability at least $(1-\delta)$, we have that
    \begin{align*}
        (1-\varepsilon)\cdot \#\edgesubs{\Phi,k}{G} & = (1-\varepsilon) \sum_{H\in \Phi_k} \#\subs{H}{G}\\
        ~ & = \sum_{H\in \Phi_k} (1-\varepsilon)\cdot \#\subs{H}{G}\\
        ~ & \leq \sum_{H\in \Phi_k} r_H ~(= r)\\
        ~ & \leq \sum_{H\in \Phi_k} (1+\varepsilon)\cdot \#\subs{H}{G}\\
        ~ & = (1+\varepsilon)  \sum_{H\in \Phi_k} \#\subs{H}{G}\\
        ~ & = (1+\varepsilon)\cdot \#\edgesubs{\Phi,k}{G} \,,
    \end{align*}
    which concludes the proof.
\end{proof}

\begin{proof}[Proof of \cref{thm:approx_main_intro}]
    Holds by \cref{lem:MS_approx,lem:tw_approx}.
\end{proof}

\section{Detection of Small Subgraph Patterns}\label{sec:dec}
In this section, we study the complexity of the decision problem $\edgesubsprob(\Phi)$. As
a first observation we observe that $\edgesubsprob(\Phi)$ essentially subsumes the
(parameterized) subgraph isomorphism problem: consider for instance the property $\Phi$
defined as $\Phi(H)=1$ if and only if $H\cong K_{\ell,\ell}$ for some positive integer~$\ell$. Then $\edgesubsprob(\Phi)$ is equivalent to the problem $k$-$\textsc{BICLIQUE}$
which was only recently shown to be $\W{1}$-hard by the seminal result of Lin~\cite{Lin18}
after being unresolved for at least a decade.

More generally, let $\mathcal{H}$ denote a class of graphs and define $\embsprob(\mathcal{H})$
as the problem that asks, given a graph $H\in \mathcal{H}$ and an arbitrary graph $G$,
whether there is a subgraph embedding from $H$ to $G$; the parameterization is given
by $|H|$. Plehn and Voigt~\cite{PlehnV90} proved $\embsprob(\mathcal{H})$ to be
fixed-parameter tractable whenever the treewidth of graphs in~$\mathcal{H}$ is bounded by
a constant. On the other hand, the question whether $\embsprob(\mathcal{H})$ is
$\W{1}$-hard in all remaining cases is one of the ``most infamous''~\cite[Chapter
33.1]{DowneyF13} open problems in parameterized complexity. Since $\edgesubsprob(\Phi)$
subsumes\footnote{To be precise, $\edgesubsprob(\Phi)$ subsumes $\embsprob(\mathcal{H})$
whenever $\mathcal{H}$ does not contain two graphs with the same number of edges, which
is, however, true for most of the natural instances of the subgraph isomorphism problem
such as finding cliques, bicliques, cycles, paths and matchings, only to name a few.}
$\embsprob(\mathcal{H})$ as we have seen in case of $k$-$\textsc{BICLIQUE}$, a complete
classification of $\edgesubsprob(\Phi)$ seems to be elusive at the moment.

However, we identify the following tractable instances of $\edgesubsprob(\Phi)$, which significantly extends the case of bounded treewidth.

\decclass*

In case of $\Phi$ satisfying the matching or the star criterion, fixed-parameter
tractability is obtained by a surprisingly simple Win-Win approach relying on the
treewidth and the maximum degree of a graph. Assume, for example, that $\Phi$ is true for
all matchings. Now, given a graph $G$ and an integer $k$, we can easily verify whether $G$
contains a maximum matching of size at least $k$. If the latter is true, $G$ contains a
subgraph with $k$ edges that satisfies $\Phi$. More interestingly, if the matching number of $G$ is bounded by $k$, then its vertex-cover number (and thus its treewidth) is bounded by $2k$, and we can
efficiently use dynamic programming over a tree-decomposition of small width of~$G$ to
verify whether $\edgesubs{\Phi,k}{G}\neq \emptyset$. Formally, the latter can be
established by an easy application of Courcelle's Theorem~\cite{Courcelle90} as shown in
the following lemma:

\begin{lemma}
    Let $\Phi$ denote a computable graph property. There is a computable function $g$ and
    an algorithm~$\mathbb{A}$ that, given a graph $G$ and a positive integer $k$,
    correctly decides whether $\edgesubs{\Phi,k}{G}\neq \emptyset$ in time
    $g(\mathsf{tw}(G),k) \cdot |G|$.
\end{lemma}
\begin{proof}
    We use Courcelle's Theorem as stated in~\cite[Theorem~11.37]{FlumG06}. Thus it
    remains to provide an MSO-sentence\footnote{We refer the reader to e.g.\ Chapter~4
    in~\cite{FlumG06} for an introduction to Monadic Second Order (MSO) logic.} $\varphi$
    such that $G$ satisfies $\varphi$ if and only if $\edgesubs{\Phi,k}{G}\neq \emptyset$.
    To this end, let $H\in \Phi_k$ and assume that $V(H)=\{1,\dots,v_H\}$. Consider the
    following sentence
    \[ \varphi_H := \exists x_1,\dots, \exists x_{v_H} : \bigwedge_{i\neq j} x_i \neq x_j \wedge \bigwedge_{\{i,j\} \in E(H)} E(x_i,x_j) \,.\]
    Observe that $G$ satisfies $\varphi_H$ if and only if $H$ is a subgraph of $G$. Consequently, we set
    \[\varphi := \bigvee_{H\in \Phi_k} \varphi_H \,.\]
    Since the length of $\varphi$ only depends on $\Phi$ and $k$, the lemma holds by Courcelle's Theorem.
\end{proof}

We are now able to establish fixed-parameter tractability of $\edgesubsprob(\Phi)$ whenever $\Phi$ satisfies the matching criterion.

\begin{lemma}\label{lem:dec_matchings}
    Let $\Phi$ denote a computable graph property that satisfies the matching criterion.
    Then the problem $\edgesubsprob(\Phi)$ is fixed-parameter tractable.
\end{lemma}
\begin{proof}
    Since $\Phi$ satisfies the matching criterion, there is a constant $c$ (only
    depending on $\Phi$) such that $\Phi(M_k)=1$ for all $k\geq c$. The FPT algorithm is
    constructed as follows:

    Given a graph $G$ and a positive integer $k$, we can assume that $k \geq c$, solving the case $k<c$ by brute force enumeration of all $k$-subsets of edges.
    In the case $k\geq c$, we compute a maximum matching $M$ of $G$ in polynomial time by, e.g.,
    the Blossom Algorithm~\cite{Edmonds65}. If $|M|\geq k$, then we can output $1$, since
    any $k$-subset $A$ of $M$ satisfies that $\Phi(G[A])=1$ by assumption.

    In the remaining case, we can thus assume that the matching number of $G$ is bounded
    by $k$. Consequently, the vertex cover number of $G$ is bounded by $2k$. Since the
    treewidth of a graph is bounded by its vertex cover number, we conclude that
    $\mathsf{tw}(G)\leq 2k$. Invoking the algorithm from the previous lemma thus yields an
    overall running time bounded by
	\[  m^{O(1)} + g(2k,k)\cdot |G| \,,\]
	which proves fixed-parameter tractability.
\end{proof}

We continue with the case of $\Phi$ satisfying the star criterion. To this end, we
require the following result, which is implicitly implied by the counting version of the Frick-Grohe-Theorem~\cite{FrickG01}; we provide a proof based on the
bounded search-tree paradigm for completeness.

\begin{lemma}
    Let $\Phi$ denote a computable graph property. There is a computable function $g$ and
    an algorithm~$\mathbb{A}$ that, given a graph $G$ and a positive integer $k$,
    correctly decides whether $\edgesubs{\Phi,k}{G}\neq \emptyset$ in time
    $g(\mathsf{deg}(G),k) \cdot |G|$.
\end{lemma}
\begin{proof}
    We check for each $H\in \Phi_k$ whether $H$ is a subgraph of $G$ and output $1$ if
    (and only if) at least one of those checks is positive.

    Assume for a moment that $H$ is connected. In this case, the strategy is very simple:
    We guess a vertex~$v$ of~$G$ and search for a subgraph embedding of $H$ in $G$ that
    includes $v$. Since $H$ is connected and has $k$ edges, the image of the subgraph
    embedding can only contain vertices of distance at most $k$ from $v$. This allows us
    to search for a copy of $H$ in the $\leq k$ neighbourhood of $v$ by brute-force, since
    the latter contains at most $\mathsf{deg}(G)^k$ vertices. The overall running time of
    finding a subgraph isomorphic to $H$ in $G$ is thus bounded by $|V(G)|\cdot
    \mathsf{deg}(G)^k$.

    The situation becomes slightly more complicated if $H$ is not connected. We would like
    to perform the previous strategy for each connected component of $H$, adding an
    additional factor of $k$ in the worst case. However, since a subgraph embedding needs
    to be injective, we have to guarantee that we do not construct a solution that uses
    vertices of $G$ twice. This issue is solved by a standard application of
    colour-coding: We choose a function $\mathsf{col}:V(G)\rightarrow V(H)$ uniformly at
    random. If $G$ contains a subgraph isomorphic to $H$, then with probability at least
    $p(k)>0$ there is a subgraph embedding $\psi:V(H)\rightarrow V(G)$ such that
    additionally $\mathsf{col}(\psi(v))=v$ for each vertex $v\in V(H)$, and such a
    subgraph embedding can be found in time $O(k\cdot |V(G)|\cdot \mathsf{deg}(G)^k)$ by
    adapting the above strategy for every connected component $H$ accordingly.
    Finally, derandomization can be achieved by perfect hashing as shown
    in~\cite{AlonYZ95} (see also~\cite[Chapter 13.3]{FlumG06}).
\end{proof}

Let us now establish fixed-parameter tractability of $\edgesubsprob(\Phi)$ whenever $\Phi$ satisfies the star criterion.

\begin{lemma}\label{lem:dec_stars}
    Let $\Phi$ denote a computable graph property satisfying the star criterion. Then the problem $\edgesubsprob(\Phi)$ is fixed-parameter tractable.
\end{lemma}
\begin{proof}
    Since $\Phi$ satisfies the star criterion, there is a constant $c$ (only depending on
    $\Phi$) such that $\Phi(K_{1,k})=1$ for all $k\geq c$. The FPT algorithm is
    constructed as follows:

    Given a graph $G$ and a positive integer $k$, we can again solve the case $k<c$ by brute force and thus assume $k \geq c$.
    Then, we check whether $G$ contains a vertex $v$ of degree at least $k$, in
    which case we can output $1$, since any $k$-subset $A$ of the incident edges of $v$
    satisfies that $\Phi(G[A])=1$ by assumption.

    In the remaining case, we can thus assume that $\mathsf{deg}(G)\leq k$. Invoking the
    algorithm from the previous lemma thus yields an overall running time bounded by
    \[  m^{O(1)} + g(k,k)\cdot |G| \,,\]
    which proves fixed-parameter tractability.
\end{proof}

\begin{proof}[Proof of \cref{thm:dec_classification}]
    In case $\Phi$ satisfies the matching criterion or the star criterion, the claim holds
    by \cref{lem:dec_matchings} and \cref{lem:dec_stars}. If $\Phi$ has bounded treewidth,
    then, given $G$ and $k$, we can use the algorithm of Plehn and Voigt~\cite{PlehnV90}
    for each $H\in \Phi_k$. Since the size of $\Phi_k$ is bounded by a function in $k$,
    the overall running time still yields fixed-parameter tractability.
\end{proof}

Our main result regarding minor-closed properties is now obtained by the combination of our results in the realms of exact counting, approximate counting, as well as decision:

\minclosehard*
\begin{proof}
    Note that each minor-closed property is computable (even in polynomial time) by the
    Robertson-Seymour Theorem~\cite{RobertsonS04}.
    The classification of exact counting follows by
    \cref{thm:main_minor_exact_only}. For approximate counting and decision, we
    claim that each minor-closed property $\Phi$ either has bounded treewidth or satisfies
    both, the matching \emph{and} the star criterion. If the latter holds, then the
    existence of an FPTRAS for approximate counting follows by
    \cref{thm:approx_main_intro}, and the FPT algorithm for decision follows by
    \cref{thm:dec_classification}.

    To prove the claim, we assume that $\Phi$ has unbounded
    treewidth; otherwise we are done. In that case, by the
    Excluded-Grid-Theorem~\cite{RobertsonS86-ExGrid},
    $\Phi$ holds for a sequence of graphs containing arbitrarily large grids as minors.
    Since every planar graph (including matchings and stars) is a minor of a
    grid~\cite{RobertsonST94}, and $\Phi$ is minor-closed, we conclude that $\Phi$ holds
    for all matchings and all stars and thus satisfies both the matching and the star
    criterion.
\end{proof}

\subsection{Separating Approximate Counting and Decision}\label{sec:approx_dec}
Below we establish the existence of a (computable) graph property $\Psi$ such that
$\edgesubsprob(\Psi)$ is fixed-parameter tractable but $\#\edgesubsprob(\Psi)$ does not
admit an FPTRAS unless $\W{1}$ coincides with $\mathsf{FPT}$, under randomised
parameterized reductions.

\noindent We rely on the subgraph isomorphism problem restricted to \emph{grids}: given a positive
integer $k$, the $k$\emph{-grid}, denoted by $\boxplus_k$, has vertices $[k]\times[k]$ and
edges
\[ E(\boxplus_k):= \{ \{(i,j),(i',j')\} ~:~ |i-i'|+|j-j'|=1 \} \,.\]
In other words, $\boxplus_k$ is obtained from the torus $\torus_k$ by removing edges $\{(0,j),(k-1,j)\}$ and $\{(i,0),(i,k-1)\}$ for all $i,j \in [k]$. Hence observe that $\#E(\boxplus_k)=2k(k-1)$ and consider the definition of $\Psi$:

\begin{defn}[Property $\Psi$] \label{Def:propertyPsi}
    Let $H$ denote a graph. We set
    \begin{equation}
        \Psi(H)=1 :\Leftrightarrow
        \begin{cases}
            H \cong \boxplus_k + K_{1,k^2+2k} \vee H \cong M_{3k^2} & \exists k: \#E(H)=3k^2 \\
            H \cong M_{\#E(H)} &\text{otherwise}
        \end{cases}
    \end{equation}
    Here, $\boxplus_k + K_{1,k^2+2k}$ is the (disjoint) union of the $k$-grid and the star of size $k^2+2k$. In particular, $\boxplus_k + K_{1,k^2+2k}$ has precisely $2k(k-1)+k^2+2k=3k^2$ edges.
    \lipicsEnd
\end{defn}
Observe that $\Psi$ is clearly computable and satisfies the matching criterion (but not
the star criterion). Thus $\edgesubsprob(\Psi)$ is fixed-parameter tractable by
\cref{thm:dec_classification}.

Write $\boxplus$ for the set of all grids and recall that the problem
$\embsprob(\boxplus)$ asks, given as input a grid $\boxplus_k$ and a graph~$G$, to
correctly decide whether there is an embedding from $\boxplus_k$ to $G$. Chen, Grohe and
Lin~\cite{ChenGL17} proved that this problem is $\W{1}$-hard. The following result
establishes thus hardness of approximating $\#\edgesubsprob(\Psi)$.
\begin{theorem}\label{thm:decapproxsep}
    If $\#\edgesubsprob(\Psi)$ admits an FPTRAS, then there is a randomized decision procedure~$\mathbb{P}$ which, given a graph $G$, a positive integer $k$, and a rational number $0<\delta<1$, satisfies
    \[ \pr{\mathbb{P}(G,k,\delta)=1 \Leftrightarrow \embs{\boxplus_k}{G}\neq \emptyset}\geq (1-\delta) \,.\]
    Furthermore, the running time of $\mathbb{P}$ is bounded by
    $f(k)\cdot \mathsf{poly}(|G|,\log(1/\delta))$ for some computable function~$f$.
    \lipicsEnd
\end{theorem}

The previous theorem is an easy consequence of the following lemma.
\begin{lemma}\label{lem:amplification_gadget}
    Let $G$ denote a graph with $n$ vertices, let $k$ denote a positive integer, and set $G'=G+K_{1,n^6}$. We have
    \[
        \#\edgesubs{\Psi,3k^2}{G'} \begin{cases}
            \geq (3k^2)^{-3k^2}n^{6k^2+12k} & \text{if }\embs{\boxplus_k}{G} \neq \emptyset \\
            \leq 2 n^{6k^2+4}&\text{otherwise}
        \end{cases}
    \]
\end{lemma}
\begin{proof}
    We rely on the following well-known bounds on the binomial coefficient:
    \[\frac{n^k}{k^k} \leq \binom{n}{k} \leq n^k\,.\]
    Let us start with the lower bound; thus assume that $\embs{\boxplus_k}{G}\neq \emptyset$. Consequently, there is a subset $A$ of $2k(k-1)$ many edges in $G$ such that $G[A]\cong \boxplus_k$. Observe further that there are \[\binom{n^6}{k^2+2k}\geq (k^2+2k)^{-(k^2+2k)} n^{6(k^2+2k)} \geq (3k^2)^{-3k^2}n^{6k^2+12k} \]
    edge subsets $A'$ of $K_{1,n^6}$ that induce $K_{1,k^2+2k}$. Thus, for any such $A'$,
    we have that $G'[A\dot\cup A']\cong \boxplus_k + K_{1,k^2+2k}$. In particular, there
    are hence at least $(3k^2)^{-3k^2}n^{6k^2+12k}$ edge subsets of size $3k^2$ of $G'$
    that induce a graph satisfying~$\Psi$.

    \noindent For the second case, recall that we wish to upper bound the number of $3k^2$-edge
    subsets $A$ of $G'$ such that $\Psi(G[A])=1$. By definition of $\Psi$ and $G'$, and
    under the assumption that $\embs{\boxplus_k}{G}= \emptyset$, it remains to upper bound
    the number of $3k^2$-matchings of $G'$. Note that each matching of $G'$ can use at
    most one edge of the star $K_{1,n^6}$. In particular, this allows us to partition the
    $3k^2$-matchings in two groups: For the first one, every edge of the matching must be
    contained in $G$, and for the second one, precisely one edge is contained in
    $K_{1,n^6}$ (for which there are $n^6$ possibilities), and all remaining edges must be
    contained in $G$. Since $G$ has less than $n^2$ edges, we can generously bound the
    number of $3k^2$-matchings of $G'$ as follows:
    \[ \binom{n^2}{3k^2} + n^6\cdot \binom{n^2}{3k^2-1} \leq n^{6k^2} + n^6\cdot n^{6k^2-2}\leq 2n^{6k^2+4} \,.\]
    The proof is thus concluded.
\end{proof}

\begin{proof}[Proof of \cref{thm:decapproxsep}]
    Assume that ${\mathbb{A}}$ is an FPTRAS for $\#\edgesubsprob(\Psi)$. Given $G$ with
    $n$ vertices and~$\boxplus_k$ (for which we wish to decide whether
    $\embs{\boxplus_k}{G}\neq \emptyset$), we first check whether $n\leq 6(3k^2)^{3k^2}$.
    If this is the case, then $\mathbb{P}$ search for an embedding from $\boxplus_k$ to
    $G$ via brute-force, the running time of which is bounded by $g(k)$ for some
    computable function $g$ since the size of $G$ is bounded by a function in $k$.\pagebreak

    \noindent Thus assume that $n> 6(3k^2)^{3k^2}$. Then $\mathbb{P}$ constructs $G'$ as in
    \cref{lem:amplification_gadget} in time $n^6$ and simulates $\mathbb{A}$ on $G'$,
    $3k^2$, $\delta$ and $\varepsilon=1/2$. Finally, $\mathbb{P}$ outputs $0$ if the
    output of $\mathbb{A}$ is at most $3n^{6k^2+4}$, and $\mathbb{P}$ outputs $1$
    otherwise.
	Since $\mathbb{A}$ is an FPTRAS and $\varepsilon=1/2$, its running time is bounded by
    $f'(3k^2)\cdot \mathsf{poly}(|G'|,\log(1/\delta))$ for some computable function $f'$.
    Since $|G'|$ is bounded polynomial in $|G|$, we conclude that $\mathbb{P}$ has the
    desired running time.

    It remains to prove correctness. Given that $\mathbb{A}$ is an FPTRAS and
    $\varepsilon=1/2$, we note that, with probability at least $(1-\delta)$, the output
    $X$ of $\mathbb{A}$ satisfies
    \[ 1/2\cdot  \#\edgesubs{\Psi,3k^2}{G'} \leq X \leq 3/2\cdot
    \#\edgesubs{\Psi,3k^2}{G'},.\]
    Assume first that $\embs{\boxplus_k}{G}=\emptyset$. Then, by
    \cref{lem:amplification_gadget}, we have
    \[ X \leq 3/2\cdot  \#\edgesubs{\Psi,3k^2}{G'} \leq 3n^{6k^2+4} \,,\]
    and thus the output of $\mathbb{P}$ is correct.

    Now assume that $\embs{\boxplus_k}{G}\neq\emptyset$. We have to show that $X >
    3n^{6k^2+4}$ for $\mathbb{P}$ being correct. Using the assumption that $n>
    6(3k^2)^{3k^2}$, and relying on \cref{lem:amplification_gadget} once more, we have
    \[ X \geq  1/2\cdot  \#\edgesubs{\Psi,3k^2}{G'} \geq  1/2 \cdot
    (3k^2)^{-3k^2}n^{6k^2+12k} > 3 n^{6k^2+12k-1} > 3 n^{6k^2+4} \,, \]
    where the last inequality is trivial since $k>0$. This concludes the proof.
\end{proof}
Finally, since $\embsprob(\boxplus)$ is $\W{1}$-hard, the previous theorem yields that the
existence of an FPTRAS for $\#\edgesubsprob(\Phi)$ would imply that $\W{1}$ coincides with
$\ccFPT$ under randomized parameterized reductions, which proves \cref{thm:sep_approx_dec_intro}.

\section{A Parameterized Tutte Polynomial}\label{sec:tutte}
In the last part of the paper, we take a step back and revisit exact counting: Recall that
problem $\#\edgesubsprob(\Phi)$ can be interpreted as the problem of evaluating a linear
combination of subgraph counts, given by
\[ \#\edgesubs{\Phi,k}{\ast} = \sum_{H\in \Phi_k} \#\subs{H}{G} \,,\]
where $\Phi_k$ is the set of all $k$-edge graphs that satisfy $\Phi$. In particular, each
coefficient in this linear combination is $0$ or $1$. We have seen that the values of
$\Phi$ on the fixed-points of certain group actions on (fractures of) Cayley graphs can be
used to obtain explicit criteria for ($\#\W{1}$-)hardness of $\#\edgesubsprob(\Phi)$. In
the current section, we show that the aforementioned method applies to the significantly
more general problem of computing weighted linear combinations of $k$-edge subgraph
counts. More precisely, we consider a natural parameterized variant of the Tutte
polynomial and obtain an exhaustive classification for the complexity of evaluating it at
any rational coordinates.

Recall that the (classical) Tutte polynomial is defined as follows:
\[T_G(x,y) := \sum_{A \subseteq E(G)} (x-1)^{k(A)-k(E(G))} \cdot (y-1)^{k(A)+\#A-\#V(G)} \,, \]
where $k(S)$ is the number of connected components of the graph $(V(G),S)$.

\noindent In this work,
we consider the specialization of the Tutte polynomial to edge-subsets of size $k$,
which we call the \emph{parameterized Tutte polynomial}:
\[T^k_G(x,y) := \sum_{A \in \binom{E(G)}{k}} (x-1)^{k(A)-k(E(G))} \cdot (y-1)^{k(A)+k-\#V(G)} \,. \]
We emphasize that the parameterized Tutte polynomial is related to a generalisation of the bases generating function for matroids investigated by Anari et al.\ in their work on approximate counting (and sampling) via log-concave generating polynomials~\cite[Section~1.2]{AnariLGV19}.

Similarly to the classical counterpart due to Jaeger et al.\ \cite{JaegerVW90}, our goal
is to understand the parameterized complexity of evaluating $T^k_G(x,y)$ for any fixed
pair of coordinates $(x,y)$, when parameterized by $k$.
Note that at points $(x,y)$ with $x \neq 1, y \neq 1$ we can write the polynomial as
\[T^k_G(x,y) = (x-1)^{-k(E(G))} (y-1)^{k-\#V(G)} \sum_{A \in \binom{E(G)}{k}} ((x-1) \cdot (y-1))^{k(A)} \,. \]
So, up to the global factor $(x-1)^{-k(E(G))} (y-1)^{k-\#V(G)}$ (which can be computed in
linear time in the input size) in this region the polynomial is really just a polynomial
in the single variable $z=(x-1)(y-1)$. Still, we keep the variables $x,y$ separate in
the treatment below. On the one hand, this facilitates comparisons to the classical Tutte
polynomial. On the other hand, we see some interesting behaviour of $T^k_G(x,y)$ at
points with $x=1$ or $y=1$.
Indeed, let us start by investigating the expressibility of the parameterized Tutte polynomial in some individual points.

\subsection{Interpretation in Individual Points}\label{sec:indpoints}
Recall that, given a graph $G$ and a subset $A \subseteq E(G)$ of its edges, we write
$G(A) = (V(G), A)$ for the graph induced by $A$. We emphasize the difference from the
construction $G[A]$ we saw before: the graph~$G[A]$ is obtained from $G(A)$ by removing
all isolated vertices.

 The most immediate information encoded in the parameterized Tutte polynomial is the number of $k$-forests in a graph:

\begin{obs}
    The number of forests with $k$ edges in a graph $G$ is given by $T^k_G(2,1)$.
    \lipicsEnd
\end{obs}
In particular, evaluating $T^k_G(2,1)$ is equivalent to evaluation $\#\indsubs{\Phi,k}{G}$ for the (minor-closed) property of being acyclic.

For further individual points, it is convenient to consider the following modification.

\begin{definition}
    Define the modified Tutte polynomial of a graph $G$ as
    \[\widetilde T_G(x,y) := \sum_{A \subseteq E(G)} (x-1)^{k(A)} \cdot (y-1)^{k(A)+\#A} \,, \]
    so that $\widetilde T_G(x,y) = (x-1)^{k(E(G))} (y-1)^{\#V(G)} T_G(x,y)$. Similarly we define the parameterized version as
    \[\widetilde T^k_G(x,y) := \sum_{A \in \binom{E(G)}{k}} (x-1)^{k(A)} \cdot (y-1)^{k(A)+\#A} \,. \]
\end{definition}

As for its classical counter-part, we observe a deletion-contraction recurrence, which
enables us establish the properties at individual points.
Setting $\widetilde T^{-1}_G(x,y)=0$ we obtain:
\begin{lemma} \label{lem:Ttildedeletioncontraction}
    Given a graph $G$ and an edge $e \in E(G)$ we have
    \[
        \widetilde T^k_G(x,y) = \widetilde T^k_{G \setminus e}(x,y) + (y-1) \widetilde T^{k-1}_{G/e}(x,y) \, ,
    \]
    for any $k \geq 0$ and similarly
    \[
        \widetilde T_G(x,y) = \widetilde T_{G \setminus e}(x,y) + (y-1) \widetilde T_{G/e}(x,y) \, .
    \]
\end{lemma}
\begin{proof}
    In the definition of $\widetilde T^k_G$ we split the sum over $A \in \binom{E(G)}{k}$
    as
    \begin{equation} \label{eqn:Ttildesplit} \widetilde T^k_G(x,y) = \sum_{\substack{A \in \binom{E(G)}{k}\\e \not \in A}} (x-1)^{k(A)} \cdot (y-1)^{k(A)+\#A} + \sum_{\substack{A \in \binom{E(G)}{k}\\e \in A}} (x-1)^{k(A)} \cdot (y-1)^{k(A)+\#A}\,. \end{equation}
    The subsets $A \in \binom{E(G)}{k}$ with $e \not \in A$ are naturally identified with
    the subsets $A \in \binom{E(G \setminus e)}{k}$ and we have $G(A) = (G \setminus
    e)(A)$. Thus the first sum in \eqref{eqn:Ttildesplit} is equal to $\widetilde T^k_{G
    \setminus e}(x,y)$. On the other hand, the subsets $A \in \binom{E(G)}{k}$ with $e \in
    A$ are naturally identified with the subsets $A' \in \binom{E(G/e)}{k-1}$ by $A
    \mapsto A'=A \setminus \{e\}$ and we have $k(A) = k(A')$ (in their respective graphs
    $G$ and $G/e$). Thus the second summand in \eqref{eqn:Ttildesplit} equals $(y-1)
    \widetilde T^{k-1}_{G/e}(x,y)$, with the factor $(y-1)$ coming from the fact that $\#A
    = \#A'+1$ in the above correspondence. The deletion-contraction formula for the
    (unparameterized) modified Tutte polynomial is obtained by summing over all $k$.
\end{proof}

Using the previous recurrence, the following transformation encapsulates the relation
between the parameterized and the classical Tutte polynomial.
\begin{proposition}\label{prop:paramTutteinterpretation}
    Given a graph $G$ and $k \geq 0$ we have
    \begin{equation} \label{eqn:paramTutteinterpretation}
        \sum_{\ell = 0}^k \binom{\#E(G)-\ell}{k-\ell} \cdot \widetilde T_G^\ell(x,y) = \sum_{A \in \binom{E(G)}{k}} \widetilde T_{G(A)}(x,y)\,.
    \end{equation}
\end{proposition}
\begin{proof}
    We prove the statement by induction on the number of edges. For $E(G) = \emptyset$ the
    two sides are zero for $k \neq 0$ and equal to $\widetilde T_G^0(x,y) = \widetilde
    T_G(x,y)$ for $k=0$.

    We show the induction step using the deletion-contraction relations above. Let $G$
    denote a graph with at least one edge $e$. Then we have
    \begin{align}
        ~& \sum_{\ell = 0}^k \binom{\#E(G)-\ell}{k-\ell} \cdot \widetilde T_G^\ell(x,y) \\
        =& \sum_{\ell = 0}^k \binom{\#E(G)-\ell}{k-\ell} \cdot \widetilde T_{G\setminus e}^\ell(x,y) + \binom{\#E(G)-\ell}{k-\ell} \cdot (y-1) \widetilde T_{G/ e}^{\ell-1}(x,y)\,. \label{eqn:Ttildrel1}
    \end{align}
    Furthermore, since $\#E(G) = \#E(G \setminus e) +1$, we can use the usual recursion of
    binomial coefficients to see
    \begin{align*}
        &\sum_{\ell = 0}^k \binom{\#E(G)-\ell}{k-\ell} \cdot \widetilde T_{G\setminus e}^\ell(x,y)\\
        =&\sum_{\ell = 0}^k \binom{\#E(G \setminus e)-\ell}{k-\ell} \cdot \widetilde T_{G\setminus e}^\ell(x,y) + \binom{\#E(G \setminus e)-\ell}{(k-1)-\ell} \cdot \widetilde T_{G\setminus e}^\ell(x,y)\\
        =&\sum_{A \in \binom{E(G \setminus e)}{k}} \widetilde T_{(G\setminus e)(A)}(x,y) + \sum_{A' \in \binom{E(G \setminus e)}{k-1}} \widetilde T_{(G\setminus e)(A')}(x,y)\,,
    \end{align*}
    where we have used the induction step. For the second summand in \eqref{eqn:Ttildrel1} we make the index shift $\ell' = \ell-1$ and obtain
    \begin{align*}
        &\sum_{\ell = 0}^k \binom{\#E(G)-\ell}{k-\ell} \cdot (y-1) \widetilde T_{G/ e}^{\ell-1}(x,y)\\
        =&\sum_{\ell' = 0}^{k-1} \binom{\#E(G /e)-\ell'}{(k-1)-\ell'} \cdot (y-1) \widetilde T_{G/ e}^{\ell'}(x,y)\\
        =&\sum_{A' \in \binom{E(G / e)}{k-1}} (y-1) \widetilde T_{(G/e)(A')}(x,y)\,.
    \end{align*}
    Combining the last two equations we can conclude using suitable identifications, for
    instance identifying the $A \in \binom{E(G)}{k}$ with $e \in A$ with $A' \in \binom{E(G
    \setminus e)}{k-1}$ via $A \mapsto A'= A \setminus \{e\}$ and using
    \[(G\setminus e)(A') = G(A) \setminus e \text{ and } (G/e)(A') = G(A)/e\,. \]
    Then we see that  \eqref{eqn:Ttildrel1} equals
    \begin{align*}
        &\sum_{A \in \binom{E(G \setminus e)}{k}} \widetilde T_{(G\setminus e)(A)}(x,y)  + \sum_{A' \in \binom{E(G \setminus e)}{k-1}} \widetilde T_{(G\setminus e)(A')}(x,y) +  (y-1)\widetilde T_{(G/e)(A')}(x,y)\\
        =&\sum_{\substack{A \in \binom{E(G)}{k}\\e \not \in A}} \widetilde T_{G(A)}(x,y) + \sum_{\substack{A \in \binom{E(G)}{k}\\e \in A}} \widetilde T_{G(A)}(x,y) \\
        =&\sum_{A \in \binom{E(G)}{k}} \widetilde T_{G(A)}(x,y)\,.
    \end{align*}
\end{proof}

Using \cref{prop:paramTutteinterpretation} we can now present combinatorial interpretations of the specialisation of $T^k_G(x,y)$ to some individual points.

\paragraph*{Chromatic Polynomial}
For $x=1-c, y=0$ the modified Tutte polynomial $\widetilde T_G(x,y)$ specializes to the
chromatic polynomial $\chi_G(c)$, so we see that the $\widetilde T_G^\ell(1-c,0)$ (for $0
\leq \ell \leq k$) contain the information of the number of pairs $(A,\sigma)$ with $A
\subseteq E(G)$ with $\#A=k$ and $\sigma$ a $c$-colouring on $G(A)$.
\[\sum_{\ell = 0}^k \binom{\#E(G)-\ell}{k-\ell} \cdot \widetilde T_G^\ell(1-c,0) = \# \left\{(A,\sigma) : A \in \binom{E(G)}{k}, \sigma \text{ $c$-coloring on }G(A) \right\}\,.\]

\paragraph*{Acyclic Orientations}
For $x=2, y=0$ the Tutte polynomial $T_G(x,y)$ specializes to the number of acyclic
orientations of $G$. We have $\widetilde T_G^k(2,0)=(-1)^{\#V(G)} T_G^k(2,0)$. Thus the
$\widetilde T_G^\ell(2,0)$ (for $0 \leq \ell \leq k$) contain the information of the
number of pairs $(A,\vec \eta)$ where  $A \subseteq E(G)$ with $\#A=k$ and $\vec \eta$ is
an acyclic orientation on $G(A)$. Indeed, multiplying \eqref{eqn:paramTutteinterpretation}
with $(-1)^{\#V(G)}$ we obtain
\[\sum_{\ell = 0}^k \binom{\#E(G)-\ell}{k-\ell} \cdot T_G^\ell(2,0) = \# \left\{(A,\vec \eta) : A \in \binom{E(G)}{k}, \vec \eta \text{ acyclic orientation on }G(A) \right\}\,.\]

\paragraph*{$k$-Edge Sets Inducing an Even Number of Components}
\begin{proposition}
    Given a graph $G$ and a positive integer $k$, we have
    \begin{equation*}
        \frac{1}{2}\left(\binom{\#E(G)}{k} + (-1)^{k(E(G))} T_G^k(0,2) \right) = \# \left\{A \subseteq E(G) : \#A = k \wedge k(A) = 0 \mod 2 \right\}\,.
    \end{equation*}
\end{proposition}
\begin{proof}
    Let $E=E(G)$. We have
    \begin{align*}
        \frac{1}{2}\left(\binom{\#E}{k} + (-1)^{k(E)} T_G^k(0,2) \right) &= \frac{1}{2}\left(\binom{\#E}{k} + (-1)^{k(E)} \sum_{A \in \binom{E}{k}} (-1)^{k(A)+k(E)}\right) \\
        &= \frac{1}{2}\left( \sum_{A \in \binom{E}{k}} 1 + (-1)^{k(A)}\right)\,.
    \end{align*}
    But observe that the summand above is $0$ for $k(A)$ odd and $2$ for $k(A)$ even. Thus after summing and dividing by~$2$ we count the subsets $A$ with the graph $G(A)$ having an even number of components.
\end{proof}

\paragraph*{$k$-Edge Sets of Even Betti Number}
The (first) \emph{Betti number}\footnote{The first Betti number is also called the circuit rank, cyclomatic number, cycle rank, or nullity.} of a graph is defined as $b_1(G)=k(E(G)) + \#E(G) - \#V(G)$ (cf.\ \cite[Chapt.\ 4]{bettinumber}).

\begin{proposition}
    Given a graph $G$ and a positive integer $k$, we have
    \begin{equation*}
        \frac{1}{2}\left(\binom{\#E(G)}{k} + T_G^k(2,0) \right) = \# \left\{A \subseteq E(G) : \#A = k \wedge b_1(G(A)) =0 \mod 2 \right\}\,.
    \end{equation*}
\end{proposition}
\begin{proof}
    We have
    \begin{align*}
        \frac{1}{2}\left(\binom{\#E(G)}{k} + T_G^k(2,0) \right) &= \frac{1}{2}\left(\binom{\#E(G)}{k} + \sum_{A \in \binom{E(G)}{k}} (-1)^{k(A)+\#A-\#V(G)}\right) \\
                                                                &= \frac{1}{2}\left( \sum_{A \in \binom{E(G)}{k}} 1 + (-1)^{b_1(G(A))}\right)\,,
    \end{align*}
    where we use $b_1(G(A)) = k(A)+\#A-\#V(G)$. But observe that the summand above is $0$ for $b_1(G(A))$ odd and $2$ for $b_1(G(A))$ even. Thus after summing and dividing by $2$ we count the subsets $A$ with $G(A)$ having even Betti number.
\end{proof}

\subsection{Classification for Rational Coordinates}
We now classify the complexity of computing $T^k_G(x,y)$ for each pair of rational
coordinates $x$ and $y$. Formally, for each such pair, we consider the parameterized
problem which expects as input $G$ and $k$ and outputs the value $T^k_G(x,y)$; the
parameterization is given by $k$.
Let us start with the following easy fact:
\begin{lemma}\label{lem:TutteEasyLineOne}
    For any $y\in \mathbb{Q}$, the problem of computing $T^k_G(1,y)$ is fixed-parameter tractable.
\end{lemma}
\begin{proof}
    Observe that $T^k_G(1,y)=0$ unless there is $A\subseteq E(G)$ of size $k$ such that
    $k(A)=k(E(G))$. In other words, $G$ has a spanning subgraph of $k$ edges.
    Consequently, $G$ can have at most $2k$ vertices, implying that $G$ has at most
    $\binom{2k}{2}\leq 4k^2$ many edges. Therefore an algorithm for computing $T^k_G(1,y)$
    is obtained as follows: Given $G$ and $k$, first check whether $|V(G)|>2k$, and output
    $0$ in that case. Otherwise, obtain $T^k_G(x,y)$ by naively computing the sum, which
    takes time
    \[O\left(\binom{4k^2}{k} \cdot |G|\right)\,,\]
    concluding the proof.
\end{proof}

Next, similarly to the classical counter-part~\cite{JaegerVW90}, we obtain a trivial algorithm for coordinates $x$ and $y$ that lie on the hyperbola $(x-1)(y-1)=1$:

\begin{lemma}\label{lem:easyHyperbola}
    Let $x$ and $y$ denote rational numbers such that $(x-1)(y-1)=1$. Then the problem of
    computing $T^k_G(x,y)$ is solvable in polynomial time (and thus fixed-parameter
    tractable as well).
\end{lemma}
\begin{proof}
    Observe that, given $(x-1)(y-1)=1$, and setting $V=V(G)$ and $E=E(G)$, we have
    \[T^k_G(x,y) = \sum_{A \in \binom{E}{k}} (x-1)^{k(A)-k(E)} \cdot (y-1)^{k(A)+k-\#V}= (x-1)^{-k(E)} (y-1)^{k+\#V}  \binom{\#E}{k} \,, \]
    which can be computed trivially.
\end{proof}

In what follows, we show that computing $T^k_G(x,y)$ is $\#\W{1}$-hard for \emph{all
remaining} rational coordinates~$x$ and~$y$.
First, it is convenient to rewrite the quantity $k(A)$ as follows: given an edge-subset
$A$ of a graph~$G$, recall that $G[A]$ is the graph obtained from $(V(G),A)$ by deleting
isolated vertices. Let us write $\mathsf{cc}(H)$ for the number of connected components of
a graph $H$.

\begin{fact}
    Let $G$ denote a graph and let $A$ denote a subset of edges of $G$. We have
    \[ k(A)= \mathsf{cc}(G[A]) + \#V(G)-\#V(G[A]) \,.\]
    \lipicsEnd
\end{fact}

Similarly as in case of $\#\edgesubsprob(\Phi)$, our goal is to reduce from a linear
combination of (colour-preserving) homomorphism counts. For this reason, we again
consider an easy modification by excluding the term $(x-1)^{\#V(G)-\mathsf{cc}(G)}$; more
precisely, consider
\[ \widehat{T}^k_G(x,y) := \sum_{A \subseteq \binom{E(G)}{k}} (x-1)^{\mathsf{cc}(G[A])-\#V(G[A])} \cdot (y-1)^{\mathsf{cc}(G[A])-\#V(G[A])+k}\,,\]
and observe that
    \[T^k_G(x,y) =(x-1)^{\#V(G)-\mathsf{cc}(G)} \cdot \widehat{T}^k_G(x,y) \,.\]
In particular $\widehat{T}^k_G(x,y)$ is trivially interreducible with $T^k_G(x,y)$ if
$x\neq 1$. Next we introduce an ($H$-)coloured version of the parameterized Tutte
polynomial; given an edge-subset $A$ of a $k$-edge-coloured graph, we write
$\mathsf{cful}(A)$ if $A$ contains each of the $k$ colours precisely once.
\begin{defn}[Colourful Parameterized Tutte Polynomial]
    Let $G$ denote a $k$-edge-coloured graph. We define
    \[\coltkg := \sum_{\substack{A \subseteq \binom{E(G)}{k}\\\mathsf{cful}(A)}} (x-1)^{\mathsf{cc}(G[A])-\#V(G[A])} \cdot (y-1)^{\mathsf{cc}(G[A])-\#V(G[A])+k}\]
    as the \emph{colourful Parameterized Tutte Polynomial}.
    \lipicsEnd
\end{defn}
The next lemma allows us to reduce the colourful version to the uncoloured version.
\begin{lemma}\label{lem:nocolsTutte}
    Let $G$ denote a $k$-edge-coloured graph and assume that the set of colours is $[k]$.
    For each pair $(x,y)$ we have
    \[ \coltkg(x,y) = \sum_{J\subseteq[k]} (-1)^{\#J}\cdot \tilde{T}^k_{G\setminus J}(x,y)\,,  \]
    where $G\setminus J$ is the graph obtained from $G$ by deleting all edges coloured with an element of $J$.
\end{lemma}
\begin{proof}
    Follows by the inclusion-exclusion principle (similarly as in \cref{lem:nocolours})
    and the fact that, given a $k$-edge-subset $A$ of $G$, deleting edges in
    $E(G)\setminus A$ does not change the quantity
    \[(x-1)^{\mathsf{cc}(G[A])-\#V(G[A])} \cdot (y-1)^{\mathsf{cc}(G[A])-\#V(G[A])+k}\,.\]
\end{proof}

Next, we express $\coltkg$ as a linear combination of colour-preserving homomorphisms
counts. More precisely, given an $H$-coloured graph $G$ such that $H$ has $k$ edges, we
implicitly assume the $k$-edge-colouring of $G$ induced by its $H$-colouring. Further,
given a fracture $\rho$ of a graph $H$, we set
$r(\sigma):=\mathsf{cc}(\fracture{H}{\sigma})-\#V(\fracture{H}{\sigma})$.
\begin{lemma}\label{lem:col_lincomb_tutte}
    Let $H$ denote a graph with $k$ edges. For every $H$-coloured graph $G$, we have
    \[\#\coltkg(x,y) = \sum_{\sigma\in\mathcal{L}(H) } (x-1)^{r(\sigma)} \cdot (y-1)^{r(\sigma)+k} \cdot \sum_{{\rho} \geq {\sigma}} \mu({\sigma},{\rho}) \cdot \#\necphoms{\fracture{H}{\rho}}{H}{G} \,,\]
    where the relation $\leq$ and the M\"obius function $\mu$ are over the lattice of fractures $\mathcal{L}(H)$.
\end{lemma}
\begin{proof}
    Every colourful $k$-edge-subset $A$ of $G$ induces a fracture of $H$, similarly as we
    have seen in \cref{sec:CountingWithColours}. In particular, if $A$ and $A'$ induce the
    same fracture $\sigma$, then $G[A]\cong G[A']\cong \fracture{H}{\sigma}$. Writing
    $[\sigma]$ for the equivalence class of the induced fracture $\sigma$, we obtain:
    \[\#\coltkg(x,y) = \sum_{\sigma\in\mathcal{L}(H) } (x-1)^{r(\sigma)} \cdot (y-1)^{r(\sigma)+k} \cdot \#[\sigma] \,.\]
    Next observe that $\#[\sigma]=\#\ncoledgesubs{\fracture{H}{\sigma}}{H}{G}$. Finally,
    we have already seen in (the proof of) \cref{lem:col_lincomb} that
    \[\#\ncoledgesubs{\fracture{H}{\sigma}}{H}{G} = \sum_{{\rho} \geq {\sigma}} \mu({\sigma},{\rho}) \cdot \#\necphoms{\fracture{H}{\rho}}{H}{G}\,, \]
    which concludes the proof.
\end{proof}

The following lemma establishes that the coefficient of the torus does not vanish apart
from a few exceptions; which eventually allows us to prove $\#\W{1}$-hardness.

\begin{lemma}
    Let $\ell>2$ denote a prime and let $(x,y)\in \mathbb{Q}^2$. There is a unique and
    computable function $a^\ell_{(x,y)}$ from fractures of $\torus_\ell$ to rational
    numbers such that
    \[ \coltlg = \sum_{\rho \in \mathcal{L}(\torus_\ell)} a^\ell_{(x,y)}(\rho) \cdot \#\necphoms{\fracture{\torus_\ell}{\rho}}{\torus_\ell}{\star}  \,. \]
    Moreover, if both the denominators of $x,y$ and the numerators\footnote{In both cases, we refer to denominators and numerators of the corresponding shortened fractions.} of $x-1$ and $(x-1)(y-1)-1$ are not divisible by $\ell$, then $a^\ell_{(x,y)}(\top) \neq 0$.
\end{lemma}
\begin{proof}
    The first claim follows immediately from the previous lemma.
    For the second claim, we rely on the following fact from commutative algebra.
    \begin{fact} \label{fact:localization}
        Let $q \in \mathbb{Z}$ denote a nonzero integer, then the \emph{localization}
        $\mathbb{Z}[1/q]$ of $\mathbb{Z}$ at $q$ is the set
        \begin{equation*}
            \mathbb{Z}[1/q] = \left\{ u \in \mathbb{Q} : \exists v \in \mathbb{Z}, m \in \mathbb{N} \text{ with }u = \frac{v}{q^m} \right\}
        \end{equation*}
        of rational numbers which can be brought to a denominator which is a power of $q$.
        The subset  $\mathbb{Z}[1/q]$ of $\mathbb{Q}$ is closed under addition and
        multiplication. Let furthermore $\ell$ denote a prime not dividing $q$, implying that
        $q$ has an inverse $q^{-1}$ mod $\ell$. Then there is a well-defined map
        \begin{equation*}
            \mathbb{Z}[1/q] \to \mathbb{Z}_\ell, \frac{v}{q^m} \mapsto v \cdot (q^{-1})^m,
        \end{equation*}
        and this map is compatible with addition and multiplication.
        \lipicsEnd
    \end{fact}
    Let us now collect the coefficients of $\#\necphoms{{\torus_\ell}}{\torus_\ell}{\star}$ in the sum appearing in \cref{lem:col_lincomb_tutte}. Completely similar to \cref{cor:collect_coeffs} we obtain
    \begin{equation} \label{eqn:tuttepolytopcoeff}
        a^\ell_{(x,y)}(\top) = \sum_{\sigma \in \mathcal{L}(\torus_\ell)} (x-1)^{r(\sigma)} (y-1)^{r(\sigma)+2\ell^2} \prod_{v \in V(\torus_\ell)} (-1)^{|\sigma_v|-1}(|\sigma_v|-1)! \in \mathbb{Q}.
    \end{equation}
    Note that in this expression we have for the exponents of $x-1$ and $y-1$ that
    $r(\sigma)\leq 0$ but $r(\sigma)+k\geq 0$. Let $q$ denote the least common multiple of the
    denominators of $x,y$ and the numerator of $x-1$, then we have that $(x-1)^{\pm 1}$
    and $(y-1)$ are elements in $\mathbb{Z}[1/q]$. By the assumption that $\ell$ does not
    divide $q$ together with \cref{fact:localization} above, we can see these expressions
    (and thus the entire sum \eqref{eqn:tuttepolytopcoeff}) as an element of
    $\mathbb{Z}_\ell$. Now recall the 15 fixed-points of the action of $\ztwol$ on the
    fractures of $\torus_\ell$ as given by \cref{obs:fixedPointsbasic}. Counting modulo
    $\ell$ allows us to rely on the same analysis as presented in the proof of \cref{lem:coef_special_case}, which yields that $a^\ell_{(x,y)}(\top)$ is, modulo $\ell$, equal to
    \begin{align*}
	-6R(M_{2\ell^2}) +4 R(M_{\ell^2} + \ell C_\ell)+ 8R(\ell^2P_2)  - R(2\ell C_\ell) -2R(\ell C_{2\ell}) -4R(\ell S_\ell) + R(\torus_\ell)\,,
    \end{align*}
    where $R(H):= (x-1)^{\mathsf{cc}(H)-\#V(H) } \cdot (y-1)^{2\ell^2 + \mathsf{cc}(H)-\#V(H)}$.\pagebreak

    \noindent Consequently, we have
    \begin{align*}
        a^\ell_{(x,y)}(\top) =& -6 (x-1)^{ 2\ell^2-4\ell^2} \cdot (y-1)^{2\ell^2 + 2\ell^2-4\ell^2}\\
        ~&+ 4 (x-1)^{ \ell^2 + \ell -3\ell^2} \cdot (y-1)^{2\ell^2 + \ell^2 + \ell -3\ell^2}\\
        ~&+ 8 (x-1)^{ \ell^2 -3\ell^2} \cdot (y-1)^{2\ell^2 + \ell^2 -3\ell^2}\\
        ~&- 1 (x-1)^{2\ell -2\ell^2} \cdot (y-1)^{2\ell^2 + 2\ell -2\ell^2}\\
        ~&- 2 (x-1)^{ \ell -2\ell^2} \cdot (y-1)^{2\ell^2 + \ell -2\ell^2}\\
        ~&- 4 (x-1)^{ \ell -2\ell^2} \cdot (y-1)^{2\ell^2 + \ell -2\ell^2}\\
        ~&+ 1 (x-1)^{ 1 -\ell^2} \cdot (y-1)^{2\ell^2 + 1 -\ell^2} \mod \ell
    \end{align*}
    The first simplification is obtained by observing that the first and the third term, and the second, fifth and sixth term, respectively, contain the same monomial. Consequently, we have that (modulo $\ell$):
    \begin{small}
        \begin{align*}
            ~&a^\ell_{(x,y)}(\top) =\\
            ~&2(x-1)^{-2\ell^2} -2 (x-1)^{-2\ell^2 +\ell } \cdot (y-1)^\ell - (x-1)^{-2\ell^2 +2\ell} \cdot (y-1)^{2\ell} + (x-1)^{1-\ell^2}\cdot (y-1)^{\ell^2 +1}
        \end{align*}
    \end{small}
    Using Fermat's little theorem, we obtain
    \begin{small}
        \begin{align*}
            ~& a^\ell_{(x,y)}(\top) \\
            =&2(x-1)^{-2} -2 (x-1)^{-2} \cdot (y-1) -  (y-1)^{2} + (y-1)^{2} \mod \ell\\
            =& 2(x-1)^{-2}\cdot (1-(x-1)\cdot(y-1)) \mod \ell
        \end{align*}
    \end{small}
    The assumption that the denominator of $x$ (which is the numerator of  $(x-1)^{-2}$) and the numerator of $(x-1)(y-1)-1$ are not divisible by $\ell$ implies that each factor in the product $2(x-1)^{-2}\cdot (1-(x-1)\cdot(y-1))$ gives a nonzero residue class mod $\ell$. Since $\ell$ is a prime, their product is still nonzero in $\mathbb{Z}_\ell$, and thus the original rational number $ a^\ell_{(x,y)}(\top)$ is likewise nonzero, concluding the proof.
\end{proof}

 We are thus able to rely on Complexity Monotonicity to establish hardness as promised.
\begin{lemma}
    Let $(x,y)$ denote a pair of rational numbers such that $(x-1)(y-1)\neq 1$ and $x\neq 1$.
    Then the problem of computing $T^k_G(x,y)$ is $\#\W{1}$-hard.
\end{lemma}
\begin{proof}
    Let $\mathcal{H}[x,y]$ denote the set of all $\torus_\ell$ such that $\ell$ is prime and
    both the denominators of $x$ and $y$ as well as the numerators of $x-1$ and
    $(x-1)(y-1)-1$ are not divisible by $\ell$. Since $x$ and $y$ are fixed, the latter is
    true for infinitely many primes $\ell$ and thus $\mathcal{H}[x,y]$ contains tori of
    unbounded size. In particular, it contains graphs with arbitrary large grid minors and
    has thus unbounded treewidth~\cite{RobertsonS86-ExGrid}, and hence, the problem
    $\#\homsprob(\mathcal{H}[x,y])$ is $\#\W{1}$-hard by the classification of Dalmau and
    Jonsson~\cite{DalmauJ04}.

    Completely analogously to the proof of \cref{lem:hardness_basis}, the problem
    $\#\homsprob(\mathcal{H}[x,y])$ reduces to computing $\coltkg(x,y)$ via Complexity
    Monotonicity (\cref{lem:newcomplexitymonotonicity}), since the coefficients of the tori do not vanish by the previous lemma.

    Next, reducing to the uncoloured version $\widehat{T}^k_G(x,y)$ can be done via
    \cref{lem:nocolsTutte}, and, finally, $\widehat{T}^k_G(x,y)$ is trivially
    interreducible with $T^k_G(x,y)$ whenever $x\neq 1$.
\end{proof}

At last, we are able to prove this section's dichotomy theorem.
\begin{theorem}
	Let $(x,y)$ denote a pair of rational numbers. The problem of computing $T^k_G(x,y)$ is
	fixed-parameter tractable if $x=1$ or $(x-1)(y-1)=1$, and $\#\W{1}$-hard otherwise.
\end{theorem}
\begin{proof}
    The fixed-parameter tractable cases follow from
    \cref{lem:TutteEasyLineOne,lem:easyHyperbola}, and the $\#\W{1}$-hard cases follow
    from the previous lemma.
\end{proof}
As an immediate consequence, the computation of each individual point considered in
\cref{sec:indpoints} is $\#\W{1}$-hard. Moreover, observe that the transformation
\[ 	\sum_{\ell = 0}^k \binom{\#E(G)-\ell}{k-\ell} \cdot \widetilde T_G^\ell(x,y) = \sum_{A \in \binom{E(G)}{k}} \widetilde T_{G(A)}(x,y)\,,\]
given by \cref{prop:paramTutteinterpretation}, is invertible in the sense that the numbers
\[\sum_{A \in \binom{E(G)}{\ell}} \widetilde T_{G(A)}(x,y)\]
for $\ell=0,\dots,k$ reveal $\widetilde T_G^k(x,y)$. Consequently, we obtain $\#\W{1}$-hardness of the information encoded in all considered individual points as well:
\begin{corollary}
    The following problems are $\#\W{1}$-hard when parameterized by $k$:
    \begin{itemize}
        \item Given $G$ and $k$, compute the number of $k$-edge subsets $A$ of $G$ such that $G(A)$ has an even number of components.
        \item Given $G$ and $k$, compute the number of pairs $(A,\sigma)$ such that $A$ is a $k$-edge subset of $G$ and $\sigma$ is a $c$-colouring of $G(A)$. Here $c\geq 2$ is a fixed integer.
        \item Given $G$ and $k$, compute the number of pairs $(A,\vec{\eta})$ such that $A$ is a $k$-edge subset of $G$ and $\vec{\eta}$ is an acyclic orientation of $G(A)$.
        \item Given $G$ and $k$, compute the number of $k$-edge subsets $A$ of $G$ such that $G(A)$ has even Betti number.
            \lipicsEnd
    \end{itemize}
\end{corollary}

\paragraph*{Comparison to the Classical Dichotomy and Real FPT Cases}
In this section, we ask which of the fixed-parameter tractable cases allow for a
polynomial-time algorithm. We can answer this question under the assumption $\ccP \neq
\#\ccP$ by considering the classical dichotomy of Jaeger, Vertigan and Welsh:\footnote{We
    state their classification only for rational numbers, but point out that the full
dichotomy includes all complex pairs.}
\begin{theorem}[\cite{JaegerVW90}]\label{thm:classicalTutte}
    Given a pair $(x,y)$ of rational numbers, computing $T_G(x,y)$ is solvable in
    polynomial time if $(x,y)\in \{(1,1),(-1,-1),(0,-1),(-1,0)\}$ or if $(x-1)(y-1)=1$. In
    all other cases the problem is $\#\ccP$-hard.
    \lipicsEnd
\end{theorem}
First, we observe that the parameterized dichotomy coincides with the classical
dichotomy, except for the three points $(-1,-1)$, $(0,-1)$, and $(-1,0)$, in which the
parameterized Tutte polynomial is ($\#\W{1}$-)hard to compute, but the non-parameterized
one is polynomial-time solvable. The latter indicates that taking the sum only over the
$k$-edge subsets can, in fact, make the problem harder.

However, the non-parameterized Tutte polynomial always reduces to the
parameterized Tutte polynomial \emph{via polynomial-time Turing reductions}, since we can
compute $T^0_G(x,y) + \dots + T^{\#E(G)}_G(x,y)$ which is equal to $T_G(x,y)$. Thus any
point $(x,y)$ in which the non-parameterized Tutte polynomial is $\#\ccP$ hard and in
which the parameterized Tutte polynomial is fixed-parameter tractable, constitutes a
``real'' FPT case. In particular, the latter shows that each point on the line $x=1$
yields a real FPT case, except for the point $(1,1)$, which needs special treatment.
More precisely, we have to determine whether computing $T^k_G(1,1)$ is not only
fixed-parameter tractable (see \cref{lem:TutteEasyLineOne}), but also polynomial-time solvable. To this end, observe that
\[T^{k}_G(1,1) = \begin{cases} T_G(1,1) & ~~~k = \#V(G)-k(E(G))\\
	0 & ~~~k \neq \#V(G)-k(E(G))
	\end{cases}  \,,\]
since, for $x=y=1$, we have \[(x-1)^{k(A)-k(E(G))}(y-1)^{k(A)+\#A-\#V(G)} = 0\,,\] unless $\#A
= \#V(G)-k(E(G))$. Thus, in point $(1,1)$, the parameterized Tutte polynomial
can be computed in polynomial time by relying on the algorithm given by
\cref{thm:classicalTutte} in case $k= \#V(G)-k(E(G))$, and outputting $0$, otherwise.

Finally, recall that by \cref{lem:easyHyperbola}, the case $(x-1)(y-1)=1$ allows for a polynomial-time algorithm. The complete picture is hence given by the following refined classification; consider \cref{fig:tutteeasy_intro} for a depiction of the tractable cases.
\tuttemain*

\subsection{Approximating the Parameterized Tutte Polynomial}
In the very last part of this paper, we identify rational points $(x,y)$ for which
$T^k_G(x,y)$ can be approximated efficiently. Recall from \cref{def:FPTRAS},
that an FPTRAS for a parameterized counting problem $(P,\kappa)$ is a (randomized)
algorithm $\mathbb{A}$ which, on input $I$, $\varepsilon$, and $\delta$, outputs a value
$z$ with probability
\[ \pr{(1-\varepsilon) P(I) \leq z \leq (1+\varepsilon) P(I)} \geq 1-\delta \,,\]
in time $f(\kappa(I))\cdot \mathsf{poly}(|I|,\varepsilon^{-1},\log(1/\delta))$ for some
computable function $f$. If $f$ is a polynomial as well, then $\mathbb{A}$ is called an
\emph{FPRAS} ``\textbf{f}ully \textbf{p}olynomial-time \textbf{r}andomized
\textbf{a}pproximation \textbf{s}cheme'' (cf.~\cite[Def.\ 11.2]{MitzenmacherU17}).

We have to be careful when speaking about approximating $T^k_G(x,y)$ since the latter can
have negative values. One way of dealing with negative valued functions is to require that
an FPTRAS/FPRAS outputs a pair $(z,s)$, such that $z$ is an approximation of the absolute
value and $s$ is the sign, that is, we require that with probability at least $(1-\delta)$,
we have
\[ (1-\varepsilon) |T^k_G(x,y)| \leq z \leq (1+\varepsilon) |T^k_G(x,y)| \text{ and } s=\mathsf{sign}(T^k_G(x,y)) \,.\]

We are now able to establish a region of rational points for which the parameterized Tutte
Polynomial admits an FPTRAS or even an FPRAS; the proof is a simple consequence of the
work of Anari et al.\ on approximate counting via log-concave
polynomials~\cite{AnariLGV19}.
\tutteapprox*
\begin{proof}
    The case $x=1$ is a trivial consequence of \cref{thm:tutte_main_param_intro}. If $x=1$,
    then exact counting is fixed-parameter tractable and thus there is an FPTRAS. We
    consider two cases for the remaining points.

    \noindent First, consider $y=1$. If $x=1$ as well, then we obtain by \cref{thm:tutte_main_param_intro}
    a polynomial-time algorithm for exact counting, and thus an FPRAS. Otherwise, we have
    \[ T^k_G(x,1) = \sum_{A \in \binom{E(G)}{k}} (x-1)^{k(A)-k(E(G))} \cdot 0^{k(A)+k-\#V(G)} = \sum_{\substack{A \in \binom{E(G)}{k}\\G(A) \text{ acyclic}}}(x-1)^{k(A)-k(E(G))} \,,\]
    since $k(A)+k-\#V(G)=0$ if and only if $G(A)$ is acyclic. Recall that $k(A)$ is the
    number of connected components of $G(A)(=(V(G),A))$ and observe that for acyclic sets
    of edges $A$ with $|A|=k$, we have $k(A)=\#V(G)-k$. Consequently,
    \[T^k_G(x,1)= (x-1)^{\#V(G)-k-k(E(G))}\cdot\#\{A\subseteq E(G)~|~\#A=k ~\wedge~ G(A) \text{ acyclic}\}\,.\]
    Since computing the number of acyclic edge-subsets of size $k$ admits an
    FPRAS~\cite{AnariLGV19,AnariD20}, we can conclude this case.

    In the remaining case, we have $x\neq 1$ and $y\neq 1$ (and thus $0< (x-1)(y-1) \leq
    1$). Let $q=(x-1)(y-1)$ and let $G$ denote a graph with edges $E=[m]$, that is, $G$ has
    $m$ edges labelled with $1,\dots,m$. Consider the polynomial
    \[f_{G,k,q}(x_1,\dots,x_m):=\sum_{A \in \binom{E}{k}} q^{-\mathsf{rk}(A)} \prod_{i\in A}x_i\,,\]
    where $\mathsf{rk}(A):=\#V(G) - k(A)$ is the rank of $A$ with respect to the graphic
    matroid of $G$. Anari et al.\ have established the existence of an FPRAS for
    evaluating $f_{G,k,q}(1_m)$ whenever $0<q\leq 1$~\cite[Section~1.2]{AnariLGV19}.
    Now observe that
    \[T^k_G(x,y) =  (x-1)^{\#V(G)-k(E(G))}(y-1)^k \cdot f_{G,k,q}(1_m)\,.\]
    Since neither $x=1$ nor $y=1$, we conclude that an FPRAS for $f_{G,k,q}(1_m)$ yields
    the desired FPRAS for $T^k_G(x,y)$.
\end{proof}

\bibliography{conference}
\end{document}